\documentclass[10pt,reqno]{amsart}
\usepackage{amsthm,amsfonts,amsmath,latexsym,amssymb,bbold,dsfont,mathtools}
\usepackage{amssymb,latexsym,amsfonts,amsmath,amsthm,mathabx,stmaryrd}
\usepackage[unicode]{hyperref} 
\usepackage{verbatim}
\usepackage[integrals]{wasysym}

\usepackage{graphics}
\usepackage{graphicx,xcolor}
\usepackage{epsfig}

\usepackage[rightcaption]{sidecap}

\usepackage{soul}    

\usepackage{mathtools}
\usepackage{mfirstuc}
\usepackage[normalem]{ulem}
\usepackage{tikz}

\usepackage[utf8]{inputenc}
\usepackage{cite}
\def\a{\alpha}

\def\g{\gamma}
\def\d{\delta}
\def\dd{\mathrm{d}}

\def\s{\sigma}

\def\D{\Delta}
\def\G{\Gamma}
\def\O{\Omega}
\def\l{\lambda}
\def\L{\Lambda}
\def\p{\partial}

\def\C{\mathbb C}

\def\Z{\mathbb Z}
\def\Q{\mathbb Q}
\def\R{{\mathbb R}}

\def\N{\mathbb N}
\def\tr{{\rm tr}\,}

\def \Lp {\mathsf L}  

\usepackage{xparse,aliascnt,hyperref,bookmark}

\NewDocumentCommand{\xnewtheorem}{m o m}
 {%
  \IfNoValueTF{#2}
   {\newtheorem{#1}{#3}}
   {%
    \newaliascnt{#1}{#2}%
    \newtheorem{#1}[#1]{#3}%
    \aliascntresetthe{#1}%
    \expandafter\newcommand\csname #1autorefname\endcsname{\makefirstuc  {\lowercase {#3}}}%
   }%
 }

\hypersetup{
    colorlinks,
    linkcolor={red!50!black},
    citecolor={blue!50!black},
    urlcolor={blue!80!black}
}

\theoremstyle{plain}
\newtheorem{thmc}{ERROR}[section]
\iftrue 
\xnewtheorem{thm}[thmc]{THEOREM}
\xnewtheorem{lemma}[thmc]{LEMMA}
\xnewtheorem{cl}[thmc]{COROLLARY}
\xnewtheorem{prop}[thmc]{PROPOSITION}
\xnewtheorem{assumptions}[thmc]{ASSUMPTIONS}
\xnewtheorem{conj}[thmc]{CONJECTURE}
\xnewtheorem{fact}[thmc]{FACT}
\theoremstyle{definition}
\xnewtheorem{defin}[thmc]{DEFINITION}

\theoremstyle{remark}
\xnewtheorem{remark}[thmc]{REMARK}
\xnewtheorem{remarks}[thmc]{REMARKS}

\else

\theoremstyle{plain}
\xnewtheorem{thm}[thmc]{Theorem}
\xnewtheorem{lemma}[thmc]{Lemma}
\xnewtheorem{cl}[thmc]{Corollary}
\xnewtheorem{prop}[thmc]{Proposition}
\xnewtheorem{assumption}[thmc]{Assumption}
\xnewtheorem{conj}[thmc]{Conjecture}
\xnewtheorem{fact}[thmc]{Fact}

\theoremstyle{definition}
\xnewtheorem{defin}[thmc]{Definition}

\theoremstyle{remark}
\newtheorem{remark}[thmc]{Remark}
\newtheorem{remarks}[thmc]{Remarks}
\fi

\newcommand{\be}{\begin{equation}}
\newcommand{\ee}{\end{equation}}
\newcommand{\bea}{\begin{eqnarray}}
\newcommand{\eea}{\end{eqnarray}}
\newcommand{\beax}{\begin{eqnarray*}}
\newcommand{\eeax}{\end{eqnarray*}}

\keywords{Entanglement entropy, area law, Landau Hamiltonian on half-plane} 
\subjclass[2010]{Primary 47G30, 35S05; Secondary 47B10, 47B35}

\numberwithin{equation}{section}

\begin{document}

\title[Entanglement entropy of the Landau Hamiltonian on the half-plane]{Entanglement entropy of ground states of the Landau Hamiltonian on the half-plane}

\date{\today}

\author[P.~Pfeiffer, W.~Spitzer]{Paul Pfeiffer and Wolfgang Spitzer}
\address{Ludwig-Maximilians-Universität München, 
	Fakultät für Mathematik, Informatik und Statistik
	Theresienstraße 39, 
	80333 München}
\address{Fakult\"at f\"ur Mathematik und Informatik, FernUniversit\"at in Hagen, Universit\"atsstra\ss e 1, 58097 Hagen, Germany} \email{pfeiffer@math.lmu.de} \email{wolfgang.spitzer@fernuni-hagen.de}

\begin{abstract} We study the entanglement entropy of ground states of a Hamiltonian defined on a domain with a boundary. Surprisingly, boundary conditions can change the spectrum and the nature of the spectrum drastically but not the leading behaviour of the entanglement entropy. As is well-known, the Landau Hamiltonian on the full plane has pure point spectrum (the infinitely degenerate Landau levels) and ground states display a so-called strict area law. On the other hand, the Landau Hamiltonian on the half-plane has purely absolutely continuous spectrum and yet we prove a strict area law for its ground states. We raise the question of what extra or finer conditions on the absolutely continuous spectrum are necessary to guarantee a logarithmically enhanced area-law as we have for the Laplace operator.
\end{abstract}

\maketitle
\centerline{\emph{Dedicated to Hajo Leschke on the occasion of his $81^\mathrm{st}$ birthday}}
\bigskip

\tableofcontents

\section{Motivation}

There have been a lot of efforts devoted to the study of entanglement entropy (EE) in general and to the scaling of this entropy for ground states of many-particle fermionic systems. Even the simplest, non-interacting systems pose considerable mathematical challenges and lead to interesting questions. These many-particle systems are completely described by operators $\mathcal H$ on the corresponding one-particle Hilbert-space such as the Laplace operator, $-\D$, or the Landau operator on $\mathsf{L}^2(\O)$, for some fixed spatial domain $\O\subset\R^d$. In our case, $\O$ will be the half-plane $\R^+\times\R$, which is (cone) invariant under scaling, that is, $L\cdot\O \subset \O$ for $L>0$. A ground state (of the many-particle) system on the full configuration space $\O$ is described by the spectral projection of $\mathcal H$, $P_\mu=1(\mathcal H\le\mu)$ for some $\mu\in\R$. Such a state is then spatially localized to $L\L$, see \eqref{def localized}. Then we ask about the behaviour of the EE of this localized state as $L$ becomes large. Since the EE of this state localized to $L\L$ is identical to the EE of the same state localized to the complement $\O\setminus (L\L)$,  it is plausible that the EE scales to leading order in $L$ with the surface area of the boundary, that is, $|\p(L\L)| = L^{d-1}|\p\L|$. If this is the case we call this a strict area-law (scaling) of EE. There are other notable cases where the EE scales differently to leading order, namely as $L^{d-1}\ln(L)|\p\L|$ and this we call a logarithmically enhanced area-law (scaling). Yet different (faster and slower) scaling behaviours are known for ground states of quantum spin chain models.

The simplest cases of the Laplace operator or Landau operator on the full Euclidean space $\R^d$ have been studied, see \cite{GK,LSS14,CE20,LSS20,PS22}. In the case of the Laplace operator one has a purely absolutely continuous spectrum with a slow decay of the integral kernel of the spectral projection $P_\mu$. This leads to a logarithmically enhanced area-law. It is relatively simple to deduce a strict leading area-law bound if the integral kernel of $P_\mu$ decays exponentially in all directions. To deduce the precise scaling in all cases is still another matter. The two-dimensional Landau operator has pure point spectrum with an exponentially decaying integral kernel $P_\mu$ and thus we see a strict area-law. 

Building on these results, stability results have been proved for perturbations by magnetic or electric potentials, see \cite{MS20,MS23,P21}. In \cite{PS23}, we studied the transition between a strict area law and a logarithmically enhanced area-law as we let the magnetic field tend to 0, or in \cite{LSS16},  where we studied the same transition as the temperature $T$ of an equilibrium state (instead of a ground state) tends to 0. 

In \cite{PfSob17}, the one-dimensional Laplace operator $\mathcal H = -\D+V$ with a periodic potential $V$ was studied. It is well-known that the spectrum of $\mathcal H$ is purely absolutely continuous and one expects a logarithmically enhanced area-law for the found states $P_\mu = 1(\mathcal H\le \mu)$. This turns out to be true (for all $\mu$ inside the spectrum) and, surprisingly, the leading coefficient of $\ln(L)$ does not depend on the potential $V$, that is, it is the same as for the one-dimensional Laplace operator. The case of a periodic potential in higher-dimension is unknown.

For (discrete) random systems with one-particle Hamiltonian $\mathcal H$, we mention the following results. In \cite{PasturSlavin14,EPS}, ground states $P_\mu = 1(\mathcal H\le \mu)$ with $\mu$ in the spectrum of $\mathcal H$ with exponentially localized eigenfunctions were studied and a strict area law was proved. What is unknown is the precise form of the coefficient of $L^{d-1}|\p\L|$, in particular, its dependence on physical parameters such as the disorder. In \cite{MPS}, the authors study the one-dimensional dimer model (an Anderson model with correlated weak disorder $v$) at the two exceptional energies, $\mu\in\{0,v\}$, at which there is delocalization. Surprisingly, they prove that the EE of ground states $P_\mu$ scales at least like $\ln(L)$ unlike for $\mu$ in the rest of the spectrum.

In this paper, we take a different configuration space from the start, namely the half-plane $\O=\R^+\times\R$ and the Landau operator on this half-plane with Dirichlet boundary conditions. Other boundary conditions such as Neumann boundary conditions, which are very popular in superconductivity, could be considered as well. Then it is interesting to ask how the boundary condition effects the (leading) scaling of the EE of ground states. While we only partially answer this question, see \autoref{main thm} and the remarks following, our initial interest in this model comes from a different direction. In contrast to the Landau model on the full plane, the Landau model on the half-plane has purely absolutely continuous spectrum. Does this now lead to a logarithmically enhanced area-law or to a strict area law? Our main result says that we have a strict area-law bound in all cases and the precise strict area-law scaling if the domain $\L$ is separated from the boundary $\{0\}\times\R$ of the half-plane and it is the same as if we considered the problem on the full plane. In other words, if the domain $\L$ is separated from the boundary, then the leading order of EE is uneffected by the Dirichlet boundary conditions. 

So, in all these models there is a relation between spectral properties and the type of area-law scaling. While $\mu$ in the point spectrum of $\mathcal H$ with exponentially decaying integral kernel of $P_\mu = 1(\mathcal H\le \mu)$ always leads to a strict area law we have now an example of a Hamiltonian $\mathcal H$ with purely absolutely continuous spectrum and yet a strict area law for its ground states. We discuss the decay properties of $P_\mu$ in the present model in \autoref{Discussion}.

\medskip

\textbf{Notation:} In our estimates, we denote constants by the letter $C$ and its value may change from one line to the next. We use the convention that $0$ is a natural number so that $\N \coloneqq \{0,1,2\ldots,\}$. Also, the set of negative integers is $\Z^-\coloneqq \{\ldots,-2,-1\}$, the set of positive integers $\Z^+$ equals $\{1,2,\ldots\}$ and hence the set of all integers $\Z = \Z^- \cup \{0\}\cup \Z^+ = \Z^-\cup\N$. 

We denote by $|\cdot|$ the Euclidean distance in $\R^d$ or $\C^d$, $d\in\Z^+$. The Hilbert space, $\mathsf{L}^2(\O,\C^d)$, of square-integrable functions on the set $\O$ to $\C^d$ is furnished with the Lebesgue measure on $\O$. If $d=1$ we simply write $\mathsf{L}^2(\O,\C) =\mathsf{L}^2(\O)$. We mainly use $\O = \R, \R^+\coloneqq (0,\infty)$, $\R^2$, $\R^+\times\R$, or $\O = I\times J, I$ for some open intervals $I,J$, and denote the corresponding $\mathsf{L}^2$-norm by $\|\cdot\|_2$, that is, $\|f\|_2\coloneqq \left(\int_\O \dd x\, |f(x)|^2\right)^{1/2}$. In case the space $\O$ is not too obvious, we write explicitly $\|\cdot\|_{\mathsf{L}^2(\O)}$. The corresponding scalar product is denoted by the angular brackets $\langle\cdot,\cdot\rangle$. We use the Dirac notion for $|\phi\rangle\langle \phi|$ to denote the projection onto $\C\cdot\phi$, for $\|\phi\|_2=1$.

We use $\mathsf{C}^m(\Omega)$, $m\in\N$, to denote the set of $m$ times differentiable functions on the open set $\O$ and $\mathsf{C}^\infty(\Omega) \coloneqq \bigcap_{m=0}^\infty \mathsf{C}^m(\Omega)$. We add the subindex $\text{c}$ when we restrict to functions whose support is compact and contained in $\O$. A distribution on $\Omega$ is an element of the dual space of $\mathsf{C}^\infty_{\text{c}}(\Omega)$. Any $f \in \Lp^2(\Omega)$ is identified with the distribution $\phi \mapsto \int_{\Omega} \dd x\,\overline{f(x)} \phi(x)$. The distributional derivative $\partial_j f$ of a distribution $f$ is given by $\phi \mapsto f(-\partial_j \phi) = -\int_{\Omega} \dd x\,\overline{f(x)} \frac{\p\phi}{\p x_j}(x)$.

For $\O\subseteq \R^d$, $d\in\{1,2\}$, and $m\in\N$ we use various Sobolev spaces defined as follows:
\begin{align*}
	\mathsf{H}^m(\O)&\coloneqq \big\{u\in \mathsf{L}^2(\O) : \mbox{ distributional derivatives } \p^\a u\in \mathsf{L}^2(\O) \mbox{ for all (multi)indices } \a, |\a|\le m\big\}\,,
	\\
	\mathsf{H}^m_0(\O)&\coloneqq \big\{u\in \mathsf{H}^m(\O) : \exists \mbox{ sequence } u_k\in \mathsf{C}^\infty_{\text{c}}(\Omega), k\in\N \mbox{ s.t. } \|u-u_k\|_{\mathsf{H}^m(\O)}\to 0\mbox{ as } k\to\infty\big\}\,,
	\\
	\mathsf{H}^m_{\text{loc}}(\O)&\coloneqq \big\{u:\O\to\C: \forall \mbox{ open and bounded }\O' \mbox{ with } \overline{\O'}\subset \O,  \mbox{ restriction } u|_{\O'}\in\mathsf{H}^m(\O')\big\}\,.  
\end{align*}
Here, we do not distinguish between $u$ and its equivalence class, $[u]$, determined by $u$.  The norm $\|\cdot\|_{\mathsf{H}^m(\O)}$ and inner (scalar) product are defined through 
\begin{align}
\|u\|^2_{\mathsf{H}^m(\O)} \coloneqq \sum_{|\a|\le m} \|\p^\a u\|_{\mathsf{L}^2(\O)}^2\,,\quad \langle u,v\rangle_{\mathsf{H}^m(\O)} \coloneqq \sum_{|\a|\le m} \langle\p^\a u,\p^\a v\rangle_{\mathsf{L}^2(\O)}\,.
\end{align}
In fact, we only need $m\in\{1,2\}$. The Sobolev spaces $\mathsf{H}^m(\O)$ and $\mathsf{H}^m_0(\O)$ are Hilbert spaces with the norm $\|\cdot\|_{\mathsf{H}^m(\O)}$, unlike $\mathsf{H}^m_{\text{loc}}(\O)$. They are often denoted by $W^{m,2}(\O)$ (with the same decorations $0$ and $\text{loc}$) in the literature (such as in our main reference \cite{Adams}) and can also be constructed as the closure of the spaces $\mathsf{C}^\infty_{(\text{c})}(\Omega)$ with respect to $\|\cdot\|_{\mathsf{H}^m(\O)}$.

Finally, we mention the two most important functions in this paper, namely $\varphi_k$ and $\psi_k$, defined  on $\R$, respectively on $[0,\infty)$ with parameter $k\in\R$. They are the normalized ground-state eigenfunctions of the shifted harmonic oscillator $-\frac{\dd^2}{\dd x^2} + (x-k)^2$ on $\R$ or $\R^+$, respectively. They will be expressed in terms of the parabolic cylinder functions $D_\nu$. In particular, $\varphi_k(x) = \pi^{-1/4} \exp\left(-(x-k)^2/2\right) = \pi^{-1/4}D_0(\sqrt{2}(x-k))$.

\section{Introduction}

We consider the Hamiltonian, $H^+$, which describes the energy of a (spin-less) electron on the half-plane, $\R^+\times \R$, in a constant magnetic field of strength $B>0$ perpendicular to this half-plane. This Hamiltonian is given by (see e.g. \cite[Equation 1.1]{BP99})
\begin{align}\label{def: Hamiltonian}
	H^+\coloneqq -\frac{\p^2}{\p x_1^2} + \Big(-\mathrm{i}\frac{\p}{\p x_2} - B x_1\Big)^2\,, \quad (x_1,x_2)\in \R^+\times \R\,,\R^+\coloneqq (0,\infty)\,.
\end{align}
with Dirichlet conditions on the boundary, $\{0\}\times\R$. It is formally defined via its form \eqref{H(+) form def} with form domain \eqref{H+ form domain}. The spectrum, $\s(H^+)$, of $H^+$ equals the half-line $[B,\infty)$ and is purely absolutely continuous. This is well-known and will become clear again from our discussion below, see \autoref{spectrum of H plus}.

Then, for $\mu>B = \inf\s(H^+)$ we introduce the so-called Fermi projection $P_\mu^+\coloneqq 1(H^+\le \mu) = 1_{(-\infty,\mu]}(H^+)$, defined for example through the functional calculus by applying the indicator function $1_{(-\infty,\mu]}$ of the closed interval $(-\infty,\mu]$ to the self-adjoint operator $H^+$. This projection $P_\mu^+$ describes the ground state of non-interacting fermions with one-particle Hamiltonian $H^+$ below the Fermi energy $\mu$. For simplicity, we will assume that $\mu<3B$. As $B$ is fixed, we set $B=1$. 

For any bounded (Lebesgue-measurable) set $\L\subset \R^+\times \R$ we define the localized Fermi projection
\begin{equation} \label{def localized}
	P_\mu^+(\L) \coloneqq 1_\L P_\mu^+ 1_\L\,.
\end{equation}
Here, $1_\L$ is understood as the indicator function of the set $\L$ as well as the multiplication operator by $1_\L$. 

In this paper, we are interested in the asymptotic behaviour of the trace of $f(P_\mu^+(L\L))$ as the scaling parameter $L$ tends to infinity for certain (test) functions $f$ defined on the interval $\R$ such as the R\'enyi entropy function $h_\a$, $\a>0$, defined as
\begin{equation}
	h_\a(t) \coloneqq \begin{cases} \frac1{1-\a} \ln\big(t^\a+(1-t)^\a\big) &\mbox{ for } \a\not=1\\-t\ln(t)-(1-t)\ln(1-t)&\mbox{ for } \a=1\end{cases} \,,
\end{equation}
and $h_\a(t)=0$ for $t\in\R\setminus(0,1)$. The quantity 
\begin{equation} \label{def EE} 
	S^+_\a(\L)\coloneqq \tr h_\a\big(P_\mu^+(\L)\big)
\end{equation} 
is called the $\a$-R\'enyi entanglement entropy of the ground state (described by $P_\mu^+$) reduced to $\L$. From a physics point-of-view, the most important case is $\a=1$ but other values are of interest, too. We are able to treat all $\a>0$.

As we want to understand this entropy at least for large domains, we fix a domain $\L$ and introduce a scaling parameter $L>0$, scale the domain $\L$ by $L$ and define $L\L\coloneqq \{x=L\cdot y : y\in\L\}$, and investigate the asymptotic behaviour of the function $L\mapsto \tr f(P_\mu^+(L\L))$ as $L\to\infty$. 

When we speak of an area law then we mean a leading asymptotic scaling of $\tr f(P_\mu^+(L\L))$ of the order $L$ (in spatial dimension two) and with a coefficient independent of $L$, which contains $|\partial\L|$, the length of the boundary of $\L$ in the ambient space, here as a subset of $\R^+\times\R$. 

First of all, we assume that $\L$ is the union of finitely many bounded domains (open connected sets) in $\R^+\times\R$, such that their closures are pairwise disjoint. Open and closed are meant in the trace topology of the Euclidean topology on $\R^+\times\R$. The open sets in this topology on $\R^+\times\R$ are the sets $\L'\cap (\R^+\times\R)$ where $\L'$ is an open set in $\R^2$. A set $\L\subset \R^+\times\R$ is closed by definition if the complement $(\R^+\times\R)\setminus\L$ is open. The boundary of an open subset $\L$ is thus defined as $\p^+\L\coloneqq \bar{\L}\setminus\L$. The length, $|\partial^+\Lambda|$, is the one-dimensional Hausdorff measure of this curve in $\R^+\times\R$. 
For example, the semi-circle $\L_0 = \{(x_1,x_2):|(x_1,x_2)| = \sqrt{x_1^2+x_2^2}<1, x_1>0,x_2\in\R\}$ is an open set in $\R^+\times\R$ with boundary $\p^+\L_0 = \{(x_1,x_2):x_1>0,x_2\in\R, |(x_1,x_2)|=1\}$ and length $|\partial^+\L_0| = \pi$.  

By $\p\L$ we denote the boundary of $\L$ considered as a subset of $\R^2$. In the above example of the semi-circle $\L_0$ we have $\p\L_0 = \p^+\L_0\cup \{(0,x_2):|x_2|\le 1\}$. 

Furthermore, we assume that the boundary curve $\p^+\L$ of $\L$ is $\mathsf{C}^3$-smooth. Altogether, we call such a set $\L$ a $\mathsf{C}^3$-smooth region. 

Here is our main result.

\begin{thm}\label{main thm} 
	For any $\mu\in (1,3)$, any $\a>0$, and any bounded $\mathsf{C}^3$-smooth region $\Lambda\subset \R^+\times\R$, the $\alpha$-R\'enyi entanglement entropy, $S_\a^+(L\L)$, of the ground state characterized by $1(H^+\le \mu)$, localized to $L\Lambda$ as defined in \eqref{def EE} satisfies an area-law bound,
	\begin{equation} \label{EE bound}
		S^+_\a(L\L)\le C L|\partial^+\L|\,,
	\end{equation}
	with some constant $C$ independent of $L$ and $\L$ but depending on $\mu$ and $\a$. Furthermore, if the distance of $\L$ to the boundary $\{0\} \times \R$ of  $\R^+\times\R$ is strictly positive, we get that $S^+_\alpha(L\L) =  L|\partial\L| \mathsf{M}_0(h_\a)+ o(L)$ with the coefficient $\mathsf{M}_0(h_\a)$ explained in \eqref{M_0 coefficient}.
\end{thm}

\begin{remarks}
\begin{enumerate} 
	\item[(i)] 
		In our estimates on the entropy and such quantities, we make no assumption on the region $\L$ besides being bounded. We take the proven result \cite{CE20,LSS20} for the scaling behaviour of the entanglement entropy \eqref{def EE on full plane} on the full plane with the ground state projection $P_\mu$ and compare it to the situation on the half-plane with the ground state projection $P^+_\mu$. The weakest assumptions we are aware of are stated in \cite{PS23} which assume a $\mathsf{C}^2$-smooth domain or $\L$ being a polygon. This comparison gets worse as
		$\mu\downarrow1$ since then $P_\mu^+$ goes to 0 unlike $P_\mu$, which is independent of $\mu\in(1,3)$. In other words, by our method of proof the constant $C$ in \eqref{EE bound} can be bounded uniformly in $\mu$ on any compact subinterval of $(1,3)$.
	\item[(ii)] 
		 Notice that in the case $\operatorname{dist}(\L,\{0\}\times \R)>0$, we have $\p^+\L = \p\L$ and hence $|\p^+\L| = |\p\L|$. This is the area law proved in \cite{CE20,LSS20}. In general, $|\p\L| \le 2 |\p^+\L|$ and hence we can also write $CL|\p\L|$ on the right-hand side of \eqref{EE bound}.
	\item[(iii)] 
		With some extra work it might be possible to prove the asymptotic area-law behaviour, $S^+_\alpha(L\L) =  L|\partial^+\L| \mathsf{M}_0(h_\a)+ o(L)$ also for subsets $\L$ whose closure intersects the boundary $\R^+\times\R$. %
\end{enumerate}
\end{remarks}

Our approach is to compare $\tr f\big(P_\mu^+(L\L)\big)$ with a known trace (asymptotics), namely that of the localized Fermi projection, $P_\mu(L\L)\coloneqq 1_{L\L} P_\mu1_{L\L}$, of the Landau Hamiltonian, 
\begin{align}
H \coloneqq -\frac{\p^2}{\p x_1^2} + \Big(-\mathrm{i}\frac{\p}{\p x_2} - x_1\Big)^2\,, \quad (x_1,x_2)\in \R^2\,, 
\end{align}
defined in the same manner as $H^+$ but on the whole plane $\R^2$. For a formal definition via the quadratic form, see \eqref{H(+) form def} and \eqref{H form domain}. 
In contrast to $H^+$, the Landau Hamiltonian $H$ on $\R^2$ has pure point spectrum equal to $\{2\ell+1:\ell\in\N\}$. These eigenvalues are called Landau levels\footnote{Recall our convention that $0$ is a natural number.}. The entanglement entropy of the ground state displays a strict area-law behaviour. In \cite{CE20,LSS20}, the precise leading asymptotic behaviour or scaling as $L\to\infty$ was proved,
\begin{equation} \label{asympt for P}
	\tr f\big(P_\mu(L\L)\big) = L|\p\L| \mathsf{M}_0(f) + o(L)\,,
\end{equation}
where, for suitable functions, the functional $M_0(\cdot)$ is given by
\begin{equation} \label{M_0 coefficient}
	\mathsf{M}_0(f) \coloneqq \int_\R \frac{\dd\xi}{2\pi} \, f\big(\l(\xi)\big)\quad \mbox{with}\quad \l(\xi) \coloneqq \pi^{-1/2} \int_\xi^\infty\dd t\,\exp(-t^2)\,.
\end{equation}
Note that for $\L\subset \R^+\times\R$, the localized Fermi projection $P_\mu(\L)$ satisfies
\[ P_\mu(\L) = 1_\L P_\mu 1_\L = 1_\L 1_{\R^+\times\R}P_\mu1_{\R^+\times\R} 1_\L 
\] 
and hence $\tr f\big(1_{\L} 1_{\R^+\times\R}P_\mu1_{\R^+\times\R} 1_{\L}\big) = \tr f\big(P_\mu(\L)\big)$. As above, the quantity
\begin{align}\label{def EE on full plane}
S_\a(\L)\coloneqq \tr h_\a\big(P_\mu(\L)\big)
\end{align}
is called the $\a$-R\'enyi entanglement entropy of the ground state on the full plane (described by $P_\mu$) reduced to $\L$. The area law contains $|\p\L|$, the length of the boundary as a set of $\R^2$. Also the boundary $\p\L$ may only be a \emph{piecewise} smooth curve in $\R^2$.

We use the notation $H^{(+)}$ to mean either $H$ or $H^+$, when talking about statements that apply to both. We also use this notation for other objects associated to the operators, like $\R^{(+)}$.

The Hamiltonian $H^{(+)}$ (and hence the Fermi projection $P_\mu^{(+)}$) depends on the gauge but, due to cyclicity of the trace, the quantity $\tr f\big(P_\mu^{(+)}(\L)\big)$ is independent of this gauge. The index $0$ in $\mathsf{M}_0(f)$ indicates the $0^{\mathrm{th}}$ or lowest Landau level since $\mu\in(1,3)$ is between the lowest ($\ell=0$) and the next lowest ($\ell=1$) Landau level, which is all that counts for the coefficient $\mathsf{M}_0(f)$. In \cite{LSS20}, the precise leading asymptotic behaviour for any $\mu>1$ was proved which resulted in a more complicated coefficient $\mathsf{M}_{\le n}(f)$ with $n = \lfloor(\mu-1)/2\rfloor$ but is otherwise the same. Note that since $f(1)=0$, the volume term (of the order $L^2$) is absent in \eqref{asympt for P}. 

We close this section with a brief overview of the proof of our main result, \autoref{main thm}.

We will utilize a partial Fourier transform with respect to $x_2$, $\mathcal F_2$ (see \eqref{Fourier 2 def}) and relate $H^{(+)}$ to a direct integral over the following operators. For any $k\in\R$, let 
\begin{equation} 
	H^+(k) \coloneqq -\frac{\dd^2}{\dd x^2} + (x-k)^2
\end{equation}
be the (shifted by $k$) harmonic oscillator (operator) on $\mathsf{L}^2(\R^+)$ with Dirichlet boundary condition at $0$, formally defined via the quadratic form \eqref{H(+)k form def} with form domain \eqref{H+k form domain}. Let $\psi_k$ be the ground-state eigenfunction of $H^+(k)$, that is, the one with the lowest eigenvalue.

Furthermore, let $H(k)\coloneqq -\frac{\dd^2}{\dd x^2} + (x-k)^2$ be the operator on $\Lp^2(\R)$ (see \eqref{H(+)k form def} and \eqref{Hk form domain} for the formal definition via its form) and let $\varphi_k$ be its (normalized) ground-state eigenfunction. We will show that
\begin{equation}  \label{F2 operator trafo intro eq}
	\mathcal F_2 H^{(+)} \mathcal F_2^*= \int_\R^\oplus  \dd k\, H^{(+)}(k) \,.
\end{equation}
This will be done in \autoref{section Fourier reduction}, where we will also provide a precise definition of the direct integral, as well as all the discussed forms and operators. The notation $H^{(+)}$ ($H^{(+)}(k)$ or $P_\mu^{(+)}$, and later similar ones) stands either for $H$ or $H^+$ and, if there are several lines of calculations, this choice is kept consistently.

As a consequence of that, for any $a\in \N,b\in \Z$, we will show the Hilbert--Schmidt norm (see \eqref{def Schatten norm} with $p=2$) identity
\begin{align}\label{HS norm identity intro eq}
	\Big\| &1_{{(a,a+1)\times(b,b+1)}}\big(P_\mu^+ - P_\mu(\R^+\times\R)\big)\Big\|^2_{2} 
	\\
	&= \frac 1 {2\pi}  \int_{\R} \mathrm d s \int_a^{a+1} \mathrm d x_1 \int_{\R^+} \mathrm d y\, \big[1(s\ge k_F)\psi_s(x_1)\psi_s(y_1) -\varphi_s(x_1)\varphi_s(y_1) \big]^2\,, \nonumber
\end{align}
where $k_F\ge 0 $ is a constant depending on $\mu$. Finding sufficient upper bounds for this expression is a key part of this paper.

In \autoref{Section: Local and global estimates}, we will show that
\begin{align*}
	&\left \lVert 1_{{(a,a+1)\times(b,b+1)}}\big(P_\mu^+ - P_\mu(\R^+\times\R)\big)\right\rVert_{2} \le C\exp(-a/2) \, .
\end{align*}
Afterwards, the rough idea is to use the triangle inequality and the summability of the previous expression over $a$ to arrive at
\begin{align*}
	\left\lVert 1_{(0,L)\times (-L,L)} \big(P_\mu^+ - P_\mu(\R^+\times\R)\big)\right\rVert_{2} \le 2(L+1) \sum_{a=0}^LC \exp(-a/2)\le CL \, .
\end{align*}
This gets a bit more complicated then stated here, as we actually need to estimate the $p$-Schatten quasi norm for some $p<1$ instead of the Hilbert--Schmidt norm, see \autoref{Section: Local and global estimates}.


\section{Forms and operators, Fourier transform, and direct integrals} \label{section Fourier reduction}

The goal of this section is to formalize and prove \eqref{F2 operator trafo intro eq}. To this end, we also provide definitions of our differential operators via their associated forms. All this is mainly for the convenience of the reader and nothing new. As standard references to this topic we recommend \cite{Kato,RS2,RS4} for details. This section should also prepare for the calculations in the sections following.


Let us first formally define our operators $H^{(+)}$ by the sesquilinear forms
\begin{align} \label{H(+) form def}
	q^{(+)}(f,g) \coloneqq \int_{\R^{(+)} \times \R}\mathrm d x\,\Big(\overline{ \partial_1 f (x) } \partial_1 g(x) + \overline{ (\mathrm{i} \partial_2+ x_1)f(x)} (\mathrm{i} \partial_2+x_1)g(x)\Big) 
\end{align}
with form domains 
\begin{align}
	\mathcal Q (q)    	&\coloneqq \big\{ f \in (\mathsf{H}^1_{\text{loc}} \cap \Lp^2)(\R^2) : q(f,f)<\infty \big\} \, , \label{H form domain}\\
	 \mathcal Q (q^+)	& \coloneqq  \big\{ f \restriction_{\R^+\times \R} : f \in \mathcal Q(q), f \restriction_{\R^-\times \R}=0 \big\} \,.\label{H+ form domain}
\end{align}

To relate the definition of $\mathcal Q(q^+)$ to the Dirichlet boundary condition $f(0,x_2)=0$ for almost every $x_2 \in \R$, we have to define this expression. A priori, our $f$ would only be defined up to sets of measure $0$, like $\{0\}\times \R$. Let $f \in \mathcal Q(q)$. Clearly, $f \in \Lp^2(\R^2)$ and $\partial_1 f \in \Lp^2(\R^2)$. Thus, by Fubini, for almost every $x_2 \in \R$, we get $f(\cdot,x_2), \partial_1 f(\cdot,x_2) \in \Lp^2(\R)$, which states $f(\cdot,x_2) \in \mathsf{H}^1(\R)$. Due to the Sobolev imbedding theorem~\cite[Theorem 5.4, Case C]{Adams}, such $f(\cdot,x_2)$ have a unique continuous representative. This yields that $f(0,x_2)$ is well-defined for almost every $x_2 \in \R$. If now $f\restriction_{\R^- \times \R } =0$, we see that $f(0,x_2)=0$ for almost every $x_2 \in \R$. Thus, any $g \in \mathcal Q(q^+)$ can be uniquely extended to $\{0\}\times \R$ and vanishes at $g(0,x_2)$ for almost every $x_2 \in \R$. 	

We equip these sesquilinear form domains with the (graph) scalar products $q^{(+)}+\langle \cdot, \cdot \rangle_{\Lp^2(\R^{(+)}\times \R)}$, which turns them into Hilbert spaces. To show that we note that the sesquilinear forms $q^{(+)}$ satisfy the following properties:
\begin{enumerate}
\item[(i)] non-negativity, that is, $q^{(+)}(f,f)\ge0$, which is trivial by the very definition;
\item[(ii)] symmetry, that is, $q^{(+)}(f,g) = \overline{q^{(+)}(g,f)}$, which is also trivial by the very definition;
\item[(iii)] closedness, that is, the form domain $\mathcal Q(q^{(+)})$ is complete under the norm $\sqrt{q^{(+)}(f,f)+\|f\|^2}$, which happens to be the norm induced by the above scalar product. Closedness is non-trivial to prove.  
\item[(iv)] density, that is, $\mathcal Q (q^{(+)})$ is a dense subset of $\Lp^2(\R^{(+)}\times \R)$.
\end{enumerate}

The last two points are proved in \autoref{App C+ density}, see \autoref{Q are Hilbert spaces} and \autoref{Q is dense}.

The sesquilinear forms $q^{(+)}$ uniquely determine the self-adjoint operators, $H^{(+)}$, known as the Friedrichs extensions. To this end, one starts with defining the operators $H^{(+)}$ on some dense domain, for example on $\mathsf{C}^\infty_{\text c}(\R^{(+)}\times\R)$. Recall that the lower index $\text{c}$ in $\mathsf{C}^\infty_{\text{c}}$ stands for compactly supported functions. Clearly, $H^{(+)}$ are positive, symmetric operators and for $f,g\in \mathsf{C}^\infty_{\text c}(\R^{(+)}\times\R)$ we have
\begin{align}\label{form vs op}
q^{(+)}(f,g) = \langle f,H^{(+)}g\rangle\,.
\end{align}
Then, there are unique self-adjoint operators, also denoted by $H^{(+)}$, which are extensions of the operators originally defined on $\mathsf{C}^\infty_{\text c}(\R^{(+)}\times\R)$ and whose operator domains, $\mathcal D(H^{(+)})$, are contained in the form domains of $q^{(+)}$. See \cite[Theorem X.23]{RS2}. We are going to show now that 
\begin{align}
	\mathcal D (H) &= \big\{ g \in (\mathsf{H}^2_{\text{loc}} \cap \Lp^2)(\R^2) : Hg \in \Lp^2(\R^2) \big\} \, ,\\
	\mathcal D (H^+)& = \big\{ g \in (\mathsf{H}^2_{\text{loc}}  \cap \Lp^2)(\R^+ \times \R) : H^+g \in \Lp^2(\R^+\times\R), g \restriction_{\{0\}\times \R}=0 \big\} \,.
\end{align}
As we discussed after \eqref{H+ form domain}, the boundary condition $g\restriction_{\{0\}\times \R}=0$ is well-defined for all $f \in \mathcal Q(q^+) \supset \mathcal D (H^+)$. Any function $g$ in the claimed form domains clearly satisfies \eqref{form vs op} for every $f \in \mathcal Q(q^{(+)})$. We have to show that those are all functions that do so.

The conditions $H^{(+)}g \in \Lp^2(\R^{(+)}\times \R)$ are just part of the Friedrichs extensions. The non-trivial claim is that $g \in \mathsf{H}^2_{\text{loc}} (\R^{(+)}\times \R)$. Since $g$ is in the form domain, we know that $g \in\mathsf{H}^1_{\text{loc}}(\R^{(+)}\times \R)$. As the identity \eqref{form vs op} holds for all $f \in \mathsf{C}^\infty_{(\text{c})}(\R^{(+)}\times \R)$, we can conclude that the distributional derivative $\big(-\p_1^2 + (-\mathrm{i}\p_2 - B x_1 )^2 \big) g$ is square integrable. Since $g \in \mathsf{H}^1_{\text{loc}}(\R^{(+)}\times \R)$, every term which does not contain a second derivate is locally square integrable. Thus, $(-\p_1^2-\p_2^2)g \in \Lp^2_{\text{loc}}(\R^{(+)}\times \R)$. A well-known statement on elliptic regularity (see \cite[Theorem 1, Chp.~6.3.1]{Evans}) now implies $g \in \mathsf{H}^2_{\text{loc}}(\R^{(+)}\times \R)$.

Moreover, \eqref{form vs op} holds for all $f\in \mathcal Q\big(q^{(+)}\big), g\in \mathcal D\big(H^{(+)}\big)$.


We will now introduce a convenient dense subset, the linear span of $\mathcal C^{(+)}$. Let us define the one-dimensional Schwartz space
\begin{align}
	\mathcal S (\R) \coloneqq \big\{ f \in \mathsf{C}^\infty(\R) : \forall m , n \in \N \colon \sup_{x \in \R}  \lvert x \rvert^n \lvert f^{(m)} (x) \rvert < \infty \big\} \,.
\end{align}
We now consider 
\begin{align} \label{def C}
	\mathcal C ^{(+)} &\coloneqq 
	\big\{f:\R^{(+)}\times\R\to\C: f(x)= f_1(x_1) f_2(x_2), x=(x_1,x_2), f_1 \in \mathsf{C}^\infty_{\text{c}}\big(\R^{(+)}\big), f_2 \in \mathcal S(\R) \big\}\, .
\end{align}
The linear spans of $\mathcal C ^{(+)}$ are dense in the respective Hilbert spaces $\mathcal Q\big(q^{(+)}\big)$, as we show in \autoref{App C+ density}. 


We continue by introducing the partial Fourier transform $\mathcal F_2$, given by
\begin{equation} \label{Fourier 2 def}
	\mathcal F_2(f)(x_1,k) \coloneqq (2\pi)^{-1/2} \int_\R \dd x_2 \, f(x_1,x_2) \exp(-\mathrm{i}x_2 k)\,,\quad (x_1,k)\in\R^2\,.
\end{equation}
This is well-defined for a function $f\in\mathsf{L}^1(\R^2)$, and extended to $\mathsf{L}^2(\R^2)$ as a limit. It can trivially be restricted to $\mathsf{L}^2(\R^+\times \R)$ and is an isometry on both $\mathsf{L}^2(\R^2)$ and $\mathsf{L}^2(\R^+\times \R)$ by Plancherel's identity and Fubini. Its adjoint (and thus inverse) is given by
\[ 
	(\mathcal F_2^{*} f)(x_1,x_2) = (2\pi)^{-1/2} \int_\R \dd k\, f(x_1,k) \exp(\mathrm{i}x_2 k)\,,\quad (x_1,x_2)\in \R^{(+)}\times\R\,.
\]
For $f \in \mathcal C^{(+)}$, we note that $(\mathcal F_2 f )(x_1,k)=f_1(x_1) (\mathcal F f_2)(k)$, where 
\[
	(\mathcal F f_2)(k) \coloneqq (2\pi)^{-1/2} \int_\R \mathrm d t\, f_2(t)\exp(-\mathrm{i} tk)
\]
is the one dimensional Fourier transform. In particular, as $\mathcal F$ preserves the Schwartz space $\mathcal S(\R)$, we can conclude that $\mathcal F_2$ preserves the spaces $\mathcal C^{(+)}$, meaning $\mathcal F_2\,\mathcal C^{(+)} = \mathcal C^{(+)} = \mathcal F_2^* \,\mathcal C^{(+)}$.
 
 
For any $k \in \R$, the operators $H^{(+)}(k)$ are formally defined via the sesquilinear forms
\begin{align} \label{H(+)k form def}
	q^{(+)}_k(f,g)\coloneqq \int_{\R^{(+)}} \mathrm d x\, \Big(\overline {f'(t)} g'(x) + \overline{(x-k) f(x) } (x-k) g(x)  \Big)
\end{align}
with form domains
\begin{align}
	\mathcal Q (q_k) &\coloneqq \big\{ f \in \mathsf{H}^1 (\R) : q_k(f,f)<\infty \big\} \, ,\label{Hk form domain}\\
	\mathcal Q (q^+_k)	&\coloneqq  \big\{ f \in \mathsf{H}_0^1(\R^+) : q^+_k(f,f)<\infty\big\} \,.\label{H+k form domain}
\end{align}
As $q^{(+)}_k(f,f)\ge \lVert f' \rVert_{\Lp^2(\R^{(+)})} $, the form domains are just subspaces of the global Sobolev spaces, that is, $\mathcal Q (q_k)\subset \mathsf{H}^1(\R)\subset\Lp^2(\R)$ and $\mathcal Q (q^+_k)\subset \mathsf{H}^1_0(\R^+)\subset\Lp^2(\R^+)$, instead of its local versions, as we saw before. 

Let us now introduce the operator domains associated to these forms,
\begin{align}
	\mathcal D (H(k)) &=\big\{ g \in (\mathsf{H}^2_{\text{loc}}\cap \mathsf{H}^1)(\R), Hg \in \Lp^2(\R) \big\} \, ,\\
	\mathcal D(H^+(k))&=\big\{ g \in (\mathsf{H}^2_{\text{loc}}\cap \mathsf{H}^1_0)(\R^+), H^+g \in \Lp^2(\R^+)\big\} \, .
\end{align}
Trivially, for all $g \in \mathcal{D} \big(H^{(+)}(k)\big), f \in \mathcal Q \big(q^{(+)}_k\big)$, we have
\begin{align}
	q^{(+)}_k(f,g)= \langle f, H^{(+)}(k) g \rangle \, .
\end{align}
Similar to the arguments for $\mathcal D (H^{(+)})$, we start by observing that the last identity holds for all $f \in \mathsf{C}^\infty_{(\text{c})}(\R^{(+)})$, which tells us that the distribution $(-\p_1^2+(x_1-k)^2)g$ is square integrable and, as $(x_1-k)^2g$ is locally square integrable, this directly says that $\p_1^2 g = g''$ is locally square integrable. Since $g \in \mathsf{H}^1(\R^{(+)})$, this yields $g \in \mathsf{H}^2_{\text{loc}}(\R^{(+)})$.

We continue by constructing the direct integral operators $\int_\R^\oplus  \dd k\, H^{(+)}(k) $ via the sesquilinear forms
\begin{align}
	\hat q^{(+)} (f,g)& \coloneqq \int_\R \mathrm d k\, q^{(+)}_k\big(f(\cdot, k), g(\cdot,k)\big)\\
	& = \int_\R \mathrm d k\int_{\R^{(+)}} \mathrm d x_1 \Big(\overline{\partial_1 f (x_1,k)} \partial_1 g(x_1,k) + \overline{(x_1-k)f(x_1,k)}(x_1-k)g(x_1,k) \Big)\,,
\end{align}
with form domains
\begin{align}
	\mathcal Q\big( \hat q^{(+)}\big) \coloneqq \big\{ f \in \Lp^2(\R^{(+)}\times \R) : \text{ for almost every } k \in \R \colon  f(\cdot, k) \in \mathcal Q\big(q^{(+)}_k\big),  \hat q^{(+)}(f,f)<\infty \big\} \, .
\end{align}
In this paper, we do not need to specify the operator domains and refrain from doing so. Using the Hilbert space tensor product, we can identify these form domains as
\begin{align}
	\mathcal Q(\hat q) &= \big\{f \in \mathsf{H}^1(\R) \otimes \Lp^2(\R) \colon \hat q (f,f) <\infty \big\} \, , \label{hat H form domain}\\
	\mathcal Q\big(\hat q ^{+}\big) &= \big\{ f \in \mathsf{H}^1(\R^+)\otimes \Lp^2(\R) \colon \hat q^+ (f,f) <\infty, f \restriction_{\{0\}\times \R}=0 \big\}  \label{hat H+ form domain} 
	\\
	&= \big\{ f \in \mathsf{H}^1_0(\R^+)\otimes \Lp^2(\R) \colon \hat q^+ (f,f) <\infty \big\}\nonumber\,.
\end{align}
As before, we equip $\mathcal Q\big( \hat q^{(+)}\big)$ with the scalar product $\hat q^{(+)}+\langle \cdot, \cdot \rangle _{\Lp^2(\R^{(+)}\times \R)}$, which turns them into Hilbert spaces. Once again, the linear span of  $\mathcal C^{(+)}$ is dense in $\mathcal Q \big( q^{(+)}\big)$, as we show in \autoref{App C+ density}.

The forms $\hat{q}^{(+)}$ are non-negative, symmetric, and closed. Hence, they determine the unique self-adjoint (and positive) operators $\int_\R^\oplus  \dd k\, H^{(+)}(k)$ such that 
\begin{align} \hat{q}^{(+)}(f,g) = \big\langle f,\int_\R^\oplus  \dd k\, H^{(+)}(k)g\big\rangle \,,
\end{align}
for all $f\in \mathcal Q\big( \hat q^{(+)}\big), g\in \mathcal D\big(\int_\R^\oplus  \dd k\, H^{(+)}(k)\big)$.


As $\mathcal F_2$ is an isometry on $\Lp^2(\R^{(+)}\times \R)$, the claim
\begin{equation} \label{Fourier direct integral identity final eq}
	\mathcal F_2 H^{(+)} \mathcal F_2^*= \int_\R^\oplus  \dd k\, H^{(+)}(k) 
\end{equation}
boils down to showing that $\mathcal F_2$ is an isometry from $\mathcal Q ( \hat q^{(+)})$ to $\mathcal Q (q^{(+)})$;  \eqref{Fourier direct integral identity final eq} is then valid as an identity of operators including the respective operator domains. These domains are determined abstractly through the Friedrichs extension, which we do not specify further.

To verify \eqref{Fourier direct integral identity final eq}, let $f\in \mathcal C^{(+)}$. With an integration by parts, we arrive at
\begin{align*}
	(x_1-k) (\mathcal F_2 f)(x_1,k)&=f_1(x_1) (x_1-k)  (2\pi)^{-1/2}\int_\R \mathrm d t\,f_2(t) \exp(-\mathrm{i}t k)  \\
	&= f_1(x_1)(2\pi)^{-1/2} \int_\R \mathrm d t\,\big(x_1f_2(t)+\mathrm{i} f_2'(t)\big) \exp(-\mathrm{i} t k) \\
	&=f_1(x_1)  \big(\mathcal F (x_1 f_2+\mathrm{i} f_2'\big)(k)\,.
\end{align*}
Thus, for any $f,g \in \mathcal C^{(+)}$, using Plancherel's identity, we get
\begin{align*}
	&\hat q^{(+)}(\mathcal F_2 f, \mathcal F_2 g) \\
	&= \int_\R  \mathrm d k\int_{\R^{(+)}} \mathrm d x_1\,\Big(\overline{\partial_1 (\mathcal F_2f) (x_1,k)} \partial_1  (\mathcal F_2g)(x_1,k) + \overline{(x_1-k) (\mathcal F_2f)(x_1,k)}(x_1-k) (\mathcal F_2g)(x_1,k) \Big)\\
	&=\int_\R \mathrm d k\int_{\R^{(+)}} \mathrm d x_1\, \overline{f_1'(x_1) (\mathcal F f_2)(k) } {g_1'(x_1) (\mathcal F g_2) (k)}  \\
	&+\int_\R \mathrm d k\int_{\R^{(+)}} \mathrm d x_1\,\overline{f_1(x_1)  (\mathcal F (x_1 f_2+\mathrm{i} f_2'))(k)}g_1(x_1)  (\mathcal F (x_1 g_2+\mathrm{i} g_2'))(k)  \\
	&=\int_{\R^{(+)}} \mathrm d x_1\, \overline {f_1'(x_1)} g_1'(x_1) \int_\R \mathrm d k\, \overline{\mathcal F f_2(k)} \mathcal F g_2(k) \\
	&+\int_{\R^{(+)}} \mathrm d x_1 \,\overline{f_1(x_1)} g_1(x_1) \int_\R \mathrm d k\, \overline{(\mathcal F (x_1 f_2+\mathrm{i} f_2'))(k)} (\mathcal F (x_1 g_2+\mathrm{i} g_2'))(k) \\
	&=\int_{\R^{(+)}} \mathrm d x_1\, \overline {f_1'(x_1)} g_1'(x_1) \int_\R \mathrm d x_2 \, \overline{ f_2(x_2)} g_2(x_2) \\
	&+\int_{\R^{(+)}} \mathrm d x_1\, \overline{f_1(x_1)} g_1(x_1) \int_\R \mathrm d x_2\, \overline{(x_1 f_2+\mathrm{i} f_2')(x_2)}  (x_1 g_2+\mathrm{i} g_2')(x_2)  \\
	&= \int_{\R^{(+)}\times \R}   \dd x\, \Big(\overline{\partial_1 f (x)} \partial_1 g(x) + \overline{(x_1+\mathrm{i}\partial_2)f(x)} (x_1 +\mathrm{i} \partial_2)g(x)\Big) \\
	&= q^{(+)}(f,g) \, .
\end{align*}
As $\mathcal F_2 \,\mathcal C^{(+)}=\mathcal C^{(+)}$, we can set $\tilde f = \mathcal F_2 f , \tilde g \coloneqq \mathcal F_2 g$ and thus see for all $\tilde f, \tilde g \in \mathcal C^{(+)}$ that
\begin{align}
	\hat q^{(+)}(\tilde f, \tilde g)= q^{(+)}(\mathcal F_2^* \tilde f, \mathcal F_2^* \tilde g) \,. \label{reverse form identity}
\end{align}
As $q^{(+)}$ and $\hat q^{(+)}$ are sesquilinear and the span of $\mathcal C^{(+)}$ is dense in $\mathcal Q\big(q^{(+)}\big)$, we can conclude that $q^{(+)}(f,g) = \hat q^{(+)}(\mathcal F_2 f, \mathcal F_2 g)$ for all $f,g \in \mathcal Q\big(q^{(+)}\big)$ and $\mathcal F_2 \mathcal Q\big(q^{(+)}\big)\subset \mathcal Q \big(\hat q^{(+)}\big)$. Now, \eqref{reverse form identity} tells us that the we can do this in the other direction and arrive at $\mathcal F_2^* \mathcal Q\big(\hat q^{(+)}\big)\subset \mathcal Q \big(q^{(+)}\big)$. Thus, we have shown $\mathcal F_2 \mathcal Q\big(q^{(+)}\big)= \mathcal Q \big(\hat q^{(+)}\big)$, which was the last step towards proving \eqref{Fourier direct integral identity final eq}.

\section{Proof of local Hilbert--Schmidt norm identity \eqref{HS norm identity intro eq}}\label{Section 4}

In this section, we will establish an identity for the Hilbert--Schmidt norm of the difference of the spectral projection $P_\mu^+$ and of $P_\mu(\R^+\times\R) = 1_{\R^+\times\R}P_\mu1_{\R^+\times\R}$ localized to a square $Q_{a,b} = (a,a+1)\times (b,b+1)$ of side length 1. From this integral expression we will obtain all the necessary estimates in \autoref{Section: Local and global estimates}. 

This identity is expressed in the following theorem.
\begin{thm} \label{HS norm integral rep thm}
	For any $a\in \N,b\in \Z$, we have the Hilbert--Schmidt norm (see \eqref{def Schatten norm} with $p=2$) identity
	\begin{align} 
		&\left \lVert 1_{{(a,a+1)\times(b,b+1)}}\big(P_\mu^+ - P_\mu(\R^+\times\R)\big)
		\right\rVert_{2}^2 \\
		&=\label{HS norm integral rep eq} 
		\frac 1 {2\pi}  \int_{\R} \mathrm d s \int_a^{a+1} \mathrm d x_1
		\int_{\R^+} \mathrm d y_1 \, 
		\big[1(s\ge k_F)\psi_s(x_1)\psi_s(y_1) -\varphi_s(x_1)\varphi_s(y_1) \big]^2 \nonumber 
		\,,
	\end{align}
	where $k_F\in \R^+$ is defined via the identity $\nu(k_F)=\tfrac{\mu-1}2$ with $\nu(\cdot)$ given by \eqref{def of nu(k)}.
\end{thm}

In order to prove this, we need a few intermediate results. 

\begin{lemma}\label{F2 on projections lem}
	For any $\mu \in \R$, we have the identities 
	\begin{align}
		\mathcal F_2 1_{(-\infty,\mu]}\big(H^{(+)}\big) \mathcal F_2^*
		= \int_\R^\oplus  \dd k\, 1_{(-\infty,\mu]}(H^{(+)}(k)) \, .
	\end{align}
\end{lemma}

\begin{proof} See \cite[Theorem XIII.85]{RS4} for an abstract and a (slightly) more general statement replacing the indicator function by a bounded Borel function. In our explicit situation, we argue as follows. At first, extend \eqref{Fourier direct integral identity final eq} to polynomials, e.g.
\begin{align*} \mathcal F_2 \big(H^{(+)}\big)^2 \mathcal F_2^*&= \mathcal F_2 H^{(+)} \mathcal F_2^* \mathcal F_2 H^{(+)} \mathcal F_2^*
\\
&=\int_\R^\oplus  \dd k\, H^{(+)}(k) \cdot \int_\R^\oplus  \dd k\, H^{(+)}(k) 
\\
&=\int_\R^\oplus  \dd k\, \big(H^{(+)}(k)\big)^2\,.
\end{align*}
Here we use that $\mathcal C^{(+)}$ is invariant under $H^{(+)}$. On these sets, we can then prove that for $t\in\R$,
\begin{align}\label{exp direct integral identity}
 \mathcal F_2 \exp\big(\mathrm{i}t H^{(+)}\big) \mathcal F_2^* = \int_\R^\oplus  \dd k\, \exp\big(\mathrm{i}tH^{(+)}(k)\big)\,,
\end{align}
and then extend this to all of $\mathsf{L}^2(\R^{(+)}\times\R)$. Next, we write the indicator function as $1_{[0,\mu]} = \mathcal F^{-1}(\mathcal F 1_{[0,\mu]})$, that is, 
\begin{align*} 1_{[0,\mu]}(x) = \frac{1}{2\pi} \lim_{T\to\infty}\int_{-T}^T\dd t\, \frac{1}{\mathrm{i}t} \big(1-\exp(-\mathrm{i}t\mu)\big)\,\exp(\mathrm{i}t x)\,,\quad x\in\R\,,
\end{align*}
so that by positivity of $H^{(+)}$,
\begin{align*}
\mathcal F_2 1_{(-\infty,\mu]}\big(H^{(+)}\big) \mathcal F_2^* &= \mathcal F_2 1_{[0,\mu]}\big(H^{(+)}\big) \mathcal F_2^* 
\\
&= \frac{1}{2\pi} \lim_{T\to\infty}\int_{-T}^T\dd t\, \frac{1}{\mathrm{i}t} \big(1-\exp(-\mathrm{i}t\mu)\big)\,\mathcal F_2\exp\big(\mathrm{i}t H^{(+)}\big)\mathcal F_2^*\,.
\end{align*}
If we apply \eqref{exp direct integral identity} to this identity, interchange integrals and limits we have proved the formula in this lemma.
 
\end{proof}

\begin{lemma}\label{compact resolvent}
	For any $k \in \R$, the operators $H^{(+)}(k)$ have compact resolvent.
\end{lemma}

\begin{proof}  There are many ways to prove this. One is to prove the stronger trace-class property of the semigroup by using the Golden--Thompson inequality, or somewhat weaker to prove that $\exp(\D)\exp(-V)$ is trace-class. We prefer to use ultra-contractivity, see \cite[Theorem 2]{Simon08}. To apply this theorem, we need to 
	show that both the Laplace operator on $\Lp^2(\R)$ as well as 
	the Dirichlet Laplace on $\Lp^2(\R^+)$ are ultracontractive, which means that their
	respective heat kernels map $\Lp^2$ to $\Lp^\infty$, which can be derived from the 
	known explicit representations of the heat kernels. 
\end{proof}

For the remaining part of this section, we rely on results in \autoref{App PCF}, which discuss properties of the parabolic cylinder functions $D_\nu$ and how they relate to the eigenfunctions of the operators $H^+(k)$. In this Appendix, we use an integral representation of the parabolic cylinder functions to derive all the relevant properties by hand. For this section, we only rely on them for the following lemma:

\begin{lemma}\label{lemma k_F}
	For any $\mu \in (1,3)$, there is a constant $k_F\ge 0$, such that for any 
	$k \ge k_F$, the operator $H^+(k)$ has exactly one (simple) eigenvalue 
	$\le \mu$, while for  $k<k_F$,  $H^+(k)$ has no eigenvalues $\le \mu$. This $k_F$ is defined via the identity $\nu(k_F)=\tfrac{\mu-1}2$ with $\nu(\cdot)$ given by \autoref{nu(k) asymptotics lem}.
\end{lemma}

\begin{proof}
	This mainly relies on \autoref{psik is a Dnu cl}. It tells us that there is at most one eigenvalue below $3$, none for $k\le 0$ and that the eigenvalue satisfies $\lambda(k)=2\nu(k)+1$, where the strictly increasing function $\nu(\cdot)$ is defined in \eqref{def of nu(k)}. Thus, we can define $k_F$ by the property $\nu(k_F)=\tfrac{\mu-1}2$.
\end{proof}

Since $H(k)$ is just a shifted harmonic oscillator, we see that the only eigenvalue below $3$ of $H(k)$ is $1$ with the eigenfunction $\varphi_k(x)=\varphi_0(x-k) =\pi^{-1/4}\exp(-(x-k)^2/2)$. Obviously, the function $\varphi_k$ is an eigenvector of $H(k)$. Since it is positive, it has to be the ground state by Perron--Frobenius. The first excited state of $H(k)$ is $\sqrt{2}\pi^{-1/4}(x-k)\exp(-(x-k)^2/2)$. This is clearly an eigenvector with one root and therefore we have found the first excited state. Its energy eigenvalue is equal to $3$.

We recall that $\psi_k$ is the normalized ground-state eigenfunction of $H^+(k)$. 

\begin{cl}\label{direct integral projector evaluation}
	For $\mu \in (1,3)$ and any $f \in \Lp^2(\R^2), g  \in \Lp^2(\R^{+}\times \R)$, we have the identities 
	\begin{align}
		\left(\int_\R^\oplus  \dd k\, 1_{(-\infty,\mu]}\big(H(k)\big) f \right)(x,t) &= \varphi_t(x) \int_{\R}\mathrm d y\, \varphi_t(y)f(y,t) \, ,\quad (x,t)\in\R^2\,,\\
		\left(\int_\R^\oplus  \dd k\, 1_{(-\infty,\mu]}\big(H^+(k)\big) g \right)(x,t) &= 1_{[k_F,\infty)}(t)\psi_t(x) \int_{\R^+}\mathrm d y \,\psi_t(y)g(y,t) \, ,\quad (x,t)\in\R^+\times\R\,,
	\end{align}
	where $k_F$ is defined via the identity $\nu(k_F)=\tfrac{\mu-1}2$ with $\nu(\cdot)$ given by \autoref{nu(k) asymptotics lem}.
\end{cl}

\begin{remark} Let $\rho\in\mathsf{L}^2(\R)$. If we insert the functions $f_\rho$ and $g_\rho$ defined as $f_\rho(x,k) \coloneqq \varphi_k(x)\rho(k)$ for $(x,k)\in\R^2$ and $g_\rho(x,k) \coloneqq \psi_k(x) 1_{[k_F,\infty)}(k)\rho(k)$ for $(x,k)\in\R^+\times\R$, then we obtain
\begin{align*} 
\int_\R^\oplus  \dd k\, 1_{(-\infty,\mu]}\big(H(k)\big) f_\rho =f_\rho\,,\quad \int_\R^\oplus  \dd k\, 1_{(-\infty,\mu]}\big(H^+(k)\big) g_\rho = g_\rho\,.
\end{align*}
If we choose an orthonormal basis $(\rho_n)_{n\in\N}$ of such functions, either on $\R$ or on $[k_F,\infty)$, then these functions $f_n\coloneqq f_{\rho_n}$ and $g_n\coloneqq g_{\rho_n}$ form an orthonormal basis of eigenvectors of the image of the respective direct-integral projection operator.
\end{remark}

\begin{proof}
	Due to the definition of $k_F$, we see that $1_{(-\infty,\mu]}\big(H^+(k)\big)=0$ for $k< k_F$. Furthermore, for $k\ge k_F$, since $\psi_k$ is the only eigenfunction of $H^+(k)$ with eigenvalue below $\mu$, we get 
	\begin{align}
		1_{(-\infty,\mu]}\big(H^+(k)\big)=1_{[k_F,\infty)}(k) \lvert \psi_k \rangle \langle \psi_k \rvert \, ,
	\end{align}
	where $\lvert \psi \rangle \langle \psi \rvert$ is the usual bra--ket notation for the projection onto the span of $\psi$. 
	
	Similarly, we get
	\begin{align}
		1_{(-\infty,\mu]}\big(H(k)\big)=\lvert \varphi_k \rangle \langle \varphi_k \rvert\,.
	\end{align}
	Taking the direct integral and applying it to the functions $f$ and $g$, 
	\begin{align*}
	\left(\int_\R^\oplus  \dd k\,|\varphi_k\rangle \langle \varphi_k|f\right)(x,t)={\varphi_t}(x)\,\langle\varphi_t,f(\cdot,t)\rangle
	=\varphi_t(x)\int_{\R}\mathrm d y\, \overline{\varphi_t}(y)f(y,t)\,,
	\end{align*}
	and
	\begin{align*}
	\left(\int_\R^\oplus  \dd k\,1_{[k_F,\infty)}(k) \,|\psi_k\rangle \langle \psi_k|g\right)(x,t)&=1_{[k_F,\infty)}(t){\psi_t}(x)\,\langle\psi_t,g(\cdot,t)\rangle
	\\&=1_{[k_F,\infty)}(t){\psi_t}(x)\int_{\R^+}\mathrm d y\, \overline{\psi_t}(y)g(y,t)\,.
	\end{align*}
	Thus, we arrived at the claimed representation, as $\varphi_t$ and $\psi_t$ are real-valued. 
\end{proof}

There is small notational incompleteness in the sense that we did not define $\psi_k$ for $k<0$ but use them in the direct integral representation. We could take the true ground states of $H^+(k)$ or any reasonable family of states for $k<0$ but all this does not matter since the indicator function $1_{[k_F,\infty)}$ in front maps them to 0.

\begin{cl} \label{spectrum of H plus}
The spectrum of $H^+$ inside $(-\infty,3]$ equals $[1,3]$ and is purely absolutely continuous. 
\end{cl}

As we consider $\mu<3$, the set $[1,3]$ is the only spectral set of interest to us. One can show that the full spectrum of $H^+$ equals $[1,\infty)$ and is purely absolutely continuous. 

\begin{proof} 
To show the first claim we use that $\l$ is in the spectrum of $H^+$ if and only if the spectral projection $1_{(\l-\varepsilon,\l+\varepsilon)}(H^+)\not=0$ for all $\varepsilon>0$. Let $\l\in(1,3)$ and, without loss of generality, $\varepsilon$ small enough so that $\l+\varepsilon<3$. By writing $1_{(a,b)} = 1_{(-\infty,b)} - 1_{(-\infty,a]}$, we get 
\[\mathcal F_2 \,1_{(\l-\varepsilon,\l+\varepsilon)}(H^+) \mathcal F_2^{-1} = \int_\R^\oplus\dd k\, 1_{(k_F^+(\l),k_F^-(\l))}(k) |\psi_k\rangle\langle\psi_k|\,,
\]
where $\nu\big(k_F^\pm(\l)\big) = \tfrac12(\l\pm\varepsilon-1)$. The right-hand side is $\not=0$ and hence $\l$ is in the spectrum of $H^+$. [The same argument shows that $(1,3)$ is not part of the spectrum of $H$, because the direct-integral in \autoref{direct integral projector evaluation} does not depend on $\mu$ for $\mu\in(1,3)$ and hence $\mathcal F_2 \,1_{(\l-\varepsilon,\l+\varepsilon)}(H) \mathcal F_2^{-1} = 0$ for $\l\in(1,3)$ and sufficiently small $\varepsilon$.]

To prove the second claim, we take a look at the spectral theorem. It tells us, that there is a spectral measure $E$, such that for any $f \in \Lp^\infty(\R)$, we have the identity
\begin{align}
	f(H^+)=\int_\R f(\lambda)\, \mathrm d E(\lambda)\,.
\end{align} 
Now, for any $\tilde{g}\in\mathsf{L}^2(\R^+\times\R)$, this means
\begin{align}
	\langle \tilde g , f(H^+) \tilde g \rangle = \int_\R f(\lambda) \,\langle \tilde g, \mathrm d E (\lambda) \tilde g \rangle \eqqcolon \int_\R f(\lambda) \,\mathrm d E_{\tilde g}(\lambda) \, .
\end{align}	
The spectrum of $H^+$ in the interval $[1,3]$ is absolutely continuous, if and only if for any $\tilde g$, the Borel measure $E_{\tilde g}$ is absolutely continuous. 

Since we are studying a Borel measure, it suffices to use the test functions $f=1_{(-\infty,\mu]}$ for $\mu \in [1,3]$. 
Using \eqref{eigenvalue equ for H(k) plus} and \autoref{F2 on projections lem}, we find
\begin{align*} 
	\mathcal F_2  1_{(-\infty,\mu]} (H^+)\mathcal F_2^{-1} &= \int_\R^\oplus \dd k\, 1_{(-\infty,\mu]}(H^+(k))   
	\\
	&=\int_{\R^+}^\oplus \dd k \, 1_{(-\infty,\mu]} (2\nu(k)+1)\,|\psi_k\rangle\langle\psi_k|\,.
\end{align*}
For $g \coloneqq \mathcal F_2^{-1} \tilde g$, we compute
\begin{align*}
	\langle \tilde g \mid 1_{(-\infty,\mu]}(H^+) \tilde g \rangle 
	&= \Big\langle g\Big| \int_{\R^+}^\oplus \dd k\, 1_{(-\infty,\mu]} (2\nu(k)+1) \,|\psi_k\rangle\langle\psi_k|g\Big\rangle \\
	&=\int_{\R^+} \dd x\int_{\R^+}\dd t\, \bar{g}(x,t) 1_{(-\infty,\mu]} (2\nu(t)+1)
	\, \psi_t(x) \int_0^\infty \dd y\, \psi_t(y) g(y,t) \\
	&= \int_{\R^+}\dd t \,1_{(-\infty,\mu]} (2\nu(t)+1) \int_0^\infty \dd x\, \bar{g}(x,t) \psi_t(x) \,  \int_0^\infty\dd y\, {g}(y,t) \psi_t(y) \\
	&= \int_{\R^+}\dd t \,1_{(-\infty,\mu]} (2\nu(t)+1) \, \Big|\int_0^\infty \dd x\, {g}(x,t) \psi_t(x)\Big|^2\,.
\end{align*}
By Fubini--Tonelli, since $g \in \Lp^2(\R^+\times\R)$, for almost every $t \in \R$, we have $g(\cdot,t)\in \Lp^2(\R^+)$. Hence, can establish
\begin{align*}
	G(t) \coloneqq \Big|\int_0^\infty \dd x\, {g}(x,t) \psi_t(x)\Big|^2 \le \lVert \psi_t \rVert_{\Lp^2(\R^+)}^2  \lVert g(\cdot,t) \rVert_{\Lp^2(\R^+)}^2 < \infty\,.
\end{align*}
Hence, $G$ is an almost everywhere well-defined function and integrable with $\int_{\R^+} \dd t\, G(t) = \lVert g \rVert_{\Lp^2(\R^+\times \R)}^2$. Now, we perform a change of variables via $\lambda \coloneqq 2 \nu (t)+1$. This uses that $\nu^{-1}(s)=-\zeta(s)/\sqrt{2}$ is strictly decreasing and differentiable due to \autoref{nu(k) asymptotics lem} and \autoref{D nu unique simple zero cl}. Thus, we arrive at
\begin{align*}
	\int_\R 1_{(-\infty,\mu]}(\lambda) \,\dd E_{\tilde g} ( \lambda)
	&= \int_{\R^+} \dd t  \,1_{(-\infty,\mu]} (2\nu(t)+1) G(t) \\
	&= \int_{(1,3)} \dd \l \, 1_{(-\infty,\mu]}(\lambda) G\left(\nu^{-1}\left(\tfrac{\lambda-1}2 \right)\right) \tfrac 1 2 (-\nu^{-1})'\left(\tfrac{\lambda-1}2 \right)\,.
\end{align*} 
As this identity holds for all $\mu \in (-\infty,3]$, the Borel measures on both sides agree on all $f \in \Lp^{\infty}((-\infty,3))$. Thus, we have shown that the measure $E_{\tilde g}$ is absolutely continuous, since we have calculated its density with respect to the Lebesgue measure. Since this holds for any $\tilde g \in \Lp^2(\R^+ \times \R)$, we can conclude that the spectral measure $E$ is absolutely continuous on the interval $(-\infty ,3]$, which means that $H^+$ has only absolutely continuous spectrum on that interval.

\end{proof}

\begin{lemma} \label{F2 indicator lem}
	Let $I,J$ be intervals, with $J=(b_0,b_1)$ for $b_0,b_1 \in \R$. Then, the following  identities hold for all $f \in \Lp^2(\R^2)$ and all $(x_1,x_2)\in\R^2$:
	\begin{align}
		(\mathcal F_2 1_{I \times \R} \mathcal F_2^*f)(x_1,x_2) =& 1_I(x_1) f(x_1,x_2) \, , \\
		(\mathcal F_2 1_{I \times J} \mathcal F_2^*f)(x_1,x_2) =& 1_I(x_1) \int_{\R}  \mathrm dy\,\frac{\exp\big(-\mathrm{i}b_1(x_2-y)\big)-\exp\big(-\mathrm{i}b_0(x_2-y)\big)}{-2\pi \mathrm{i}(x_2-y)   }f(x_1,y) \, .
	\end{align}
\end{lemma}
\begin{remarks} 
\begin{enumerate}
\item[(i)] In the case $(b_0,b_1)=(b,b+1)$, we can express the function inside the integral as
	\begin{align}\label{F2 indicator sine kernel remark}
		\frac{\exp(-\mathrm{i}b_1t)-\exp(-\mathrm{i}b_0 t )}{-2 \pi \mathrm{i} t} =\exp\big(-\mathrm{i}(b+\tfrac 1 2 t)\big) \, \frac{\sin(t/2)}{\pi t} \,.
	\end{align}
\item[(ii)] If we use the second formula with the interval $J=(-b,b)$, $b>0$, we get the expression 
\[ 1_I(x_1) \int_\R \dd y\, \frac{\sin(y)}{\pi y} \, f(x_1,y/b +x_2)\,.
\]
Sending $b\to\infty$ and using $\int_\R \dd y\, \frac{\sin(y)}{\pi y} =1$ we formally obtain the same formula as the first one as it should be the case. 
\end{enumerate}
\end{remarks}

\begin{proof}
	As all the expressions clearly define linear operators and the expressions on the left-hand side clearly are bounded operators on $\Lp^2(\R^2)$, it suffices to prove the identities on a subset with dense span. We choose the set $\mathcal C$, as defined in \eqref{def C}. Thus, let $f \in \mathcal C$. Hence $f(x_1,x_2)=f_1(x_1)f_2(x_2)$ with $f_1\in \mathsf{C}^1_{\text{c}}(\R)$ and $f_2 \in \mathcal S(\R)$. Since $(\mathcal F_2 f)(x_1,x_2)=f_1(x_1)(\mathcal F f_2)(x_2)$ and $(\mathcal F_2^* f)(x_1,x_2)=f_1(x_1)(\mathcal F^* f_2)(x_2)$, the first claim follows as $\mathcal F \mathcal F^*= \operatorname{id}$. For the second claim, we observe
	\begin{align*}
		(\mathcal F_2 1_{I \times J} \mathcal F_2^*f)(x_1,x_2) = 1_I(x_1)f_1(x_1) (\mathcal F 1_J \mathcal F^* f_2)(x_2) \,,
	\end{align*}
	and continue with
	\begin{align*}
		2 \pi (\mathcal F 1_J \mathcal F^* f_2)(x_2) &= \int_{J} \dd s \,\exp(-\mathrm{i} s x_2) \int_{\R} \mathrm d y \,\exp(\mathrm{i} s y) f_2(y) \\
		&= \int_{\R} \mathrm d y \,f_2(y) \int_J \mathrm d s \,\exp\big(-\mathrm{i}s(x_2-y)\big) \\
		&=\int_{\R} \mathrm d y \,f_2(y) \frac{\exp\big(-\mathrm{i}b_1(x_2-y)\big)-\exp\big(-\mathrm{i}b_0(x_2-y)\big)}{-\mathrm{i}(x_2-y)}\,.
	\end{align*}
	The last two equations combined yield the claim for all $f \in \mathcal C$ and thus, by linearity and density, for all $f \in \Lp^2(\R^2)$.
\end{proof}

We have completed the preparations and can now move to the 
\begin{proof}[Proof of \autoref{HS norm integral rep thm}]
	Let us first note that, since $\mathcal F_2$ is a unitary operator, we get 
	\begin{align*}
		\big\lVert 1_{{(a,a+1)\times(b,b+1)}} \big(P_\mu^+ - P_\mu(\R^+\times\R)\big) \big\rVert_{2} =
		\left \lVert \mathcal F_2 1_{{(a,a+1)\times(b,b+1)}} \big(P_\mu^+ - P_\mu(\R^+\times\R)\big) \mathcal F_2^* \right \rVert_{2}  \, .
	\end{align*}
	We recall that $P_\mu^{(+)}=1_{(-\infty,\mu]}(H^{(+)})$ for a $\mu \in (1,3)$ and that $P_\mu(\R^+ \times \R)=1_{\R^+\times \R} P_\mu 1_{\R^+\times \R}$. As the Hilbert--Schmidt norm is the square integral of the integral kernel, we will now find the integral kernel of this operator. Thus, for any $f \in \Lp^2(\R^+\times \R)\subset \Lp^2(\R^2)$ and any $x_1,x_2 \in \R$, we calculate 
	\begin{align*}
		&\left(\mathcal F_2 1_{{(a,a+1)\times(b,b+1)}} \big(P_\mu^+ - P_\mu(\R^+\times\R)\big) \mathcal F_2^* f\right) (x_1,x_2)\\
		&=\left(\mathcal F_2 1_{{(a,a+1)\times(b,b+1)}} \mathcal F_2^* 
		 \big( \mathcal F_2 P_\mu^+ \mathcal F_2^* - 1_{\R^+\times\R} \mathcal F_2 P_\mu\mathcal F_2^* 1_{\R^+\times\R} \big) f\right) (x_1,x_2) \\
		&=1_{(a,a+1)}(x_1) \int_{\R}\mathrm d y_2\,
		 \frac{\exp\big(-\mathrm{i}(b+1)(x_2-y_2)\big)-\exp\big(-\mathrm{i}b(x_2-y_2)\big)}{-2\pi \mathrm{i}(x_2-y_2)   }  \\
		 &\times\left(\left( \int_\R^\oplus \mathrm d k \,1_{(-\infty,\mu]}(H^+(k)) 
		 - 1_{\R^+\times\R} \int_\R^\oplus \mathrm d k  \,1_{(-\infty,\mu]}(H(k)) 1_{\R^+\times\R}  \right) f\right)(x_1,y_2) \\
		 &= 1_{(a,a+1)}(x_1) \int_{\R} \mathrm d y_2\,
		 \exp\big(-\mathrm{i}(x_2-y_2)(b+\tfrac 1 2 )\big) \frac{\sin\big((x_2-y_2)/2\big)} {\pi (x_2-y_2)} \\
		 &\times\Big( 1_{[k_F,\infty)}(y_2)\psi_{y_2}(x_1)\int_{\R^+} \mathrm d y_1 \,\psi_{y_2}(y_1) f(y_1,y_2)
		 -1_{\R^+}(x_1) \varphi_{y_2}(x_1)\int_{\R^+} \mathrm d y_1\, \varphi_{y_2}(y_1)f(y_1,y_2)
		 \Big) \\
		 &=\int_{\R^+}\mathrm d y_1 \int_\R \mathrm dy_2 \,f(y_1,y_2) 1_{(a,a+1)}(x_1) \exp\big(-\mathrm{i}(x_2-y_2)(b+\tfrac 1 2 )\big)\, \frac{\sin\big((x_2-y_2)/2\big)} {\pi (x_2-y_2)} \\
		 &\times\Big( 1_{[k_F,\infty)}(y_2)\psi_{y_2}(x_1) \psi_{y_2}(y_1) 
		 - \varphi_{y_2}(x_1) \varphi_{y_2}(y_1)
		 \Big)\,.
	\end{align*}
	We have used \autoref{F2 indicator lem}, \autoref{F2 on projections lem}, \eqref{F2 indicator sine kernel remark}, \autoref{direct integral projector evaluation} and in the final step $(a,a+1)\subset \R^+$ since $a\in \N$. The final equation provides us with an integral representation of the operator whose Hilbert--Schmidt norm we are interested in. To calculate the square integral of this kernel, we note $\big\lvert \exp\big(-\mathrm{i}(x_2-y_2)(b+\tfrac 1 2 )\big) \big\rvert=1$ and that every other part of this kernel is real-valued. We rename $x_2$ to $k$ and $y_2$ to $s$ to avoid double indices. Hence, we arrive at 
	\begin{align*}
		&\big\lVert 1_{{(a,a+1)\times(b,b+1)}} \big(P_\mu^+ - P_\mu(\R^+\times\R)\big) \big\rVert_{2}^2 \\
		&=\int_{\R}\mathrm d k \int_{\R}\mathrm d s  \int_{a}^{a+1} \mathrm d x_1 \int_{\R^+} \mathrm d y_1 \,\left( \frac{\sin\big((k-s)/2\big)}{\pi(k-  s)}		\big[1(s\ge k_F)\psi_s(x_1)\psi_s(y_1) -\varphi_s(x_1)\varphi_s(y_1) \big]
		\right)^2\,.
	\end{align*}
	The integral over $k$ can be resolved immediately:
	\begin{align}
		\int_\R \mathrm d k\, \frac{\big(\sin((k-s)/2)\big)^2}{\pi^2(k-s)^2} = \int_\R  \mathrm d t\,\frac{\big(\sin(t/2)\big)^2}{\pi^2t^2} = \frac 1 {2\pi}\,.
	\end{align}
	For this explicit calculation, we have used that $\sqrt{2\pi} \sin(t/2)/(\pi t)$ is the Fourier transform of the indicator function of the interval $[-1/2,1/2]$ and the unitarity of the Fourier transform. Thus, we arrive at
	\begin{align*}
		&\big\lVert 1_{{(a,a+1)\times(b,b+1)}} \big(P_\mu^+ - P_\mu(\R^+\times\R)\big) \big\rVert_{2}^2  \\
		&= \frac 1 {2\pi}\int_{\R}\mathrm d k \int_{\R}\mathrm d s  \int_{a}^{a+1} \mathrm d x_1 \int_{\R^+} \mathrm d y_1 \	\big[1(s\ge k_F)\psi_s(x_1)\psi_s(y_1) -\varphi_s(x_1)\varphi_s(y_1) \big]^2\,.
	\end{align*}
	This was our claim.	
\end{proof}

\section{Estimates on parabolic cylinder functions}\label{Estimates on pcf}

In \autoref{App PCF} we recall some basic facts about the parabolic cylinder function, $D_\nu(\cdot)$ which satisfies the differential equation 
\begin{equation} \label{DE for D}
	\frac{\dd^2}{\dd x^2}w(x) - \Big(\frac14 x^2-\nu-\tfrac12\Big) w(x) = 0\,,\quad x\in\R\,.
\end{equation}
$D_\nu(x)$ decays exponentially in $x^2$ as $x \to \infty$. 

By scaling, for any $\lambda \in \R$, the function $w(x)\coloneqq D_{\tfrac{\lambda-1}{2}}\big(\sqrt{2}(x-k)\big)$ satisfies the differential equation,
\[ \frac{\dd^2}{\dd x^2}w(x) - \big((x-k)^2 + \lambda\big)w(x) =0\,,\quad x\in\R\,,
\]
or put differently,
\begin{equation} \label{eigenvalue problem} H^{(+)}(k) w  = \lambda w\,.
\end{equation}
In order for this to define an eigenfunction of $H^+(k)$, it has to satisfy the boundary condition $w(0)=0$ or $D_{\tfrac{\lambda-1}{2}}(-\sqrt 2 k)=0$, see \autoref{eigenfunctions are D nus lem}. There is a bijection between $k \in \R^+$ and eigenvalues $\lambda \in (1,3)$ of $H^+(k)$, denoted by $\lambda= 2\nu(k)+1$, see \autoref{psik is a Dnu cl} and \autoref{nu(k) asymptotics lem}.

In the following three corollaries we derive the essential estimates on parabolic cylinder functions, their norms, and on the eigenfunctions $\varphi_k$ and $\psi_k$. These are based on the results of \autoref{App PCF}. In the first corollary we estimate the pointwise difference of $D_{\nu(k)}(x)$ and $D_0(x)$ for large values of $k$. 
\begin{cl} \label{D nu(k) to D0 diff cl}
For $k \ge 0, \lvert  x \rvert  \le  \sqrt 2 k$ and with $\nu(k)$ defined in \eqref{def of nu(k)}, we have the estimate
\begin{align}
\lvert D_{\nu(k)}(x)- D_0 (x)  \rvert  \le C(1+k) \exp(-k^2/2) \, ,
\end{align}
where $C$ is a constant independent of $k$.
\end{cl}
\begin{proof}
In the case $k \le 1$, the claim is trivial by continuity. For $k \ge 1$, we estimate
\begin{align*}
\lvert D_{\nu(k)}(x)- D_0 (x)  \rvert  &\le \nu(k) \sup_{\nu \in (0,1)} \left \lvert \tfrac {\dd}{ \dd \nu} D_{\nu}(x) \right \rvert \\
&\le  C k \exp(-k^2) \left( \sup_{\nu \in (0,1)} \left \lvert \tfrac {\dd}{ \dd \nu} \tilde D_{\nu}(x)  \right \rvert \Gamma(1+\nu) + \lvert \tilde D_\nu(x) \rvert \Gamma'(1+\nu) \right) \\
&\le  C k \exp(-k^2) \exp(x^2/4) \le C k \exp(-k^2/2) \, . 
\end{align*}
We used $D_\nu(x)=\Gamma(\nu+1) \tilde D_\nu(x)$ (see \eqref{loop repr of D}, \eqref{loop rep of tilde D}) as well as the estimates \eqref{nu(k) asymptotic}, \eqref{d nu D bound} and \eqref{D large x}, \eqref{D large negative x}. The last inequality uses $x^2 \le 2 k^2$. 
\end{proof}

In our next result, we estimate the difference of the $\mathsf{L}^2$-norms of $D_{\nu(k)}$ and of $D_0$, in particular for large $k$.
\begin{cl} \label{D nu(k) norm estimate cl}
For $k \ge 0$ and with $\nu(k)$ defined in \eqref{def of nu(k)}, we have the estimate
\begin{align}
\left \lvert \lVert D_{\nu(k)} \rVert_{\Lp^2((-\sqrt 2 k, \infty))} - \lVert D_0 \rVert_{\Lp^2(\R)} \right \rvert \le C (1+k)^{3/2} \exp(-k^2/2) \, ,
\end{align}
where $C$ is a constant independent of $k$.
\end{cl}

\begin{proof}
We define $I_- \coloneqq (-\infty,-\sqrt 2 k), I_0 \coloneqq (-\sqrt 2 k, \sqrt 2 k) , I_+ \coloneqq (\sqrt 2 k , \infty)$. An application of the triangle inequality leads to the inequality
\[\Big|\|f_0+f_+\| - \|g_-+g_0+g_+\|\Big| \le \Big|\|f_0\|-g_0\|\Big| + \|f_+\|+\|g_0\|+\|g_+\|\,,
\]
which yields
\begin{align*}
& \left \lvert \lVert D_{\nu(k)} \rVert_{\Lp^2((-\sqrt 2 k, \infty))} - \lVert D_0 \rVert_{\Lp^2(\R)} \right \rvert \\
&\le  \left \lvert \lVert D_{\nu(k)} 1_{I_0} \rVert_{\Lp^2(\R)} - \lVert D_0 1_{I_0} \rVert_{\Lp^2(\R)} \right \rvert 
+ \lVert D_{\nu(k)} 1_{I_+} \rVert_{\Lp^2(\R)} + \lVert D_0 1_{I_+} \rVert_{\Lp^2(\R)} +  \lVert D_0 1_{I_-} \rVert_{\Lp^2(\R)}  \\
&\le   \lVert D_{\nu(k)} -D_0 \rVert_{\Lp^2(I_0)} 
+ \lVert D_{\nu(k)} 1_{I_+} \rVert_{\Lp^2(\R)} + \lVert D_0 1_{I_+} \rVert_{\Lp^2(\R)} +  \lVert D_0 1_{I_-} \rVert_{\Lp^2(\R)} \\
&\le  \lVert D_{\nu(k)} -D_0 \rVert_{\Lp^\infty(I_0)} \lVert 1 \rVert_{\Lp^2(I_0)}
+ \lVert D_{\nu(k)} 1_{I_+} \rVert_{\Lp^2(\R)} +2 \lVert D_0 1_{I_+} \rVert_{\Lp^2(\R)} \\
&\le  C (1+k) \exp(-k^2/2) \sqrt k
+ \lVert D_{\nu(k)} 1_{I_+} \rVert_{\Lp^2(\R)} +2 \lVert D_0 1_{I_+} \rVert_{\Lp^2(\R)}\,.
\end{align*}
We have used \autoref{D nu(k) to D0 diff cl} in the final step. In the case $k \le 1$, we are done here. 

In the case $k \ge 1$ we use \eqref{D large x} with $\nu(k)\le1$ and we are left to estimate
\begin{align*}
\int_{\sqrt 2 k} ^\infty \mathrm dx \, (1+x)^2 \exp(-x^2/2)& \le C \int_{\sqrt 2 k} ^\infty \mathrm dx \, x^2 \exp(-x^2/2) \\
& \le C \left(  \left[ - x \exp(-x^2/2) \right]_{\sqrt 2 k}^\infty + \int_{\sqrt 2 k }^\infty \mathrm dx\,\exp(-x^2/2) \right) \\
& \le  C (1+k) \exp(-k^2) \,.
\end{align*}
Putting the two estimates together we obtain that also for $k\ge1$,
\begin{align*}
\left \lvert \lVert D_{\nu(k)} \rVert_{\Lp^2((-\sqrt 2 k, \infty))} - \lVert D_0 \rVert_{\Lp^2(\R)} \right \rvert &\le C (1+k)\sqrt{k} \exp(-k^2/2) + C \sqrt{1+k} \exp(-k^2/2) 
\\
&\le C (1+k)^{3/2} \exp(-k^2/2)\,,
\end{align*}
for some constant $C$ independent of $k$.
\end{proof}

In the last corollary, we use the above results to obtain a pointwise bound on the difference of the eigenfunctions of $H(k)$ and $H^+(k)$. 

\begin{cl} \label{phik - psik pointwise est cl}
For $k \ge 0$ and $x \ge 0$, we have the estimates for the (normalized) eigenfunction $\varphi_k$ of $H(k)$ on the full real line and of the (normalized) eigenfunction $\psi_k$ of $H^+(k)$ on the half real line,
\begin{align}\label{difference of psi and phi}
\lvert \varphi_k(x)- \psi_k(x)  \rvert  \le C
\begin{cases}
(1+k)^{3/2} \exp(-k^2/2) & \text{ if } x \le  2 k \, , \\
(1+x)\exp(-x^2/8) & \text{ if } x \ge  2 k \, ,
\end{cases} 
\end{align}
where $C$ is a constant independent of $k$.
\end{cl}

\begin{remarks}\label{remarks on difference and uniform bound}
\begin{enumerate}
\item[(i)]	This corollary also implies that for all $x \ge 0, k \ge 0$, one has
	\begin{align}
		\lvert \varphi_k(x) - \psi_k (x) \rvert \le C \min \Big((1+2k)^{3/2} \exp(-k^2/2) \, , \,  (1+x)^{3/2} \exp(-x^2/8)\Big) \, .
	\end{align}
	To see this, we replace in the bound on $x\le 2k$ the factor $(1+k)^{3/2}$ by $(1+2k)^{3/2}$ and in the second bound on $x\ge 2k$ the factor $(1+x)$ by $(1+x)^{3/2}$. Then, we use that the function $\lambda \mapsto (1+\lambda)^{3/2} \exp(-\lambda^2/8)$ is decreasing for $\lambda > 2$ and bounded (and strictly positive) before that. Thus, whenever $x \ge 2$ and $k\ge1$, the minimum swaps at $x =2k$ and is always the one compatible with the Corollary. 
	\item[(ii)] \eqref{difference of psi and phi} also shows that $\psi_k(x)$ is uniformly bounded for all $x\ge0$ and $k\ge0$ by simply writing $|\psi_k(x)| \le |\psi_k(x)-\varphi_k(x)| + \varphi_k(x)$. Another (standard and more explicit) way of showing a uniform upper bound to $\|\psi_k\|_{\mathsf{L}^\infty(\R^+)}$ is the Sobolev inequality: To this end, we use $\psi_k(0)=0$ and write 
	\begin{align*}
		|\psi_k(x)^2|&=\Big|\int_0^x \dd t\, \tfrac{\dd }{\dd t} (\psi_k(t)^2)\Big| = 2 \,\Big|\int_0^x \dd t\, \psi_k(t) \tfrac{\dd }{\dd t} \psi_k(t)\Big|\le 2 \int_0^x \dd t\, |\psi_k(t)|\,  \big|\tfrac{\dd }{\dd t} \psi_k(t)\big| 
		\\
		& \le 2 \,\|\psi_k\|_{\mathsf{L}^2(\R^+)} \,\|\tfrac{\dd }{\dd t} \psi_k\|_{\mathsf{L}^2(\R^+)}
		= 2\, \|\tfrac{\dd }{\dd t} \psi_k\|_{\mathsf{L}^2(\R^+)} 
		\le 2 \sqrt{\|\tfrac{\dd }{\dd t} \psi_k\|_{\mathsf{L}^2(\R^+)}^2 + \|(x-k)\psi_k\|_{\mathsf{L}^2(\R^+)}^2}
		\\
		&=2 \sqrt{2\nu(k)+1} \le 2 \sqrt{3}\,.
	\end{align*}
\end{enumerate}
\end{remarks}

\begin{proof}
We begin with the case $x \le 2k$. We note
\begin{align*}
\varphi_k(x)= \frac{\exp(-(x-k)^2/2)}{ \pi ^{\tfrac 1 4}} = \frac{ D_0(\sqrt 2 (x-k) ) } { 2^{-\tfrac 1 4} \lVert D_{0} \rVert_{\Lp^2(\R)}} \, . 
\end{align*}
This looks very similar to \autoref{psik is a Dnu cl}. For $a_1,a_2,b_1,b_2 \in \R^+$, we observe
\begin{align*}
\frac {a_1} {a_2} - \frac {b_1}{b_2}
= \frac {a_1} {a_2} -\frac{b_1}{a_2} +\frac{b_1}{a_2} - \frac {b_1}{b_2}
=\frac{a_1-b_1}{a_2} + b_1 \frac { b_2-a_2}{a_2 b_2} \, .
\end{align*}
Thus, with \autoref{D nu(k) to D0 diff cl} (applied to $\hat x = \sqrt 2 (x-k)\in [-\sqrt 2 k, \sqrt 2 k]$) and \autoref{D nu(k) norm estimate cl}, we get
\begin{align*}
&\lvert \varphi_k(x)- \psi_k(x)  \rvert  \\
 \le& \frac{   \lvert D_0 (\sqrt 2 (x-k )) - D_{\nu(k)}(\sqrt 2 (x-k))  \rvert }{ \pi^{\tfrac 1 4 } } + D_{\nu(k)}(\sqrt 2 (x-k))  \frac {\left \lvert \lVert D_{\nu(k)} \rVert_{\Lp^2((-\sqrt 2 k, \infty))} - \lVert D_0 \rVert_{\Lp^2(\R)} \right \rvert } { 2^{-\tfrac 1 4} \lVert D_{\nu(k)} \rVert_{\Lp^2((-\sqrt 2 k, \infty))}  \lVert D_0 \rVert_{\Lp^2(\R)}  } \\
 \le & C(1+k) \exp(-k^2/2) + \left( \exp(-(x-k)^2/2) + C(1+k) \exp(-k^2/2)  \right) \frac {C(1+k)^{3/2} \exp(-k^2/2) } { \lVert D_{\nu(k)} \rVert_{\Lp^2((-\sqrt 2 k , \infty))}}  \\
  \le & C(1+k) \exp(-k^2/2) + C \frac {(1+k)^{3/2} \exp(-k^2/2) } { \lVert D_{\nu(k)} \rVert_{\Lp^2((-\sqrt 2 k , \infty))}}  \\
 \le & C(1+k)^{3/2} \exp(-k^2/2) \left(1+ \frac 1 { \lVert D_{\nu(k)} \rVert_{\Lp^2((-\sqrt 2 k , \infty))}}\right)\,.
\end{align*}
Thus, we are left to show that the expression $\lVert D_{\nu(k)} \rVert_{\Lp^2((-\sqrt 2 k , \infty))}$ is bounded below by a positive constant. A very rough lower bound can be established by
\begin{align*}
\lVert D_{\nu(k)} \rVert_{\Lp^2((-\sqrt 2 k , \infty))}^2 \ge \inf_{\nu \in [0,1]} \int_{1}^2 \mathrm d x\, D_\nu(x)^2 \ge \inf_{\nu \in [0,1], x \in [1,2]} D_\nu(x)^2 >0 \, ,
\end{align*}
as $D_\nu$ is continuous in $\nu,x$ and has no zeros in the given range. Thus, we have shown the first claim. 

 For the case $x \ge 2k$, we bound
\begin{align*}
\lvert \varphi_k(x)- \psi_k(x)   \rvert & \le  \lvert \varphi_k(x) \rvert + \lvert \psi_k(x) \rvert  \\
& \le C \lvert D_{\nu(k)}(\sqrt 2 (x-k) ) \rvert + C  \lvert D_0(\sqrt 2 (x-k) ) \rvert \\
& \le C (1+(x-k)) \exp(-(x-k)^2/2) \le C (1+x) \exp(-x^2/8) \, ,
\end{align*}
which relies on $\nu(k)\le 1$, \eqref{D large x} and $x-k\ge x/2$. That this constant $C$ can be chosen independent of $k$ uses again the above rough lower bound on $\lVert D_{\nu(k)} \rVert_{\Lp^2((-\sqrt 2 k , \infty))}$.
\end{proof}

To estimate the integral in \eqref{HS norm integral rep eq}, we begin with the only part, that directly interacts with $\psi_k$, which is estimated in the following proposition:

\begin{prop} \label{main prop} 
	There exists a constant $C<\infty$ such that for any $x_1>0$
	\begin{align} \label{main estimate}
		\int_{0}^\infty\dd k\int_0^\infty\dd y_1\,\big|\psi_k(x_1)\psi_k(y_1) -\varphi_k(x_1)\varphi_k(y_1)\big|^2 \le C (1+x_1)^3 \exp(-x_1^2/4)\,.
	\end{align}
\end{prop}

\begin{proof}
As a first step we write the difference in the integrand as
\begin{align}\Big|\psi_k(x_1)\psi_k(y_1) &- \varphi_k(x_1)\varphi_k(y_1)\Big|^2 \nonumber
\\
&=\Big|\big(\psi_k(x_1) - \varphi_k(x_1)\big)\psi_k(y_1) +  \varphi_k(x_1) \big(\psi_k(y_1)-\varphi_k(y_1)\big)\Big|^2 \nonumber
\\
&\le 2 \big|\psi_k(x_1) - \varphi_k(x_1)\big|^2 \,\big|\psi_k(y_1)\big|^2 + 2 \varphi_k(x_1)^2 \,\big|\psi_k(y_1)-\varphi_k(y_1)\big|^2\,,\label{first and second term}
\end{align}
where we used the inequality $|a+b|^2\le 2|a|^2 + 2|b|^2$.

For the first term in \eqref{first and second term} we can perform the integration with respect to the $y_1$ variable and use the normalization $\int_0^\infty\dd y_1\, \big|\psi_k(y_1)\big|^2 =1$, see \eqref{eigenvalue equ for H(k) plus}. It remains to prove that
\begin{align*} \int_{0}^\infty\dd k\,\big|\psi_k(x_1) - \varphi_k(x_1)\big|^2 \le  C (1+x_1)^3 \exp(-x_1^2/4)\,.
\end{align*}
We will employ \autoref{phik - psik pointwise est cl}. Thus, we split the integral domain at $\tfrac 1 2 x_1$ and observe 
\begin{align*}
\int_{0}^\infty\dd k\,\big|\psi_k(x_1) - \varphi_k(x_1)\big|^2 & \le   C x_1 (1+x_1)^2 \exp(-x_1^2/4)+ C \int_{\tfrac 1 2 x_1}^\infty  \mathrm d k\, (1+k)^3 \exp(-k^2)  \\
&\le  C (1+x_1)^3 \exp(-x_1^2/4) +C (1+x_1)^{2} \exp(-x_1^2/4) \\
&\le  C (1+x_1)^3 \exp(-x_1^2/4)\,.
\end{align*}
Here, we bounded $(1+k)^3$ in the integral by $4 (1+k^3)$. For the error function integral, we use $\int_{x_1/2}^\infty\dd k \, \exp(-k^2)\le C \exp(-x_1^2/4)$ since $x_1>0$. In the remaining integral we write $k^3 \exp(-k^2)= -\frac12 k^2\frac{\dd}{\dd k}\exp(-k^2)$ and integrate by parts and once more for the remaining integral with integrand $k\exp(-k^2) = -\frac12 \frac{\dd}{\dd k}\exp(-k^2)$.

Altogether, this shows that the integral over the first term in \eqref{first and second term},
\begin{align*}
\int_0^\infty \dd k\int_0^\infty \dd y_1\, \big|\psi_k(x_1)-\varphi_k(x_1)\big|^2 |\psi_k(y_1)|^2\le C (1+x_1)^3 \exp(-x_1^2/4) \,.
\end{align*}
Now we deal with the second term \eqref{first and second term} and we proceed in a similar manner. We start with the integral with respect to $y_1$. Using \autoref{phik - psik pointwise est cl}, we split the integration domain at $2k$ and observe
\begin{align*} 
 \int_0^\infty \dd y_1\, \big|\psi_k(y_1) - \varphi_k(y_1)\big|^2 & \le 2k C (1+k)^3\exp(-k^2) + C \int_{2k}^\infty \dd y\, (1+y_1)^2 \exp(-y_1^2/4)  \\
 &\le C (1+k)^4 \exp(-k^2) + C (1+k) \exp(-k^2) \le C (1+k)^4 \exp(-k^2) \, .
\end{align*}
Due to an explicitly known formula for $\varphi_k$, we can now tackle the integral with respect to $k$:
\begin{align*}
&\int_0^\infty \dd k\,\varphi_k(x_1)^2\int_0^\infty \dd y_1\, \big|\psi_k(y_1) - \varphi_k(y_1)\big|^2 \\
&\le C \int_0^\infty  \dd k\,  \exp\big(-(x_1-k)^2\big) (1+k)^4 \exp(-k^2) \\
&= C \exp(-x_1^2/2) \int_{-x_1/2}^\infty  \dd k\,\Big(1+\frac{x_1}{2} + k\Big)^4  \exp(- 2k^2)
\\
&\le C \exp(-x_1^2/2) \int_\R  \dd k\,\big((1+x_1)^4 + k^4\big)  \exp(- 2k^2)
\\
&\le C \exp(-x_1^2/2) \big((1+x_1)^4+1\big) \le C  (1+x_1)^3 \exp(-x_1^2/4) \, .
\end{align*}
where we used the inequalities $(a+b)^4  = ((a+b)^2)^2\le (2a^2+2b^2)^2\le 8 (a^4+b^4)$ with $a=1+x_1/2$ and $b=k$ as well as $x_1\ge 0$ in the integrand. 

This yields the required estimate for the integrals of the second term in \eqref{first and second term}. Altogether, we have proved \eqref{main estimate}.
\end{proof}

\section{Local and global estimates on Fermi projections}\label{Section: Local and global estimates}

In the following, we denote for $a\in\N$ and $b\in\Z$ the (open) square $Q_{a,b} \coloneqq \{x=(x_1,x_2)\in\R^+\times\R : a< x_1< a+1,b<x_2< b+1\}$ with left lower corner at $(a,b)$ and with side length equal to 1. For $c\in\R$, let $I_c \coloneqq \{x\in\R : c< x< c+1\}$, then $Q_{a,b} = I_a\times I_b$. As usual, we denote the multiplication operator with the indicator function on $Q_{a,b}$ by the (same) symbol, $1_{Q_{a,b}}$. 

We decided to work with open squares $Q_{a,b}$ but this can be replaced in all bounds by closed squares with no change. For instance, we will use without further notice that $\R^+\times\R = \bigcup_{a\in\N,b\in\Z} Q_{a,b}$ up to a nullset, or $1_{\R^+\times\R} = \sum_{a\in\N,b\in\Z} 1_{Q_{a,b}}$ almost everywhere.

Here is our first crucial estimate about the local difference of the Fermi projection $P_\mu^+=1(H^+\le\mu)$ and $P_\mu(\R^+\times\R) = 1_{\R^+\times\R} 1(H\le\mu) 1_{\R^+\times\R}$, which we consider as an operator on $\mathsf{L}^2(\R^+\times\R)$. 

\begin{prop} \label{HS estimate}
For any $\mu\in(1,3)$ there exists a constant $C$, such that, for any $a\in \N,b\in\Z$, the Hilbert--Schmidt norm of the operator $1_{Q_{a,b}}\big(P_\mu^+ - P_\mu(\R^+\times\R)\big)$ satisfies the bound
\begin{equation} \|1_{Q_{a,b}}\big(P_\mu^+ - P_\mu(\R^+\times\R)\big)\big\|_{2}^{2} \le C\exp(-a)\,.
\end{equation}
\end{prop}

\begin{proof} 
	We rely on \autoref{HS norm integral rep thm}. Thus, we want to estimate  
		\begin{align}
		\frac 1 {2\pi}\int_a^{a+1} \mathrm d x_1   \int_{\R} \mathrm d s 
		\int_{\R^+} \mathrm d y_1  
		\big[1(s\ge k_F)\psi_s(x_1)\psi_s(y_1) -\varphi_s(x_1)\varphi_s(y_1) \big]^2\,.
	\end{align}
	We will ignore the integral over $x_1$ for now. 
	We start by splitting the integral over $s$ into the domains $(-\infty,k_F)$ and $(k_F,\infty)$. Let us begin with the last one of those. Here, we use \autoref{main prop} and observe
	\begin{align*}
		&\int_{k_F}^\infty \mathrm d s \int_{\R^+} \mathrm d y_1 \big[ 1(s\ge k_F)\psi_s(x_1)\psi_s(y_1) -\varphi_s(x_1)\varphi_s(y_1) \big]^2 \\
		&\le \int_{\R^+}\mathrm d s \int_{\R^+} \mathrm d y_1 \big[ \psi_s(x_1)\psi_s(y_1) -\varphi_s(x_1)\varphi_s(y_1) \big]^2 \\
		& \le C (1+x_1)^3 \exp(-x_1^2/4) \le C \exp(-x_1) \, .
	\end{align*}
	For the remaining interval, we use  the normalization $\lVert \varphi_s \rVert_{\Lp^2(\R)}=1$ and $\exp(-a^2/2)\le \exp(\tfrac 12) \exp(a)$. We get
	\begin{align*}
		&\int_{-\infty}^{k_F} \mathrm d s \int_{\R^+} \mathrm d y_1 \,\big[ 1(s\ge k_F)\psi_s(x_1)\psi_s(y_1) -\varphi_s(x_1)\varphi_s(y_1) \big]^2 \\
		&\le \int_{-\infty}^{k_F} \mathrm d s \,\varphi_s(x_1)^2 \int_\R \mathrm d y_1 \,\varphi_s(y_1)^2 \\
		&= \int_{-\infty}^{k_F} \mathrm d s\, \pi^{-\tfrac 1 4} \exp(-(x_1-s)^2/2)  \\
		&\le  \sqrt \mathrm{e} \, \pi^{-\tfrac 1 4} \int_{-\infty}^{k_F} \dd s\, \exp(s-x_1)=C\exp(k_F-x_1) =C\exp(-x_1) \, .
	\end{align*}
	This estimate is pretty bad in the case $k_F>x_1$, but since $k_F$ is just a constant for us, it still suffices.

	Altogether, this shows that
	\begin{align*} 
		\left \lVert 1_{Q_{a,b}}\big(P_\mu^+ - P_\mu(\R^+\times\R)\big)\right \rVert_{2}^2 
		\le & \int_a^{a+1}\mathrm d x_1 \,\frac 1 {2\pi} \left[ C\exp(-x_1)+C\exp(-x_1)\right] \le C \exp(-a) \, ,
	\end{align*}
	which proves the statement of this proposition.
\end{proof}

Let us introduce some more notation. For $p>0$ we denote by
\begin{equation} \label{def Schatten}
\mathcal S^p \coloneqq \mathcal S^p(\R^+\times\R) \coloneqq\big\{A : \mathsf{L}^2(\R^+\times\R)\to \mathsf{L}^2(\R^+\times\R)\mbox{ linear, compact and }  \tr (A^* A)^{p/2} <\infty \big\}
\end{equation} 
the Schatten--von Neumann ideal of compact (linear) operators with $p$-summable singular values. $\mathcal S_p$ is equipped with the $p$(-quasi)-norm 
\begin{equation}  \label{def Schatten norm}
\|A\|_{{\mathcal S}^p} \coloneqq \left(\tr (A^* A)^{p/2}\right)^{1/p} \,.
\end{equation} 
We simply write $\|\cdot\|_p$ for this (quasi-)norm. We denote by $\|\cdot \|$ the operator norm which equals $\|\cdot\|_\infty$. For any $0<p<q \le \infty$, we have
	\begin{equation} \label{p norm monotone}
		\lVert A \rVert_q \le \lVert A \rVert_p \, .
	\end{equation}
For $p\in(0,1]$, the $p$-quasi norm satisfies the $p$-triangle inequality,
\begin{equation} \label{p triangle}
	\|A+B\|_{p}^p\le \|A\|_{p}^p + \|B\|_{p}^p
\end{equation}
rather than the usual triangle inequality. Let us also recall H\"older's inequality in the form
\begin{align}\label{HolderI}
	\|AB\|_r \le \|A\|_p\cdot\|B\|_q\,,
\end{align}
if $r,p,q>0$ and $\frac1r = \frac1p + \frac1q$ (with the usual convention $1/\infty=0$) and in the form
\begin{equation} \label{HolderII}
	\|A\|_{r} \le \|A\|_{p}^\a \|A\|_{q}^{1-\a}\,,
\end{equation}
where $\frac1r = \a \frac1p + (1-\a)\frac1q$ with $r,p,q>0$, $\a\in(0,1)$ and the convention $1/\infty=0$.

Our second local estimate is the following bound

\begin{prop} \label{local p norm of P oben plus} For all $p>0$ there is a constant $C$ depending on $p$ but not on $(a,b)\in\N\times\Z$ such that
\begin{equation} \|1_{Q_{a,b}} \big(P_\mu^+-P_\mu(\R^+\times\R)\big)\|_{p} \le C\,.
\end{equation}
\end{prop}

\begin{proof} By the $p$-triangle inequality \eqref{p triangle}, it suffices to prove that $\|1_{Q_{a,b}} P_\mu^+\|_{p}\le C$ and $\|1_{Q_{a,b}} P_\mu(\R^+\times\R)\|_{p} \le  \|1_{Q_{a,b}}P_\mu\|_{p}\le C$. Let us look first at $1_{Q_{a,b}} P_\mu^+$; for $1_{Q_{a,b}} P_\mu$ such a bound is known, see \cite[Lemma 12]{LSS20}, but we will give an independent proof here, which is taken from \cite[Lemma 4.2.5]{P24thesis} with almost no changes in the proof. The idea is to insert $\exp(-t H^+)\exp(tH^+)$ between the indicator $1_{Q_{a,b}}$ and the projection $P_\mu^+$ with $t$ depending suitably on $p$, basically $t=2/p$. In the end, there will be the operator $\exp(t H^+)P_\mu^+$, which is bounded in operator norm (because $H^+\le \mu$ on $P_\mu^+$) and the operator $1_{Q_{a,b}}\exp(-t H^+)$, which will be bounded in the $p$-quasi norm.  

We use a diamagnetic inequality for the integral kernel of the semigroup, $\exp(-tH^+)$, of $H^+$ 
\begin{equation} \label{diamagnetic} \big|\exp(-t H^+) (x,y)\big| \le \exp(t\D^+)(x,y)
\end{equation}
for all $t>0$, and almost all $x,y$ in $\R^+\times\R$. Here, $\D^+$ is the Laplacian on $\R^+\times\R$ with the same Dirichlet boundary conditions as $H^+$ on $\{0\}\times \R$. Once one uses the Feynman--Kac--It\^{o} formula for $\exp(-t H^+)$ and the triangle inequality, this is trivially proved. See a more general statement in \cite[Theorem 1.1, Remark 1.2(iii)]{HS} including proofs. 

The integral kernel of $\exp(t\D^+)$ is of the form
\begin{align} \exp(t\D^+)&\big((x_1,x_2),(y_1,y_2)\big) 
\\
&= \frac{1}{4\pi t}\Big[\exp\big(-(x_1-y_1)^2/(4t)\big) - \exp\big(-(x_1+y_1)^2/(4t)\big)\Big]\exp\big(-(x_2-y_2)^2/(4t)\big)\,,\nonumber
\end{align}
with $(x_1,x_2),(y_1,y_2)\in\R^+\times\R$. The second factor $\exp\big(-(x_1+y_1)^2/(4t)\big)\le \exp\big(-(x_1-y_1)^2/(4t)\big)$ is not really helpful so that we can drop it and use the bound
\begin{align*}
\Big|\exp(-t H^+) (x,y)\Big| \le \frac{1}{4\pi t}\exp\big(-(x_1-y_1)^2/(4t)\big) \exp\big(-(x_2-y_2)^2/(4t)\big) = \exp(t\D)(x,y)\,,
\end{align*}
where $\Delta$ is the Laplace operator on $\R^2$. By translation invariance of $\Delta$ we can thus fix $a=b=0$ and the uniformity of the estimate is obvious, once it is proved for this special choice.

Let us return to the proof of $\|1_{Q_{a,b}}P_\mu^+\|_{p}\le C$. We claim that for $k\in\{-1\}\cup \N$  and $(c,d)\in\N\times\Z$,
\begin{align*}
\big\|1_{Q_{a,b}} \exp(-2^{k+1}H^+) 1_{Q_{c,d}}\big\|_{2^{-k}}^{2^{-k}} \le C \exp\Big(-2^{-3k-5}\big((a-c)^2 + (b-d)^2\big)\Big)\,,
\end{align*}
where the constant $C$ may depend on $k$ but not on $a,b,c,d$. The proof follows the same lines as \cite[(4.2.29)]{P24thesis}. It is done by induction on $k$. The initial step at $k=-1$ is just a Hilbert--Schmidt norm estimate based on the upper bound for the integral kernel, while the induction step relies on  \eqref{p triangle} and \eqref{HolderI}, which yield
\begin{align*}
	&\big\|1_{Q_{a,b}} \exp(-2^{k+2}H^+) 1_{Q_{c,d}}\big\|_{2^{-(k+1)}}^{2^{-(k+1)}} \\
	\le &  \sum_{(a',b') \in \Z^2} \big\|1_{Q_{a,b}} \exp(-2^{k+1}H^+) 1_{Q_{a',b'}}\big\|_{2^{-k}}^{2^{-k}} \big\|1_{Q_{a',b'}} \exp(-2^{k+1}H^+)  1_{Q_{c,d}}\big\|_{2^{-k}}^{2^{-k}} \, .
\end{align*}
The remaining estimates can be found in \cite[(4.2.29)]{P24thesis}. 

Due to this exponential bound and \eqref{p triangle} it is easy to get rid of the second factor $1_{Q_{c,d}}$,
\begin{align*}
\big\|1_{Q_{a,b}} \exp(-2^{k+1}H^+)\big\|_{2^{-k}}^{2^{-k}} &\le \sum_{a\in\N,b\in\Z}\big\|1_{Q_{a,b}} \exp(-2^{k+1}H^+)1_{Q_{c,d}}\big\|_{2^{-k}}^{2^{-k}} 
\\
&\le C\sum_{a\in\N,b\in\Z} \exp\Big(-2^{-3k-5}\big((a-c)^2 + (b-d)^2\big)\Big)
\\
&\le C\,,
\end{align*}
where $C$ depends on $k$, only. Now, inserting $\exp(-2^{k+1} H^+)\exp(2^{k+1}H^+)$ between $1_{Q_{a,b}}$ and $P_\mu^+$ and using the H\"older inequality \eqref{HolderI}, we have the following estimate 
\begin{align*}
\big\|1_{Q_{a,b}} P_\mu^+\big\|_{2^{-k}}^{2^{-k}} &\le \big\|1_{Q_{a,b}} \exp(-2^{k+1}H^+)\big\|_{2^{-k}}^{2^{-k}} \cdot \big\|\exp(2^{k+1}H^+) P_\mu^+\big\|_{\infty}^{2^{-k}}\le C \exp(2^{k+1}\mu)\,.
\end{align*}
This completes the proof for $p=2^{-k}$. For $p$ different from these special values we can use \eqref{p norm monotone}.

Finally a word to the proof of $\|1_{Q_{a,b}}P_\mu\|_{p}\le C$. Here, we proceed along the same line of arguments since the well-known diamagnetic inequality $\big|\exp(-tH)(x,y)\big|\le \exp(t\D)(x,y)$ can be used from the start.
\end{proof}

The next result combines the local Hilbert--Schmidt estimate of \autoref{HS estimate} and the local $p$-quasi norm estimate of \autoref{local p norm of P oben plus} to a global estimate.

\begin{cl}\label{global p estimate}
For any $p\in(0,1]$, there exists a constant $C$ depending on $p$, $\mu$ and $\L$ but not on $L \ge 1$ such that
\begin{equation} 
\big\|1_{L\L}\big(P_\mu^+-P_\mu(\R^+\times\R)\big)\big\|_{p}^p \le C L\,.
\end{equation}
Moreover, if the distance between $\L$ and the boundary $\R^+\times \R$ is (strictly) positive then 
\begin{equation} \label{global p estimate far from boundary}
\big\|1_{L\L}\big(P_\mu^+ - P_\mu(\R^+\times\R)\big)\big\|_{p} \le C\,.
\end{equation}
\end{cl}

\begin{proof} We insert $1_{L\L} = \sum_{(a,b)\in(\N\times\Z)\cap B_{\sqrt 2}(L\L)}1_{L\L} 1_{Q_{a,b}}$ and apply the $p$-triangle inequality \eqref{p triangle},
\begin{align*}
\big\|1_{L\L} \big(P_\mu^+ - P_\mu(\R^+\times\R)\big)\big\|_{p}^p &\le\sum_{(a,b)\in(\N\times\Z)\cap B_{\sqrt 2}(L\L)}  \big\|1_{Q_{a,b}} \big(P_\mu^+ -P_\mu(\R^+\times\R)\big)\big\|_{p}^p\,.
\end{align*}
Let $A_{a,b}\coloneqq 1_{Q_{a,b}}\big(P_\mu^+-P_\mu(\R^+\times\R)\big)$. We use \autoref{HS estimate} and \eqref{p norm monotone} to see
\[
\lVert A_{a,b} \rVert = \lVert A_{a,b} \rVert_\infty \le \lVert A_{a,b} \rVert_2 \le C \exp(-a/2) \, .
\]
We apply H\"older's inequality \eqref{HolderII} and \autoref{local p norm of P oben plus} to $A_{a,b}$ and obtain
\begin{align*}
\|A_{a,b}\|_{p}^p &\le \|A_{a,b}\|_{p/2}^{p/2}\cdot \|A_{a,b}\|_{\infty}^{p/2}\le C^{p/2} \cdot C^{p/2} \exp(-ap/4)\,.
\end{align*}
As $\lvert b \rvert \le CL$, the number of possible values of $b$ is bounded by $CL$. Thus, we can conclude that
\begin{align*}
\big\|1_{L\L} \big(P_\mu^+ - P_\mu(\R^+\times\R)\big)\big\|_{p}^p &\le C L \sum_{a\in\N} \exp(-a p/4)\le C L\,.
\end{align*}
The final constant $C$ depends through the sum over $b$ also on the diameter of $\L$ in the vertical direction, that is, on the largest vertical line inside $\L$.

The second statement follows from the same line of arguments except that the sum over $a$ runs over the set $\{\lfloor\d L\rfloor,\lfloor\d L\rfloor+1,\ldots\}$ for some $\d>0$ depending on the distance of $\L$ to $\{0\}\times\R$ so that the sum over $a$ is of order $\exp(-\delta p L/4)$. So in fact, \eqref{global p estimate far from boundary} can be improved to an exponentially small upper bound $CL\exp(-\delta pL/4)$, which is however not needed here.
\end{proof}

\section{Proof of main theorem}

The method of proving \autoref{main thm} is standard by now and based on the following inequality of Sobolev \cite[Theorem 2.4]{Sobolev17}: For any function $f\in\mathsf{C}^2(\R\setminus\{x_0\})\cap \mathsf{C}(\R)$ with support in $[x_0-R,x_0+R]$, $x_0\in\R$, $R>0$ and self-adjoint operators $A,B\in\mathcal S^{\s q}$, see \eqref{def Schatten}, 
\begin{equation} \label{Sobolev inequ}
\|f(A)-f(B)\|_{q} \le C R^{\g-\s} \|f\|_{x_0,\g} \,\||A-B|^\s\|_{q}\,.
\end{equation}
Here, $q\le 1$, $\s\in(0,1], (2-\s)^{-1}<q,\s<\g$ and
\begin{equation} 
\|f\|_{x_0,\g} \coloneqq \max_{0\le k\le 2} \sup\big\{|f^{(k)}(x)| |x-x_0|^{-\g+k} : x\in\R\big\}\,.
\end{equation}
The constant $C$ in \eqref{Sobolev inequ} is independent of the operators $A,B$, of the function $f$ and the parameter $R$. By a partition of unity, inequality \eqref{Sobolev inequ} can be extended to functions $f$ which are $\mathsf{C}^2$-smooth except on a finite number of points $x_0,\ldots,x_N$. We may write $f=\sum_{i=0}^N f_i$ with $\|f_i\|_{x_i,\g}<\infty$ such that the single term $\|f\|_{x_0,\g}$ on the right-hand side in \eqref{Sobolev inequ} is replaced by $C\|f\|_{\{x_0,\ldots,x_N\},\g}$ with $\|f\|_{\{x_0,\ldots,x_N\},\g} = \sum_{i=0}^N \|f_i\|_{x_i,\g}$ and a constant $C$ that depends on $N$ and $q$, in general. For $q=1$ the latter constant is 1. 

The R\'enyi entropy functions $h_\a$ are such examples with $x_0=1, x_1=1$, $R=1$ and with $\g = \a$ if $\a<1$ and any $\g<1$ if $\a\ge1$.

We apply inequality \eqref{Sobolev inequ} with $f=h_\a$, $A\coloneqq 1_{L\L}P_\mu^+1_{L\L}$, $B\coloneqq 1_{L\L}P_\mu(\R^+\times\R)1_{L\L}$, and $q=1$. Then,
\begin{align*} \big|S^+_\a(L\L) - S_\a(L\L)\big|&\le \big\|h_\a(A) - h_\a(B)\big\|_1
\\
&\le C  \|h_\a\|_{\{0,1\},\g} \cdot \big\|\big|1_{L\L}\big(P_\mu^+-P_\mu(\R^+\times\R)\big)1_{L\L}\big|^\s\big\|_{1}
\\
&\le C \big\|1_{L\L}\big(P_\mu^+-P_\mu(\R^+\times\R)\big)\big\|_{\s}^\s
\\
&\le C L\,.
\end{align*}
In the last step we have used \autoref{global p estimate}.   The constant $C$ depends (besides $\s$ and $\mu$) on the diameter of $\L$ in the vertical direction, which we can bound by $|\p^+\L|$. Also, in general, $|\p\L|\le 2 |\p^+\L|$. Altogether, this yields $S^+(L\L) = S(L\L) + O(L) \le C L|\p^+\L|$. 

If the distance between $\L$ and $\{0\}\times\R$ is (strictly) positive then we have 
\begin{align*} \big|S^+_\a(L\L) - S_\a(L\L)\big|&\le C\,,
\end{align*}
which yields $S^+_\a(L\L) = L |\p\L|\mathsf{M}_0(h_\a) + o(L)$, since the first term is the leading asymptotic term in the expansion of $S(L\L)$ with error $o(L)$. We do not have a rigorously proven statement for the entanglement entropy beyond the leading term of the order $L$. If we had, combined with the exponentially small bound (in $L$) of \eqref{global p estimate far from boundary}, we would then obtain the same expansion for $S^+_\a(L\L)$ as for $S_\a(L\L)$ if $\L$ has positive distance to the boundary $\{0\}\times\R$.

\section{Decay of Fermi projections and an open question}\label{Discussion}

In \autoref{Section 4}, we saw that the integral kernel of $P_\mu^+$ is of the form
\[
P_\mu^+\big((x_1,x_2),(y_1,y_2)\big) = \int_{k_F}^\infty\frac{\dd k}{2\pi} \, \exp\big(\mathrm{i}k(x_2-y_2)\big)\, \psi_k(x_1)\psi_k(y_1)\,,\quad (x_1,x_2),(y_1,y_2)\in\R^+\times\R\,,
\]
with $\psi_k(x) = N_k D_{\nu(k)}\big(\sqrt{2}(x-k)\big)$; $N_k$ is the normalization factor so that $\|\psi_k\|_{\mathsf{L}^2(\R^+)} = 1$, see the normalization in \eqref{eigenvalue equ for H(k) plus} and \autoref{lemma k_F} for the definition of $k_F$. Then, in the horizontal direction with $t>0$, we get
\begin{align*} 
\lvert P_\mu^+&\big((x_1,x_2),(x_1+t,y_2)\big)  \rvert
\\
&\le \int_{k_F}^\infty\frac{\dd k}{2\pi} \, \psi_k(x_1)\psi_k(x_1+t)
\\
&\le \int_{0}^\infty\dd k\, \psi_k(x_1)\psi_k(x_1+t)
\\
&\le C  \int_0^{x_1+\tfrac t 2 } \mathrm d k\,\psi_k(x_1+t)   + C \int_{x_1 + \tfrac t 2 } ^\infty  \mathrm d k\,\psi_k(x_1) \\
&\le C \int_0^{x_1+\tfrac t 2 } \mathrm d k\,\big[\lvert \psi_k(x_1+t)-\varphi_k(x_1+t) \rvert + \varphi_k(x_1+t)\big] \\
& \quad + C \int_{x_1 + \tfrac t 2 } ^\infty \mathrm d k \,\big[\lvert \psi_k(x_1)-\varphi_k(x_1) \rvert +  \varphi_k(x_1) \big]\\
& \le C \int_0^{x_1+\tfrac t 2 } \mathrm d k\,\big[(1+x_1+t)^{3/2} \exp(-(x_1+t)^2/8) + \exp(-(x_1+t-k)^2/2)\big] \\
& \quad + C \int_{x_1 + \tfrac t 2 } ^\infty \mathrm d k\,\big[(1+2k)^{3/2} \exp(-k^2/2) +  \exp(-(x_1-k)^2/2) \big] \\
&\le C\exp(-t^2/10)\,.
\end{align*}
In this (non-optimal) bound we have used that $k_F\ge0$, $0\le \psi_k(x_1)\le C$ is upper bounded as a function of $k$ and $x_1$, and \autoref{phik - psik pointwise est cl} as formulated in \autoref{remarks on difference and uniform bound}.

	That shows an exponential decay of the integral kernel in the horizontal direction.

To compare this to the situation on the full plane, we have for $(x_1,x_2),(y_1,y_2)\in\R^2$
\begin{align*}
P_\mu\big((x_1,x_2),(y_1,y_2)\big)&=\int_\R \frac{\dd k}{2\pi}\, \exp\big(\mathrm{i}k(x_2-y_2)\big) \varphi_k(x_1)\varphi_k(y_1)
\\
&=\frac{1}{2\pi} \, \exp\big(-(x_1-y_1)^2/4 -(x_2-y_2)^2/4   +\mathrm{i}(x_1+y_1)(x_2-y_2)/2\big)\,.
\end{align*}
Thus, $\big|P_\mu\big((x_1,x_2),(x_1+t,y_2)\big)\big| = (2\pi)^{-1}\exp(-(x_2-y_2)^2/4) \exp(-t^2/4)$ and $\big|P_\mu\big((x_1,x_2),(y_1,x_2+t)\big)\big| = (2\pi)^{-1}\exp(-(x_1-y_1)^2/4) \exp(-t^2/4)$ for any $t\in\R$ imply exponential localization in the horizontal and vertical direction.

More generally, suppose we consider the spectral Fermi projection $P_\mu = 1(\mathcal H\le\mu)$ of a Hamiltonian $\mathcal H$ with an exponentially decaying integral kernel in all directions. Then one can prove an area-law bound of the entanglement entropy as in \eqref{EE bound}, that is, $S_\a(P_\mu(\Lambda))\le C |\partial\Lambda|$. Such a proof in the discrete case is presented in \cite{PasturSlavin14}. To prove the precise area-law scaling is still another matter.

On the other hand, for the half-plane Hamiltonian, in the vertical direction with $t\in\R$, we have 
\begin{align*}
P_\mu^+\big(&(x_1,x_2),(y_1,x_2+t)\big)
\\
&=\int_{k_F}^\infty\frac{\dd k}{2\pi} \, \exp\big(-\mathrm{i}kt\big)\, \psi_k(x_1)\psi_k(y_1)
\\
&=-\frac{1}{2\pi\mathrm{i}t}\int_{k_F}^\infty\dd k \, \frac{\dd}{\dd k}\exp\big(-\mathrm{i}kt\big)\, \psi_k(x_1)\psi_k(y_1)
\\
&=\frac{1}{2\pi \mathrm{i}t} \exp\big(-\mathrm{i}k_Ft\big)\, \psi_{k_F}(x_1)\psi_{k_F}(y_1) +\frac{1}{2\pi\mathrm{i}t}\int_{k_F}^\infty\dd k \, \exp\big(-\mathrm{i}kt\big)\, \frac{\dd}{\dd k}\big(\psi_k(x_1)\psi_k(y_1)\big)\,.
\end{align*}

Assuming that $k \mapsto \tfrac {\mathrm d } {\mathrm dk} \psi_k(x) \in \Lp^1((k_F,\infty))$ for $x\in\{x_1,y_1\}$ and using that $0\le \psi_k(x)\le 2$, we see that the last integral is (in absolute value) bounded by 
\[ 
\frac{1}{2\pi |t|} \int_{k_F}^\infty\dd k\, \Big[\psi_k(x_1)\Big|\frac{\dd}{\dd k}\psi_k(y_1)\Big| +  \psi_k(y_1)\Big|\frac{\dd}{\dd k}\psi_k(x_1)\Big|\Big]\le \frac{C}{|t|}\,.
\]
Thus, for fixed $x_1$ and $y_1$, the integral kernel $P_\mu^+\big((x_1,x_2),(y_1,x_2+t)\big)$  should decay like $C/|t|$ as $t\to\infty$ with a constant $C$ that depends on $x_1,y_1$. This is in fact even weaker than for the Laplace operator in two dimensions, where the decay of the spectral projection (in all directions) is of the order $|t|^{-3/2}$; see formula (4.2) in \cite{PS23} for the spectral projection of the Laplace operator on $\R^2$ and \cite[Proposition 2.2]{PS23} for the decay of the Bessel function. Recall that for the Laplace operator we have a logarithmically enhanced area-law. In \cite{PS22}, we proved a logarithmically enhanced area-law for the ground states of the Landau Hamiltonian on $\R^3$. Here, the Fermi projection (denoted by $\mathrm{D}_\mu$) decays exponentially in the planar coordinates ($x^\perp$) and like $C/|x^\parallel|$ in the vertical (or longitudinal) direction ($x^\parallel$) parallel to the constant magnetic field, see \cite[(2.8)]{PS22}. In contrast to the present situation, this yields an enhanced area-law.

We have now two examples of a Hamiltonian with purely absolutely continuous spectrum, namely the Landau Hamiltonian on $\R^3$ and the Landau Hamiltonian on the half-plane. In both examples, the Fermi projection has a mixed (exponential and polynomial) decay behaviour, and it is the weak polynomial decay which is responsible for absolutely continuous spectrum. In one case this leads to a logarithmically enhanced area-law while in the other case it yields a strict area law. The main purpose of this paper to show that a (purely) absolutely continuous does not suffice to have an enhanced area-law. 

We end this paper with the question: What are sufficient conditions on the spectrum (or spectral measure) of a (one-particle) Schr\"odinger operator, $\mathcal H$, to guarantee a logarithmically enhanced area-law in its ground states defined as $P_\mu = 1(\mathcal H\le\mu)$, $\mu\in\R$?

\begin{appendix}

\section{Density of linear spans of $\mathcal C^{(+)}$ in form domains}\label{App C+ density}

In this appendix, we show that the linear span of $\mathcal C$ (see \eqref{def C}) is dense in $\mathcal Q(q)$ (see \eqref{H form domain}) and $\mathcal Q(\hat q)$ (see \eqref{hat H form domain}) as well as that the linear span of $\mathcal C^+$ (see \eqref{def C}) is dense in $\mathcal Q(q^+)$ (see \eqref{H+ form domain}) and $\mathcal Q(\hat q^+)$ (see \eqref{hat H+ form domain}).

Let us first introduce the support of an $\Lp^2$-function. For $v \in \Lp^2(\R^2)$, let $\operatorname{supp}(v)$ be the smallest closed set $A$, such that the set $\{x \in \R^2 : v(x)=0 , x \not \in A \}$ is a Lebesgue-nullset, or more formally
\begin{align*}
	\R^2 \setminus \operatorname{supp}(v) \coloneqq \bigcup_{\{x \in \Q^2, \varepsilon \in \Q^+ :  \lvert v^{-1}(0) \cap B_\varepsilon(x) \rvert =0 \} } B_\varepsilon(x) \, ,
\end{align*}
which uses a countable basis of the topology of $\R^2$ to show that $ v^{-1}(0) \cap \big(\R^2 \setminus \operatorname{supp}(v) \big)$ is a countable union of nullsets and thus, a nullset. 

Let $\Omega \subset \R^2$ be open. Since $\mathcal Q(q), \mathcal Q(\hat q) \subset \Lp^2(\R^2)$, we can define
\begin{align}
	X(\Omega) \coloneqq \big\{ v \in \mathcal Q(q) : \operatorname{supp}(v) \subset \overline{\Omega} \big\} \, , \quad 
	\hat X (\Omega) \coloneqq \big\{ v \in \mathcal Q(\hat q) : \operatorname{supp}(v) \subset \overline{\Omega} \big\} \, .
\end{align}
We equip $X(\O)$ and $\hat{X}(\O)$ with the graph norm $\lVert v \rVert_{X(\Omega)} \coloneqq \sqrt{q(v,v) + \langle v,v\rangle}$ and $\lVert v \rVert_{\hat{X}(\Omega)} \coloneqq \sqrt{\hat{q}(v,v) + \langle v,v\rangle}$, respectively.

By definition, we have the canonical identification $\mathcal Q (q^+)=X(\R^+ \times \R)$. For the identification $\mathcal Q(\hat q^+)=\hat X(\R^+ \times \R)$, we use the extension theorem \cite[Lemma 3.22]{Adams} on $\mathsf{H}^1_0(\R^+)\subset \mathsf{H}^1(\R)$, which implies that we can extend functions in $\mathcal Q(\hat q^+)$ by zero on the negative half-plane.

Let us also define the linear operators $T \colon \mathcal Q(q) \to \Lp^2(\R^2,\C^3), \hat T \colon \mathcal Q(\hat q) \to \Lp^2(\R^2,\C^3)$:
\begin{align}
	Tv= \begin{pmatrix} v \\ \mathrm{i} \partial_1 v \\ (\mathrm{i} \partial_2 + x_1) v \end{pmatrix} \, , \quad  \hat Tv= \begin{pmatrix} v \\ \mathrm{i} \partial_1 v \\ ( x_1-x_2) v \end{pmatrix} \, .
\end{align}
We see that for any $\Omega$, $\lVert v \rVert_{X(\Omega)}=\lVert Tu \rVert_{\Lp^2(\Omega,\C^3)}$ and $\lVert v \rVert_{\hat X(\Omega)}=\lVert \hat Tv \rVert_{\Lp^2(\Omega,\C^3)}$.

\begin{lemma} \label{X(.) Sobolev lem}
	Let $I,J$ be bounded and open intervals. Then, the following pairs of spaces have equivalent norms:
	\begin{align*}
		X(I\times J ) \simeq \mathsf{H}^1_0(I \times J)\, , \quad \hat X (I \times J) \simeq \mathsf{H}^1_0(I) \otimes \Lp^2(J) \, .
	\end{align*}
\end{lemma}

We use the symbol $\simeq$ between two normed space $X$ and $Y$ to indicate that there is a continuous bijective map between $X$ and $Y$. We write $X=Y$ if this map is unitary.

\begin{proof}
	Let us first note that by a Sobolev space extension theorem for bounded Lipschitz domains~\cite[Lemma 3.22]{Adams} such as $I\times J$ we have 
	\begin{align*}
		\mathsf{H}^1_0(I \times J) &= \big\{v \in \mathsf{H}^1(\R^2) : \operatorname{supp}(v) \subset \overline{I \times J} \big\} \, , \\
		\quad \mathsf{H}^1_0(I) \otimes \Lp^2(J)  &= \big\{v \in \mathsf{H}^1(\R) \otimes \Lp^2(\R)  : \operatorname{supp}(v) \subset \overline{I \times J}\big \} \, .
	\end{align*}
	
	Thus, all spaces above imbed into $\Lp^2(I \times J)= \big\{ v \in \Lp^2(\R^2) : \operatorname{supp}(v) \subset \overline{I \times J} \big\}$. Let $v \in \Lp^2(I \times J) \subset \Lp^2(\R^2)$. Since $I$ and $J$ are bounded, we see that $x_1v , x_2v \in \Lp^2(I\times J)$ with norms bounded by $C\lVert v \rVert_{\Lp^2(\R^2)}$. Thus, $(x_1-x_2)v \in \Lp^2(I \times J)$, which shows that $\hat X (I \times J) \simeq \mathsf{H}^1_0(I) \otimes \Lp^2(J) $ and $(\mathrm{i} \partial_2+x_1)v \in \Lp^2(I\times J)$ if and only if $\mathrm{i} \partial_2v \in \Lp^2(I\times J) $, which yields $X(I\times J ) \simeq \mathsf{H}^1_0(I \times J)$.
\end{proof}

Let $\Phi \in \mathsf{C}^\infty(\R)$ with $\Phi(t)=1$ for $t<-2$, $\Phi(t)=0$ for $t>0$ and $\Phi(\R)= [0,1],\lVert \Phi'\rVert_{\Lp^\infty(\R)}\le 1$. For $R>2$, define $\Phi_R \in \mathsf{C}^\infty_{\text{c}}(\R^2)$ by 
\begin{align}
	\Phi_R(x)\coloneqq \Phi (\lvert x \rvert - R) \, .
\end{align}

\begin{lemma} \label{X(.) finite approx lem}
	Let $\Omega \subset \R^2$ be open. Then, for any $R>2$, the multiplication operator $X(\Omega) \to X(\Omega \cap B_R(0)), v \mapsto \Phi_R v$ is continuous with a norm bounded by $2$ and, for any $f \in X(\O)$, as $R \to \infty$, $\Phi_R v $ converges to $v$ in $X(\Omega)$.  The same holds when we replace all $X(\cdot)$'s by $\hat X(\cdot)$ or $\Lp^2(\cdot)$.
\end{lemma}

\begin{proof}
	For the $\Lp^2(\cdot)$ case, the multiplication operator is clearly bounded by $1$ and the convergence is also trivial. Now, let $\tilde X (\cdot)$ mean either $X(\cdot)$ or $\hat X(\cdot)$, chosen consistently. Let $v \in \tilde X(\Omega)$.
	
	For any function $f \in \mathsf{C}^1(\R^2)$, we observe
	\begin{align}\label{Leibniz T}
		T(vf)=T(v)f + v \begin{pmatrix} 0 \\ \mathrm{i}\partial_1 f \\ \mathrm{i} \partial_2 f \end{pmatrix} \, ,\quad
		\hat T(vf) = \hat T(v)f + v\begin{pmatrix} 0 \\ \mathrm{i}\partial_1 f \\ 0 \end{pmatrix} \, .
	\end{align}
	Since $\Phi_R$ and its differential $D\Phi_R$ are bounded and $\operatorname{supp}(v\Phi_R) \subset \operatorname{supp}(v) \cap \operatorname{supp}(\Phi_R)\subset \Omega \cap B_R(0)$, we can conclude $v\Phi_R \in \tilde X \big( \Omega \cap B_R(0)\big)$. Using $ \lVert  D \Phi_R \rVert_{\Lp^\infty(\R^2)} \le 1$, we get 
	\begin{align*}|T(v\Phi_R)| &= \big|T(v)\Phi_R + v (0,\mathrm{i}\p_1\Phi_R,\mathrm{i}\p_2\Phi_R)^T\big| 
	\le \big|T(v)\Phi_R\big| + |v| \big|(0,\mathrm{i}\p_1\Phi_R,\mathrm{i}\p_2\Phi_R)^T\big|
	\\
	&\le \big|T(v)\big| + |v| \le 2\big|T(v)\big|\,.
	\end{align*}
	Similarly, $|\hat{T}(v\Phi_R)|\le 2\big|\hat{T}(v)\big|$.  Let $\tilde T \in \{T, \hat T\}$ be associated to $\tilde X(\cdot)$. Integrating this over $\Omega$ on both sides yields 
	\begin{align*}
		\lVert v\Phi_R \rVert^2_{\tilde X(\O \cap B_R(0))} = \int_\O \dd x\, |\tilde{T}(v\Phi_R)|^2(x) \le 4\int_\O \dd x\, |\tilde{T}(v)|^2(x)
		=4\lVert v\Phi_R \rVert^2_{\tilde X(\O)} = 4 \lVert v \rVert^2_{\tilde X (\O)} \, .
	\end{align*}
	This was the first claim of this lemma.

	 Since $\lVert 1- \Phi_R \rVert_{\Lp^\infty(\R^2)} \le 1$, $ \lVert  D \Phi_R \rVert_{\Lp^\infty(\R^2)} \le 1$, we use \eqref{Leibniz T} with $f=1-\Phi_R$, and the triangle inequality in $\C^3$ to conclude 
	\begin{align}
		\big\lvert \tilde T \big(v (1-\Phi_R) \big)(x) \big\rvert \le 
		\begin{cases}
			0 &\text{ if } \lvert x \rvert \le R-2\\
			\lvert \tilde T(v)(x) \rvert + \lvert v(x) \rvert & \text{ if } \lvert x \rvert > R-2
		\end{cases}\,.
	\end{align}
	Since $\tilde T(v) \in \Lp^2(\R^2,\C^3)$ and $\lvert v \rvert \in \Lp^2(\R^2)$, as $R \to \infty$, we can show that
	\begin{align*}
		\lVert v- \Phi_R v \rVert_{\tilde X(\Omega)} = \big\lVert \tilde T \big(v (1-\Phi_R)\big) \big\rVert_{\Lp^2(\R^2,\C^3)} \le \lVert \tilde T v \rVert_{\Lp^2(B_{R-2}(0)^\complement,\C^3)} + \lVert v \rVert_{\Lp^2(B_{R-2}(0)^\complement)}  \to 0\,.
	\end{align*}
	This completes the proof.
\end{proof}

\begin{lemma} \label{C density in mixed Sobolev lem}
	Let $I,J$ be bounded, open intervals, and let $Y= \mathsf{H}_0^1(I\times J)$ or $Y=\mathsf{H}_0^1(I)\otimes \Lp^2(J)$. Then, the linear span of the set $\left\{f \in \mathcal C, \operatorname{supp} f \subset \overline{I \times J} \right\}$ is dense in $Y$ and dense in $\Lp^2(I \times J)$.
\end{lemma}

\begin{proof}
	We recall that $f \in \mathcal C$ means that there are $f_1 \in \mathsf{C}^\infty_{\text{c}}\big(\R\big),f_2 \in \mathcal S(\R)$ with $f(x_1,x_2)=f_1(x_1)f_2(x_2)$.  Thus, if $f_1 \in \mathsf{C}^\infty_{\text{c}}(I)$ and $f_2 \in \mathsf{C}^\infty_{\text{c}}(J)$, then $f(x_1,x_2)\coloneqq f_1(x_1)f_2(x_2)$ is certainly in $\mathcal C$. 
	Thus, it suffices to show that such functions are dense in $Y$.

	Hence, now we assume without loss of generality, that $I=(0,\lvert I \rvert), J=(0,\lvert J \rvert)$. The set 
	\begin{align}
		\Big \{ f_{n,m} : f_{n,m} (x_1,x_2) \coloneqq \frac{\sqrt{2}}{\sqrt{\lvert I \rvert \lvert J \rvert}} \sin\Big(\frac{n\pi x_1}{\lvert I \rvert}\Big) \sin\Big(\frac{m\pi x_2}{\lvert J \rvert}\Big) , n,m \in \Z^+, (x_1,x_2)\in I\times J\Big\}
	\end{align}
	is an orthonormal basis of $\Lp^2(I \times J)$, as these are the eigenfunctions of the Dirichlet Laplace operator on $\Lp^2(I \times J)$. By an integration by parts, we can verify that in either case, for any $g \in Y$, 
	\begin{align}
		\langle g , f_{n,m} \rangle_Y = \lVert f_{n,m} \rVert_Y^2 \,\langle g, f_{n,m} \rangle_{\Lp^2(I\times J)} \, .
	\end{align}
	 For instance, for $Y=\Lp^2(I\times J)$, we have
		\begin{align*}
		\langle g , f_{n,m} \rangle_{\Lp^2(I\times J)} &= \langle g,f_{n,m}\rangle_{\Lp^2(I\times J)} + \sum_{|\a|=1}\langle\p^\a g,\p^\a f_{n,m}\rangle_{\Lp^2(I\times J)}
		\\&=\langle g,f_{n,m}\rangle_{\Lp^2(I\times J)} + \Big(\big(\tfrac{n\pi}{|I|}\big)^2 + \big(\tfrac{m\pi}{|J|}\big)^2\Big) \,\langle g,f_{n,m}\rangle_{\Lp^2(I\times J)}
		\\
		&=\|f_{n,m}\|_{\mathsf{H}_0^1(I\times J)}^2 \,\langle g,f_{n,m}\rangle_{\Lp^2(I\times J)}\,.
		\end{align*}
	
	Thus, the functions $f_{n,m}$ are pairwise orthogonal in $Y$ and any function $g$ orthogonal to each $f_{n,m}$ is $0$, which means that $\big\{ f_{n,m}/\lVert f_{n,m} \rVert_Y : n,m \in \Z^+\big\}$ is an orthonormal basis of $Y$. 
	
	Hence, any function in $Y$ can be arbitrarily approximated by a finite linear combinations of the functions $f_{n,m}$. We are thus left to show that the elements of $\mathcal C^{(+)}$ can be used to approximate $f_{n,m}$. For any $\varepsilon >0$, choose a smooth cut-off function  $\xi_{\varepsilon,I} \in \mathsf{C}^\infty_{\text{c}}(I)$ such that $\xi_{\varepsilon,I}(x)=1$ for any $x$ with (distance) $\operatorname{dist}(x,\partial I)>\varepsilon$, $\xi_{\varepsilon,I}(x) \in [0,1]$ for each $x \in I$ and $\left \lVert \xi_{\varepsilon,I}' \right \rVert_{\Lp^\infty(I)} < \frac 2 \varepsilon$. Choose $\xi_{\varepsilon,J}$ in the same way. We see that 
	\begin{align}
		g_{n,m,\varepsilon}(x_1,x_2) \coloneqq \frac{\sqrt{2}}{\sqrt{\lvert I \rvert \lvert J \rvert}} \, \left(\sin\Big(\frac{n\pi x_1}{\lvert I \rvert}\Big)  \xi_{\varepsilon,I}(x_1)\right) \, \left(\sin\Big(\frac{m\pi x_2}{\lvert J \rvert}\Big) \xi_{\varepsilon,J}(x_2)  \right)\,,\quad x_1\in I,x_2\in J\,,
	\end{align}
	is a product of two smooth, compactly supported functions and thus in $\mathcal C$. We are left to show that $g_{n,m,\varepsilon} \to  f_{n,m}$ with respect to the $Y$-norm as $\varepsilon \to 0$. Clearly, the sequence $g_{n,m,\varepsilon}$ converges pointwise. Since $g_{n,m,\varepsilon}$ is uniformly bounded independent of $\varepsilon$, dominated convergence yields norm convergence in $\Lp^2(I \times J)$. We want to do the same for $\nabla g_{n,m,\varepsilon}$. However, the uniform boundedness is a bit less obvious, which is why we elaborate on that. Let 
	\begin{equation*}
		A_{\varepsilon} \coloneqq \big\{ x \in I \times J : \operatorname{dist}\big(x, \partial (I \times J)\big) < \varepsilon\big\}\,,\quad B_\varepsilon \coloneqq (I \times J) \setminus A_\varepsilon\,.
	\end{equation*}
	Since $f_{n,m}$ is clearly Lipschitz continuous, there is a constant $C$ (depending on $n,m$, but not on $\varepsilon$) such that $\lVert f_{n,m} \rVert_{\Lp^\infty(A_\varepsilon)} < C \varepsilon$. Thus, using the product rule, we can estimate 
	\begin{align*}
		& \left \lVert \nabla g_{n,m,\varepsilon} \right \rVert_{\Lp^\infty ( I \times J) } \\
		\le& \max \left( \lVert f_{n,m} \rVert_{\Lp^\infty(A_\varepsilon)} \frac 2 \varepsilon + \lVert \nabla f_{n,m} \rVert_{\Lp^\infty(A_\varepsilon)} \, , \, \lVert f_{n,m} \rVert_{\Lp^\infty(B_\varepsilon)} \cdot 0 + \lVert \nabla f_{n,m} \rVert_{\Lp^\infty(B_\varepsilon)} \cdot 1 \right) \le C\,.
	\end{align*}
	Thus, dominated convergence yields that $\nabla g_{n,m,\varepsilon} \to \nabla f_{n,m}$ in the $\Lp^2(I \times J)$-norm as $\varepsilon\to0$. As we had already shown  convergence in $\Lp^2(I \times J)$, the sequence also converges in $Y$.	  
\end{proof}

\begin{lemma}\label{Q are Hilbert spaces}
	Let $\O \subset \R^2$ be open. Then $X(\O), \hat X (\O)$ are Hilbert spaces.
\end{lemma}	
\begin{proof}
	Let $\tilde X(\cdot)$ be either $X(\cdot)$ or $\hat X(\cdot)$ and let $\tilde T$ (see \eqref{Leibniz T}) be chosen appropriately. Since the graph norm is clearly associated to a positive form, we only have to show that it is complete with respect to that norm. Let $v_k$ be a Cauchy sequence in $\tilde X(\O)$. Since the graph norm is bounded below by the $\Lp^2(\Omega)$ norm, $v_k$ is  a Cauchy sequence in $\Lp^2(\O)$ and thus converges to some $v^* \in \Lp^2(\O)$. In particular $\operatorname{supp}(v^*)\subset \overline{\O}$. 
	
	Due to \autoref{X(.) finite approx lem}, for any $R>2$, we get the Cauchy sequence $v_k \Phi_R \in \tilde X(\O \cap B_R(0)) \subset \tilde X((-R,R)^2) \simeq Y$, where $Y$ is either the Hilbert space $\mathsf H^1_0((-R,R)^2)$ or $\mathsf{H}^1_0((-R,R)) \otimes \Lp^2((-R,R))$ as given in \autoref{C density in mixed Sobolev lem}. Thus, we find a limit function $v^*_R \in \tilde X((-R,R)^2)$ such that $v_k \Phi_R \to v^*_R$ as $k\to\infty$. On the other hand, we observe $v_k \Phi_R \to v^* \Phi_R$ in $\Lp^2(\O)$  as $k\to\infty$. Hence, as the limit is unique, we get that $v^*_R= v^* \Phi_R$ almost everywhere. In particular, for almost every $x \in \R^2$ with $\lvert x \rvert \le R-2$, we can conclude $v^*(x)=v^*_R(x)$. Since $v^*_R \in \tilde X ((-R,R)^2)$, this tells us that $\tilde Tv^* \in \Lp^2_{\text{loc}}(\R^2)$. We estimate
	\begin{align*}
		\left \lVert \tilde T (v^*) \right \rVert_{\Lp^2(\R^2)} 
		 &= \limsup_{R \to \infty} \left \lVert \tilde T (v^*) \right \rVert_{\Lp^2(B_{R-2}(0))} \\
		 &\le \limsup_{R \to \infty}  \left \lVert \tilde T (v^* \Phi_{R}) \right \rVert_{\Lp^2(\R^2)} \\
		 & =\limsup_{R \to \infty} \left \lVert v^*_{R} \right \rVert_{\tilde X(\R^2)} \\
		 & =\limsup_{R \to \infty} \lim_{k \to \infty} \left \lVert f_k \Phi_{R} \right\rVert_{\tilde X(\R^2)} \\
		&\le 2 \lim_{k \to \infty} \left \lVert f_k \right\rVert_{\tilde X(\R^2)} < \infty\,.
	\end{align*}
	Now, we split the cases briefly. For the $X(\O)$ case, we note that since $T(v^*) \in \Lp^2(\R^2,\C^3)$, we can conclude that $v^* \in (\mathsf H^1_{\text{loc}}\cap \Lp^2)(\R^2)$ and thus $v^* \in X(\O)$, while in the other case, $\hat T(v^*) \in \Lp^2(\R^2;\C^3)$ implies that $\p_1 v^* \in \Lp^2(\R^2)$ and thus $v^* \in \mathsf H^1 (\R) \otimes \Lp^2(\R)$ and we get $v^* \in \hat X(\O)$ in that case, too.
	
	We are left to show that $v_k \to v$ in $X(\R^2)$. Let $\varepsilon>0$. As $v_k$ is a Cauchy sequence, there is a $k_0 \in \N$, such that for all $k, \ell \ge k_0$, we have $\lVert v_k - v_\ell \rVert_{X(\Omega)}<\varepsilon$. Furthermore, we choose $R$ large enough such that $\lVert v_{k_0} (1-\Phi_R) \rVert_{X(\Omega)}<\varepsilon $ and $\lVert v^* (1-\Phi_R) \rVert_{X(\Omega)}<\varepsilon $. As all norms in the following calculation are in $\tilde X (\O)$, we drop the index. Let $k \ge k_0$. We begin with the realization that $\lVert v_k \Phi_R- v^* \Phi_R \rVert = \lim_{\ell \to \infty} \lVert v_k \Phi_R- v_\ell \Phi_R \rVert \le 2 \varepsilon$. Now, we can estimate
	\begin{align*}
		\lVert v_k -v^* \rVert \le \lVert (v_k-v^*) \Phi_R \rVert + \lVert v^* (1-\Phi_R) \rVert + \lVert (v_k - v_{k_0})(1-\Phi_R) \rVert + \lVert v_{k_0} (1-\Phi_R) \rVert \le 6 \varepsilon\,.
	\end{align*}
	
\end{proof}

Finally, we conclude with the initial goal of this section. 

\begin{lemma}
	For any of the spaces $\mathcal Q \big(q^{(+)}\big)$, $\mathcal Q\big(\hat q^{(+)}\big)$ and $\Lp^2(\R^{(+)}\times \R)$, the linear spans of the sets $\mathcal C^{(+)}$ are dense subsets. 
\end{lemma}
\begin{remark}\label{Q is dense}
	We point out that this also implies that these form domains are dense in $\Lp^2(\R^{(+)}\times \R)$, as they include the dense subsets $\mathcal C^{(+)}$.
\end{remark}
	
\begin{proof}
	As we discussed at the start of this section, we can identify $\mathcal Q \big(q^{(+)}\big) =X(\R^{(+)}\times \R)$ and $\mathcal Q\big(\hat q^{(+)}\big) = \hat X(\R^{(+)}\times \R)$. We call this space $Z(\R^{(+)}\times \R)$. So let $f \in Z(\R^{(+)}\times \R)$ and $\varepsilon>0$. Due to \autoref{X(.) finite approx lem}, there is an $R>2$, such that $\lVert f-f\Phi_R \rVert_{Z(\R^{(+)}\times \R)}\le \varepsilon$. 
	
	Let $Q_R^{(+)}=(-R,R)^2 \cap (\R^{(+)}\times \R)$. We get $f \Phi_R \in Z(Q_R^{(+)})$ and, due to \autoref{X(.) Sobolev lem} and \autoref{C density in mixed Sobolev lem}, we see that $f\Phi_R$ is in the closure of the linear span of $\Big\{g \in \mathcal C : \operatorname{supp}(g)\subset \overline{Q_R^{(+)}}\Big\} \subset \mathcal C^{(+)}$. Hence, $f$ is in the closure of the span of $\mathcal C^{(+)}$.
\end{proof}

\section{Parabolic cylinder functions}\label{App PCF}

For the convenience of the reader we collect in the following section some basic properties of the parabolic cylinder functions $D_\nu$  that are needed in this paper. In our main case of interest we have $\nu\in(0,1)$ and $x \in \R$, but in some statements below $\nu$ and $x$ may be arbitrary real or even complex numbers. We recall our convention that we denote the natural numbers by $\N = \{0,1,2\ldots\}$ and the negative integers by $\Z^- = \{\ldots,-2,-1\}$.

We also remind the reader of some basic properties of the $\Gamma$-function that we need and refer to \cite{Remmert}, where all proofs of the stated formulas on the $\G$-function can be found. First of all, we define it through Euler's identity (1729), see \cite[Euler's theorem]{Remmert},
\[ \Gamma(z) \coloneqq \int_0^\infty \dd t\, \exp(-t) t^{z-1}\]
for $z$ with real part, $\Re(z)>0$, and continue this function uniquely to $\C\setminus(-\N)$ as a meromorphic function with simple poles of residue $(-1)^n/n!$ at $-n,n\in\N$, see \eqref{residue of Gamma}. It satisfies 
\begin{align} \Gamma(z+1) &= z\Gamma(z) \mbox{ for } z\in\C\setminus (-\N)  \mbox{ and } \Gamma(1)=1: \mbox{``Recurrence relation"}\,,
\\ \label{product of Gamma}
\frac{\pi}{\sin(\pi z)} &=\Gamma(z)\cdot \Gamma(1-z) \mbox{ for } z\in\C\setminus\Z: \mbox{``Euler's reflection property" (1749)}\,,
\\
\Gamma(2z)&=\Gamma(z) \pi^{-1/2} 2^{2z-1} \Gamma(z) \Gamma(z+\tfrac{1}{2}) \mbox{ for } z\in \C:\mbox{``Legendre's duplication formula" (1811)}\,.
\end{align}
The function $z\mapsto 1/\Gamma(z)$ is holomorphic on $\C$ and satisfies the Hankel loop integral representation (1863),
\begin{equation}  \label{Gamma loop integral eq}
\frac1{\Gamma(z)} = \frac1{2\pi\mathrm{i}} \int_{-\infty}^{(0+)}\mathrm{d}t\, \exp(t) t^{-z}\,,\quad z\in\C\,.
\end{equation}
Here, the (Hankel) loop contour starts at $-\infty$, circles the origin once in the positive direction, and returns to $-\infty$, see \cite[Figure 2.2]{Remmert}, \cite[Figure 19.1]{AS}, 
or our definition \eqref{gamma 0}. For $-z\not\in \N$, here and in the following, we use the principal branch of the logarithm to define $t^{-z} = \exp(-z\ln(t))$ in the integral.

The parabolic cylinder function $D_\nu$ satisfies the differential equation \eqref{DE for D} and is uniquely determined by the asymptotic behaviour of $D_\nu(x)$ as $x\to+\infty$ and $x\to-\infty$, see \eqref{D large x} and \eqref{D large negative x}.

For $\nu \in \C \setminus \Z^-$ we will use the following loop-integral representation (see \cite[19.5.1]{AS}) 
\begin{align}\label{loop repr of D}
D_\nu(x)&=\frac{\Gamma(\nu+1)}{2\pi\mathrm{i}} \exp(-x^{2}/4)\,\int_{-\infty}^{(0+)}\mathrm{d}t\, \exp\big(xt-\tfrac{1}{2}t^{2}\big) t^{-(\nu+1)}\,,\quad x\in\C\,,
\end{align}
with the above (Hankel) loop contour. On the range $x \in \mathbb C, \nu \in \C \setminus \Z^-$, this defines a holomorphic function, since the integrand is holomorphic in $x \in \C,\nu \in \C ,t \in \C \setminus (-\infty,0]$ and the complex differential with respect to either $\nu$ or $x$ of the integrand remains absolutely integrable. The restriction on $\nu$ in \eqref{loop repr of D} is due to the Gamma function, which is not defined at $-\N$ and $z\mapsto 1/\Gamma(z)$ has simple zeros at $-n, n\in\N$, see \eqref{Gamma loop integral eq}.

For our application, it is also convenient to define the function without the factor $\Gamma(\nu+1)$. Thus, let
\begin{align}\label{loop rep of tilde D}
\tilde D_\nu(x) \coloneqq \frac{1}{2\pi\mathrm{i}} \exp(-x^{2}/4)\,\int_{-\infty}^{(0+)}\mathrm{d}t\, \exp\big(xt-\tfrac{1}{2}t^{2}\big) t^{-(\nu+1)}\,,\quad x\in\C\,.
\end{align}
This expression is well-defined also for $\nu \in \Z^-$ and in that case $\tilde D_\nu(x)=0$ for all $x\in\C$, as the integrand is holomorphic on all of $\mathbb C$ and thus the closed loop-integral vanishes. 

We also want to point out that $\overline {\tilde D_\nu(x)} = \tilde D_\nu(\overline{x})$, as the complex conjugation changes $x$ to $\overline {x}$ and flips the orientation of the contour integral, which gives a factor $-1$, which cancels with the factor $-1$ from $1/(2\pi \mathrm{i})$. In particular, $D_\nu(\R)\subseteq \R$. 

\begin{remark}
Since the above integral definition \eqref{loop repr of D} of $D_\nu$ does not apply for $\nu \in \Z^-$, we exclude the case $\nu \in \Z^-$ from several upcoming results, which are still true in that case, but are not required for this paper. 
\end{remark}

We need a minor contour integral identity to help with the upcoming lemma that covers the relevant properties of the parabolic cylinder functions for our paper.

\begin{lemma} \label{first alternate loop lem}
	For any $ x, \nu \in \C$ and any $d>0$ we have the identity
	\begin{align}
		\tilde D_\nu(x)= \frac 1 {2\pi \mathrm{i}} \exp(-x^2/4) \left(  \int_{\R-\mathrm{i} d} - \int_{\R+\mathrm{i}d} \right) \mathrm d t \,  \exp\big(xt-\tfrac{1}{2}t^{2}\big) t^{-(\nu+1)}\,.
	\end{align}
\end{lemma}
\begin{proof}
	Let us first make a small detour on contour integrals. For an injective, piecewise $\mathsf{C}^1$-function $\gamma$ from $I$ to $\Omega$ (we write shortly $\gamma\in \mathsf{pC}^1 (I,\Omega)$), where $I=(a,b)$ is a (possibly unbounded) interval and $\Omega \subset \C$ is open, and any holomorphic function $f \colon \Omega \to \C$, the contour integral is defined as
	\begin{align*}
		\int_\gamma \mathrm d t\, f(t)  \coloneqq \int_a^b \mathrm d s\, f(\gamma(s)) \gamma'(s)  \, ,
	\end{align*} 
	provided that the Riemann integral on the right-hand side exists. While there are more general definitions, this suffices here. 
	
	In this paper, $\Omega$  will always be simply-connected and usually $\Omega=\C \setminus (-\infty,0]$. For simply-connected $\Omega$, the homotopy invariance of contour integrals implies that $\int_{\gamma} \mathrm d t \,f(t) = \int_{\gamma'} \mathrm d t \, f(t)$, whenever $\gamma \in \mathsf{pC}^1([a,b],\Omega), \gamma' \in \mathsf{pC}^1([a',b'],\Omega)$ have the same endpoints, that is $\gamma(a)=\gamma'(a'),\gamma(b)=\gamma'(b')$. For that reason, we introduce,  with a slight abuse of notation,
	\begin{align}
		\int_{\gamma(a)}^{\gamma(b)} \mathrm d t\,f(t)  \coloneqq \int_\gamma \mathrm d t\, f(t) \, .
	\end{align}
	We note that when $ \gamma(a),\gamma(b) \in \R$, this agrees with the usual definition of an integral over an interval, since one can pick the path $\gamma(s)=s$ on $(a,b)=(\gamma(a),\gamma(b))$. 
	
	For $t \in \C \setminus (-\infty,0]$, let 
	\begin{align} \label{contour switching lemma proof f def}
		f(t) \coloneqq \exp(-x^2/4) \exp(xt-\tfrac 1 2 t^2)t^{-(\nu+1)} \, .
	\end{align}
	Let $\gamma_0$ be the  Hankel loop. Specifically, we define $\gamma_0 \colon \R \to \C \setminus (-\infty,0]$ by
	\begin{align}\label{gamma 0}
		\gamma_0(s)= 1-s^2 +\mathrm{i} s \exp\left(\tfrac{1-s^2} 2 \right) \, .
	\end{align}
	We note that for each $R>0$, $\gamma_0$ intersects $-R+\mathrm{i}\R$ in exactly two points, $-R-\mathrm{i}\theta(R)$ and $-R+\mathrm{i}\theta(R)$, for $\theta(R)\coloneqq \Im(\gamma(\sqrt{R+1})) \in (0,1)$. Furthermore, we note that there is a constant $C$, such that for any $s \in \R \setminus(-2,2)$, we have
	\begin{align} \label{contour switching lemma proof gamma' eq}
		\lvert \gamma_0'(s) \rvert \le C \lvert s \rvert \, .
	\end{align}
	
	For $R>3,d>0$, let $\gamma_{d,R}$ be a contour constructed as follows. We follow the path $\gamma_0$ from $-\infty$ until it reaches $-R-\mathrm{i}\theta(R)$. From here, we continue with an edge sequence through the points $-R-\mathrm{i}d,R-\mathrm{i}d,R+\mathrm{i}d,-R+\mathrm{i}d,-R+\mathrm{i}\theta(R)$. Then, we rejoin $\gamma_0$ heading to $-\infty$ in the upper half-plane. While this defines the image $\gamma_{d,R}(\R)$, to define the function itself, we say that on the edge sequence segment, $\gamma_{d,R}$ is piecewise linear and $\lvert \gamma_{d,R}' \rvert$ is constant. 

	\begin{figure}
		\centering
		\includegraphics{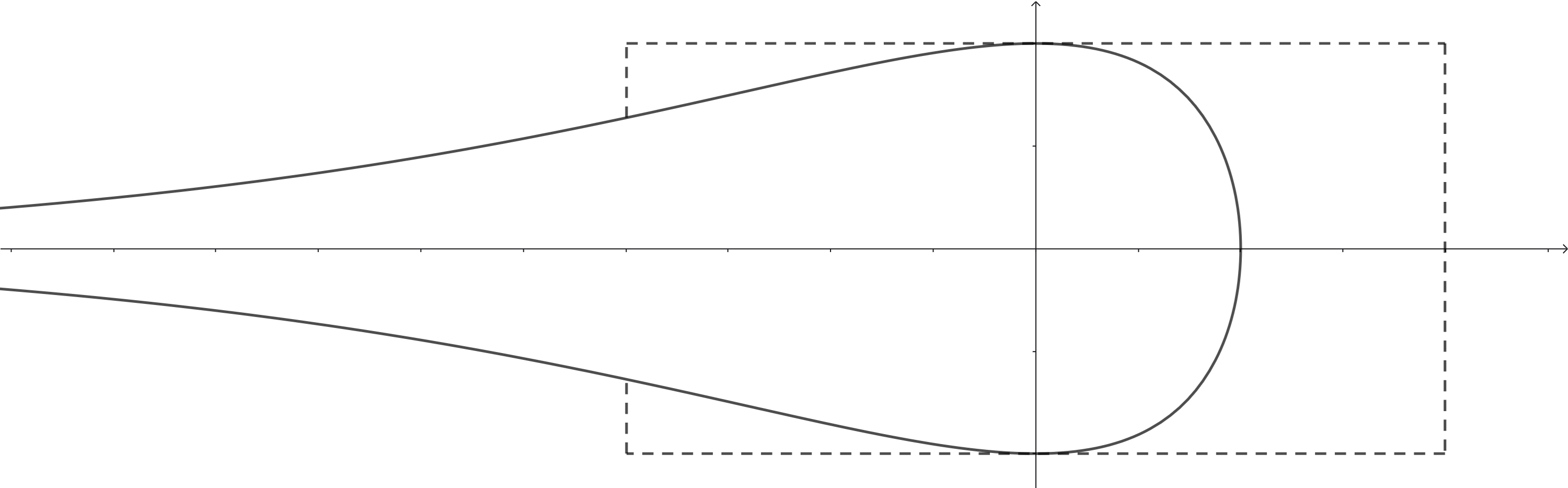}
		\caption{An illustration of $\gamma_0$ and $\gamma_{1,2}$}
	\end{figure}
	
	Since $\C \setminus (-\infty,0]$ is simply-connected the different paths that $\gamma_0$ and $\gamma_{d,R}$ take between $-R-\mathrm{i}\theta(R)$ and $-R+\mathrm{i}\theta(R)$ are homotopic. Hence, the contour integrals with respect to $\gamma_0$ and $\gamma_{d,R}$ are equal. 
	
	We are left to show that the contour integral over $\gamma_{d,R}$ converges to the claimed expression as $R \to \infty$. Therefore, we split the contour integral of $\gamma_{d,R}$ into the remnants of $\gamma_0$, the vertical line segments and the horizontal line segments.
	\begin{align*}
		\int_{\gamma_{d,R}} \dd t\, f(t) &= \int_{-\infty}^{-\sqrt{R+1}} \dd s\, f(\gamma_0(s))\gamma_0'(s) + \int_{\sqrt{R+1}}^\infty \dd s\, f(\gamma_0(s))\gamma_0'(s) \\
		&+\int_{-R-\mathrm{i}\theta(R)}^{-R-\mathrm{i}d} \mathrm dt\, f(t)  + \int_{R-\mathrm{i}d}^{R+\mathrm{i}d}\mathrm dt\, f(t) + \int_{-R+\mathrm{i}d}^{-R+\mathrm{i}\theta(R)} \mathrm dt\,f(t) \\
		&+\int_{-R-\mathrm{i}d}^{R-\mathrm{i}d} \mathrm dt\, f(t) + \int_{-R+\mathrm{i}d}^{R+\mathrm{i}d}\mathrm dt\, f(t) \,.
	\end{align*}
	Clearly, there is a constant $C$, depending on $x,\nu,d$ and $\varepsilon \in (0,1)$, such that
	\begin{align} \label{contour switching lemma proof function estimate eq}
		\lvert f(t) \rvert= \lvert \exp(-x^2/4) \exp(xt-\tfrac 1 2 t^2)t^{-(\nu+1)} \rvert \le C \exp\big(-\tfrac 1 4 \Re(t)^2\big)\,,
	\end{align}
	for all $t \in \R + \mathrm{i} [-d-1,d+1] \setminus \big( (-\infty,0]\cup B_\varepsilon(0) \big)$, where $B_\varepsilon(0)$ is the $\varepsilon$-ball around $0$. For sufficiently small $\varepsilon$, this estimate holds for all $t \in \gamma_0(\R) \cup \bigcup_{R>3} \gamma_{d,R}(\R) \cup (\R-\mathrm{i} d) \cup (\R+\mathrm{i}d)$.  
	
	With this estimate and noting that $R>3$ implies $\sqrt{R+1}\ge 2$, which lets us use \eqref{contour switching lemma proof gamma' eq}, we can proceed to bound the contour integral over the remnant of $\gamma_0$,
	\begin{align*}
		\left \lvert \int_{-\infty}^{-\sqrt{R+1}} \mathrm d s\,f(\gamma_0(s)) \gamma_0'(s) \right \rvert 
		&\le  C \int_{-\infty}^{-\sqrt{R+1}} \mathrm d s\,  \exp\left( - \tfrac {(1-s^2)^2}  4 \right) C \lvert s \rvert \\
		\\&\le C \int_{-\infty}^{-\sqrt{R+1}}  \mathrm d s\,\lvert s \rvert \exp\left( - \tfrac {s^2-1}2\right) = C \exp\left(- \tfrac R 2 \right)
		\, ,
	\end{align*}
	where we used $(1-s^2)^2 = (s^2-1)(s^2-1)\ge 2 (s^2-1)$ for $|s|\ge2$. As the integral over $(\sqrt{R+1},\infty)$ can be estimated in the same way, we conclude
	\begin{align}
		\lim_{R \to \infty} \int_{-\infty}^{-\sqrt{R+1}} \mathrm ds\,f(\gamma_0(s))\gamma_0'(s)  + \int_{\sqrt{R+1}}^\infty \mathrm ds\, f(\gamma_0(s))\gamma_0'(s) =0\,.
	\end{align}
	We proceed with the vertical line segments. Let $J \subset \C$ be such a vertical line segment. For $t \in J$, we get $\lvert \Re(t) \rvert=R$ and $\lvert J \rvert \le 2d+1$. Thus, we can bound $\lvert \int_J \mathrm d t\,f(t) \rvert \le C (2d+1)\exp(-\tfrac 1 4 R^2)$, which yields
	\begin{align}
		\lim_{R \to \infty} \left(\int_{-R-\mathrm{i}\theta(R)}^{-R-\mathrm{i}d}\mathrm dt\, f(t)  + \int_{R-\mathrm{i}d}^{R+\mathrm{i}d} \mathrm dt\,f(t)  + \int_{-R+\mathrm{i}d}^{-R+\mathrm{i}\theta(R)} \mathrm dt\,f(t)\right)  =0 \,.
	\end{align}
	We are left with the horizontal line segments. For these, we start from $\R \pm \mathrm{i}d$ and split that into parts
	\begin{align}
		\int_{\R \pm \mathrm{i}d } \mathrm d t\,f(t)  = \int_{-\infty}^{-R} \mathrm d s\, f(s \pm \mathrm{i}d) + \int_{-R\pm \mathrm{i}d } ^{R \pm \mathrm{i}d } \mathrm d t\,f(t) + \int_{R}^\infty \mathrm d s\, f(s \pm \mathrm{i}d) \,. 
	\end{align}
	For the first (and analogously last) part of this integral, we can use \eqref{contour switching lemma proof function estimate eq} and $R>3$ to conclude
	\begin{align}
		\left \lvert \int_{-\infty}^{-R} \mathrm ds\,f(s \pm \mathrm{i}d ) \right \rvert \le C \int_{-\infty}^{-R} \mathrm ds\,\exp(-\tfrac 1 4 s^2) \le C \int_{-\infty}^{-R} \mathrm ds\,\lvert s \rvert  \exp(-\tfrac 1 4 s^2) = C \exp(-\tfrac 1 4 R^2) \, ,
	\end{align}
	which yields 
	\begin{align*}
		\lim_{R \to \infty} \int_{-R \pm \mathrm{i}d } ^{R \pm \mathrm{i}d} \mathrm d t \, f(t)  = \int_{\R \pm \mathrm{i}d } \mathrm d t \, f(t) \, .
	\end{align*}
	The final step is the realization that 
	\begin{align*}
		\int_{R+\mathrm{i}d}^{-R+\mathrm{i}d} \mathrm dt\, f(t)  = - \int_{-R+\mathrm{i}d} ^{R+\mathrm{i}d} \mathrm dt\, f(t) \,,
	\end{align*}
	which is just based on $\gamma'$ being changed exactly by a factor $-1$, when we run a path in opposite direction.
	
	Thus, we have shown
	\begin{align*}
		\int_{\gamma_0}\mathrm d t \,f(t) = \lim_{R \to \infty} \int_{\gamma_{d,R}} \mathrm d t \,f(t) = \left( \int_{\R-\mathrm{i}d} - \int_{\R+\mathrm{i}d} \right)\mathrm d t \, f(t) \, ,
	\end{align*}
	which, in light of \eqref{loop rep of tilde D}, \eqref{contour switching lemma proof f def} and the fact that $\gamma_0$ is the Hankel loop, was the claim.	
\end{proof}
\begin{cl} \label{2nd alternate loop cl}
	For any $x \in \C, d>0$, we have the identity 
	\begin{align}
		\tilde D_\nu(x)= \frac 1 {2\pi \mathrm{i}} \exp(-x^2/4) \left(  \int_{(-\infty,d)-\mathrm{i} d} + \int_{d-\mathrm{i}d}^{d+\mathrm{i}d}- \int_{(\infty,d)+\mathrm{i}d} \right) \mathrm d t \,  \exp\big(xt-\tfrac{1}{2}t^{2}\big) t^{-(\nu+1)}\,.
	\end{align}
\end{cl}
\begin{proof}
	We follow the same construction as in the previous proof. The only notable difference is that the two corners $R \pm \mathrm{i}d$ of the path $\gamma_{d,R}$ are replaced by $d \pm \mathrm{i}d$. All the required estimates and identities have been discussed in the previous proof.
\end{proof}

Let us collect a few known properties of the parabolic cylinder functions in the following lemma. 
\begin{lemma}
The following properties hold:
\begin{enumerate}
	\item For all $\nu \in \R$, the function $\tilde D_\nu$ satisfies \eqref{DE for D}, that is,
	\begin{align} \label{tilde D nu diff eq}
		\frac{\dd^2}{\dd x^2} \tilde D_\nu(x) - \big(\tfrac14 x^2-\nu-\tfrac12\big) \tilde D_\nu(x) = 0\,,\quad x\in\mathbb C\,.
	\end{align}
	\item For all $\nu \in \R \setminus \Z^-$, we have the two-term asymptotic expansion as $x\in\R$ tends to $+\infty$,
	\begin{align} \label{D large x}
		D_\nu(x)&=\exp(-x^2/4)\,x^\nu \Big(1+\mathcal O\big( \tfrac{1}{x^2}\big) \Big)\,,
	\end{align}
	where the convergence is locally uniform in $\nu$.
	\item For all $\nu \in \R \setminus \Z$, we have the leading asymptotic behaviour as $x\in\R$ tends to $- \infty$,
	\begin{align} \label{D large negative x}
		D_\nu(x) = \frac{\sqrt{2\pi}}{\Gamma(-\nu)} \exp(x^2/4) \,|x|^{-\nu-1}\big(1+o(1)\big)\,.
	\end{align}
	In our main case of interest, $\nu\in(0,1)$ so that $1/\Gamma(-\nu)$ is negative and then $D_\nu(x)$ goes to $-\infty$ as $x\to-\infty$.
	\item For all $\nu > -1, x \le 0$ and $j=1,2$, we have
	\begin{align} \label{d nu D bound}
		\Big| \frac{\dd^j}{\dd\nu^j}\tilde  D_\nu(x)\Big| &\le  C \exp(x^2/4)\,,
	\end{align}
	with a constant $C$ independent of $\nu$ and $x$.
	\item As $x\in\R$ tends to $- \infty$, we have
	\begin{align} \label{d nu D at nu = 0}
		 \frac{\dd }{\dd \nu} \tilde D_\nu(x)\Big|_{\nu=0}= -\frac{\sqrt{2\pi}}{|x|} \exp(x^2/4) \Big(1+\mathcal O(|x|^{-2})\Big)  \,.
	\end{align}
	\item For all $\nu  \in \Z^-$, as $x\in\R$ tends to $\pm \infty$ we have
	\begin{align} \label{D_-n(ix)}
		 D_\nu(\mathrm{i}x) \coloneqq  \lim_{\nu' \to \nu } D_{\nu'}(\mathrm{i}x)=  \exp(x^2/4) (\mathrm{i}x)^{ \nu} \big(1+ o(1)\big) \,.
	\end{align}
	\item For all $\nu \in \R$ and all $ x \in \R$, we have
	\begin{align} \label{exp D diff 1}
		\frac{\dd}{\dd x}\left(\exp(x^2/4)\tilde D_\nu(x) \right)= \exp(x^2/4) \tilde D_{\nu-1}(x) \,.
	\end{align}
	\item For all $\nu \in \R$ and all $ x \in \R$, we have
	\begin{align} \label{exp D diff 2}
		\frac{\dd}{\dd x}\left(\exp(-x^2/4)\tilde D_\nu(x) \right)=- (\nu+1)\exp(-x^2/4) \tilde D_{\nu+1}(x) \,.
	\end{align}	
	\item For all $\nu \in \R \setminus \Z^-$, we have
	\begin{align} \label{D nu at 0 final}
		D_\nu(0) = \frac{2^{\tfrac{\nu}{2}}}{\sqrt{\pi}} \,\Gamma\big(\tfrac{\nu+1}{2}\big) \sin \big( \pi \big(\tfrac{\nu+1}{2}\big) 				\big) \,.
	\end{align}	
\end{enumerate}
\end{lemma}

\begin{proof}
Let $w\coloneqq2\pi \mathrm{i} \tilde D_\nu$, that is,
\[ w(x)= \exp(-x^{2}/4)\,\int_{-\infty}^{(0+)}\mathrm{d}t\, \exp\big(xt-\tfrac{1}{2}t^{2}\big) t^{-(\nu+1)}\,,\quad x\in \C \,.
\]
As the integrand is smooth along the contour and decays along with the relevant differentials as $t \to - \infty$, we can exchange the differential in $x$ with the integral twice. Thus (${}' = \frac{\dd}{\dd x}$), 
\begin{align*} w'(x)&=-\frac{x}{2}\exp(-x^{2}/4)\,\int_{-\infty}^{(0+)}\mathrm{d}t\, \exp\big(xt-\tfrac{1}{2}t^{2}\big) t^{-(\nu+1)}+\exp(-x^{2}/4)\,\int_{-\infty}^{(0+)}\mathrm{d}t\, \exp\big(xt-\tfrac{1}{2}t^{2}\big) t^{-\nu}
\end{align*}
so that
\begin{align*}
w''(x)&=\Big(-\frac12 +\frac{x^2}{4}\Big) \exp(-x^{2}/4)\,\int_{-\infty}^{(0+)}\mathrm{d}t\, \exp\big(xt-\tfrac{1}{2}t^{2}\big) t^{-(\nu+1)}
\\
&-2\cdot\frac{x}{2} \exp(-x^{2}/4)\,\int_{-\infty}^{(0+)}\mathrm{d}t\, \exp\big(xt-\tfrac{1}{2}t^{2}\big) t^{-\nu}
\\
&+\exp(-x^{2}/4)\,\int_{-\infty}^{(0+)}\mathrm{d}t\, \exp\big(xt-\tfrac{1}{2}t^{2}\big) t^{-\nu+1}\,.
\end{align*}
Writing in the last integral $\exp(-t^2/2) t = -\frac{\dd}{\dd t} \exp(-t^2/2)$ and integrating by parts we obtain 
\begin{align*}
w''(x) &= \Big(-\frac12 +\frac{x^2}{4}\Big) w(x)
\\
&-{x}\exp(-x^{2}/4)\,\int_{-\infty}^{(0+)}\mathrm{d}t\, \exp\big(xt-\tfrac{1}{2}t^{2}\big) t^{-\nu}
\\
&+\exp(-x^{2}/4)\,\int_{-\infty}^{(0+)}\mathrm{d}t\, \exp\big(-\tfrac{1}{2}t^{2}\big) \big(xt-\nu\big) \exp(xt)t^{-\nu-1}
\\
&=\Big(-\frac12 +\frac{x^2}{4}-\nu\Big) w(x)\,.
\end{align*}
Thus, we have shown \eqref{tilde D nu diff eq}. 

To obtain the asymptotic expansion as $x\to + \infty$, we follow Whittaker (see \cite{Whitt1902}), change coordinates and write
\[ D_\nu(x) = \frac{\Gamma(\nu+1)}{2\pi\mathrm{i}} \exp(-x^{2}/4)\,x^\nu \,\int_{-\infty}^{(0+)}\mathrm{d}t\, \exp\big(t-\tfrac{1}{2}t^{2}/x^2\big) t^{-(\nu+1)} \, .
\]
For the original integral we use \autoref{2nd alternate loop cl} with $d\coloneqq x$ to replace the Hankel loop. After the integral transformation, we get the same type of loop with $d=1$. Thus, we can write
	\[ 
	D_\nu(x) = \frac{\Gamma(\nu+1)}{2\pi\mathrm{i}} \exp(-x^{2}/4)\,x^\nu \,\left( \int_{(-\infty,1)-\mathrm{i}}+\int_{1-\mathrm{i}}^{1+\mathrm{i}}-\int_{(-\infty,1)+\mathrm{i}}\right) \mathrm{d}t\, \exp\big(t-\tfrac{1}{2}t^{2}/x^2\big) t^{-(\nu+1)} \, .
	\]
We call this piecewise linear path $\gamma$. For the error bound, we can now observe that $\Re(t^2)\ge -1$ along this contour. For $x>1$, this yields $\Re(t^2/x^2)>-1$, which implies 
\[
\frac{\big\lvert \exp(-\tfrac 1 2 t^2/x^2)-\exp(0) \big\rvert }{\big\lvert (-\tfrac 1 2 t^2/x^2) - 0 \big\rvert} \le \sup_{z \in \C, \Re(z)<1} \left \lvert \frac {\mathrm d }{\mathrm d z} \exp(z) \right \rvert=\mathrm{e} \, .
\]
Thus, we can estimate (by the triangle inequality for complex path integrals)
\[
\left \lvert \int_{-\infty}^{(0+)}\mathrm{d}t\, \exp(t) \left( \exp\big(-\tfrac{1}{2}t^{2}/x^2\big)-1 \right) t^{-(\nu+1)} \right \rvert \le \frac {\mathrm{e}}{2 x^2} \int_{\gamma(\R)} \mathrm{d} \mathcal H^1(t)  \, \lvert \exp(t) \rvert \,\lvert t\rvert ^{2-(\nu+1)}\,,
\]
where $\mathcal H^1$ is the one-dimensional Hausdorff-measure on $\g(\R)$; in the literature $\mathrm{d} \mathcal H^1(t)$ is also denoted by $|\dd t|$ for complex $t$.  
Since this integral is a continuous, real-valued function in $\nu$, this yields the error estimate in \eqref{D large x}. 

For the main term, we still need to change the contour path from $\gamma$ back to a Hankel loop like $\gamma_0$, see \eqref{gamma 0}. We use the same arguments that led to the proof of \autoref{2nd alternate loop cl} with $d\coloneqq 1$ and the function $\exp(t)$ instead of $\exp(xt-t^2/2)$. In addition, we use \eqref{Gamma loop integral eq} to see that 
\[
	\int_\gamma \mathrm{d}t\, \exp(t) t^{-(\nu+1)} =\int_{\gamma_0} \mathrm{d}t\, \exp(t) t^{-(\nu+1)} = \frac {2 \pi \mathrm{i}} { \Gamma(\nu+1)} \,,
\]
which completes the proof of \eqref{D large x}.

On the other hand, if we let $x\to-\infty$, then we first apply \autoref{first alternate loop lem} and then we complete the square in the exponent, 
so that
\begin{align} 
	\tilde D_\nu(x) &= \frac{1}{2\pi\mathrm{i}} \exp(x^{2}/4)\,\left(\int_{\R-\mathrm{i}d}-\int_{\R+\mathrm{i}d}\right) \mathrm{d}t\, \exp\big(-\tfrac{1}{2}(t-x)^{2}\big) t^{-(\nu+1)}\label{D nu modified integral original Hankel loop}
	\\
	&=\frac{1}{2\pi\mathrm{i}} \exp(x^{2}/4)\,\, \left(\int_{\R-\mathrm{i}d}-\int_{\R+\mathrm{i}d}\right)\mathrm{d}t\,  \exp\big(-\tfrac{1}{2}t^{2}\big) \left(t+x\right)^{-(\nu+1)}  \label{D nu modified integral} \\
	&=\frac{1}{2\pi\mathrm{i}} \exp(x^{2}/4)\,|x|^{-\nu-1}\, \left(\int_{\R-\mathrm{i}d}-\int_{\R+\mathrm{i}d}\right)\mathrm{d}t\, \exp\big(-\tfrac{1}{2}t^{2}\big) \left(-1+\frac{t}{\lvert x \rvert}\right)^{-(\nu+1)}\,. \nonumber
\end{align}
Since this holds for any $x<0$, we may now choose $d=1/\lvert x\rvert$. Thus, we are left to study 
\begin{align*}
	\lim_{d \to 0^+} \int_\R \mathrm d s &\left( \exp\big(-\tfrac 1 2 ( s-\mathrm{i}d )^2\big) \left(-1+d (s-\mathrm{i}d)\right)^{-(\nu+1)} \right.
	\\
	&\left.- \exp\big(-\tfrac 1 2 ( s+\mathrm{i}d )^2\big) \left(-1+d (s+\mathrm{i}d)\right)^{-(\nu+1)} \right)\,.
\end{align*}
To apply dominated convergence, we need to be a bit careful around $sd \approx 1$, since $(-1+d(s\pm \mathrm{i}d))^{-(\nu+1)}$ may reach up to $d^{-2(\nu+1)}$, if $\nu>-1$. However, $\big\lvert \exp(-\tfrac 12 (s \pm \mathrm{i}d)^2)\big\rvert$ easily compensates this, as it can be bounded by $C\exp(-\tfrac 1 d)$ in the critical case $s\in(\tfrac 1{2d}, \tfrac 3 {2d})$, for small $d$. Thus, we may apply dominated convergence.  Since the expressions $(-1+ d (s\pm \mathrm{i}d))$ approach the point $-1$ from different half planes, we get different branches of the function $t \mapsto t^{-(\nu+1)}$ evaluated at $-1$. Hence, the above limit as $d\to0^+$ equals
\[  \int_\R \dd t\, \exp\big(-\tfrac 1 2 t^2\big)\,\big(\exp(\mathrm{i} \pi(\nu+1))-\exp(-\mathrm{i} \pi(\nu+1)\big) = \sqrt{2\pi} \,2\mathrm{i} \sin\big(\pi(\nu+1)\big)\,.
\]
Using \eqref{product of Gamma} with $z=-\nu$ we arrive to leading order as $x\to-\infty$ at 
\begin{align*} D_\nu(x)&\sim -\frac{\pi}{\Gamma(-\nu) \sin(\pi\nu)} \,\frac1{2\pi\mathrm{i}}\, \sqrt{2\pi} \,2\mathrm{i} \sin(\pi(\nu+1)) \, \exp(x^2/4) |x|^{-\nu-1}=  \frac{\sqrt{2\pi}}{\Gamma(-\nu)}\, \exp(x^2/4) |x|^{-\nu-1}\,,
\end{align*}
which is \eqref{D large negative x}.

To study the differential with respect to $\nu$, we use \eqref{D nu modified integral} with $d=1$. Then, for $j\in \{1,2\}$, we observe
\begin{align}
\frac{\dd^j}{\dd\nu^j}\tilde D_\nu(x)= & \frac{1}{2\pi\mathrm{i}} \exp(x^{2}/4)\,\,\left(\int_{\R-\mathrm{i}}-\int_{\R+\mathrm{i}}\right)\mathrm{d}t\, \exp\big(-\tfrac{1}{2}t^{2}\big) \left(t+x\right)^{-(\nu+1)} (-\ln(t+x))^j\,.
\end{align}
Along this contour, $\left(t+x\right)^{-(\nu+1)} (-\ln(t+x))^j$ is bounded by some constant depending on $\nu, j$, provided that $\nu>-1$. Thus, for $\nu>-1$, as $\exp(-\tfrac 1 2 t^2)$ is clearly absolutely integrable along the contour, for $j=1,2$ we get the estimate
\begin{align*}
\left \lvert \frac{\dd^j}{\dd\nu^j}\tilde D_\nu(x) \right \rvert \le C \exp(x^2/4) \, ,
\end{align*}
which is \eqref{d nu D bound}. 

To get the asymptotic description for $\nu=0$ as $x \to - \infty$, we also go to the contour integral $\int_{\R-\mathrm{i}}-\int_{\R+\mathrm{i}}$. We assume $\lvert x \rvert>2$ and define
\begin{align}
f_x(t) \coloneqq \exp\big(-\tfrac{1}{2}t^{2}\big) \frac {-\ln(t+x) \lvert x \rvert}  {t+x}\quad\mbox{ for } t \in \C \setminus  (-\infty, -x]\,.
\end{align}
Thus, we see that 
\begin{align*}
\frac{\dd}{\dd\nu}\tilde D_\nu(x)\Big|_{\nu=0}= & \frac{1}{2\pi\mathrm{i}} \exp(x^{2}/4) \lvert x \rvert^{-1} \,\,\int_\R \mathrm{d}t \,\big(f_x(t-\mathrm{i})-f_x(t+\mathrm{i})\big) \, .
\end{align*}
We need to be careful around the singularity at $t=-x$. Thus, we split the contour into parts. On the domain $\lvert t \rvert> \sqrt{ \lvert x \rvert}$, the exponential decay suffices and we can just use that $\ln(t+x)/(t+x)$ remains bounded to get
\begin{align*}
\left \lvert \int_{\lvert t \rvert> \sqrt{ \lvert x \rvert}} \mathrm d t \, \big(f_x(t-\mathrm{i})-f_x(t+\mathrm{i})\big) \right \rvert \le C \lvert x \rvert \int_{\sqrt{\lvert x \rvert}}^\infty \mathrm d t \exp(-\tfrac 1 2 t^2) < C \exp(-\tfrac 1 4 \lvert x \rvert) \, .
\end{align*}
For the remaining integral with $|t|\le \sqrt{|x|}$, we switch to a different contour. For every $h \in (0,1)$, we can replace the straight path from $-\sqrt{\lvert x \rvert} -\mathrm{i}, +\sqrt{\lvert x \rvert}-\mathrm{i}$ by the sequence of straight paths via $-\sqrt{\lvert x \rvert}-\mathrm{i}h, \sqrt{\lvert x \rvert}-\mathrm{i}h$ and do the same for the path in the upper half-plane. For the short paths, we clearly get for any $e_1,e_1\in \{\pm 1\}$ 
\begin{align*}
\left \lvert \int_{e_1 \sqrt{\lvert x \rvert} + e_2 \mathrm{i}  }^{e_1 \sqrt{\lvert x \rvert} + e_2 \mathrm{i} h } \dd t\,f_x(t)\right \rvert \le C \exp(-\tfrac 1 4 \lvert x \rvert) \, .
\end{align*}
For the remaining leading term, we can take the limit $h \to 0^+$ and, as $\frac{t}{|x|} \le \frac{1}{\sqrt{|x|}}\le \frac1{\sqrt{2}}$ and $\lim_{h \to 0^+}\big(\ln(s+\mathrm{i}h)-\ln(s-\mathrm{i}h)\big)=2\pi \mathrm{i}$  for any $s<0$, observe
\begin{align*}
\lim_{h \to 0^+} \int_{-\sqrt {\lvert x \rvert} }^{\sqrt { \lvert x \rvert} } \mathrm d t \,\big(f_x(t-\mathrm{i}h)-f_x(t+\mathrm{i}h)\big) &=\int_{-\sqrt {\lvert x \rvert} }^{\sqrt { \lvert x \rvert} } \mathrm d t  \,\exp\big(-\tfrac{1}{2}t^{2}\big) \frac {2 \pi \mathrm{i}  }  {-1+\tfrac t {\lvert x \rvert}} \\
&= 2 \pi \mathrm{i}  \int_{-\sqrt {\lvert x \rvert} }^{\sqrt { \lvert x \rvert} } \mathrm d t  \exp\big(-\tfrac{1}{2}t^{2}\big) \left( -1- \frac t {\lvert x \rvert} + \mathcal O \big( \tfrac {t^2} {\lvert x \rvert^2} \big) \right) \\
&= -2\pi \mathrm{i} \int_\R \exp\big(-\tfrac{1}{2}t^{2}\big) \mathrm d t + \mathcal O \big( \tfrac 1 {\lvert x \rvert^2}\big) \\
&=-2\pi \mathrm{i}  \sqrt{2\pi} + \mathcal O \big( \tfrac 1 {\lvert x \rvert^2}\big)\,.
\end{align*}
Thus, we can conclude
\begin{align*}
\frac{\dd}{\dd\nu}\tilde D_\nu(x)\Big|_{\nu=0}= & - \frac { \sqrt{2 \pi}} {\lvert x \rvert}  \exp(x^{2}/4) \left( 1+ \mathcal O \big( \tfrac 1 {\lvert x \rvert^2}\big) \right) \, ,
\end{align*}
which is \eqref{d nu D at nu = 0}.

Let $-n\in\Z^-$. Assuming $\frac{\dd}{\dd\nu}\tilde D_{\nu}(\mathrm{i}x)\Big|_{\nu=-n}$ exists, we see that
\begin{align}\label{D_{-n}(ix)}
D_{-n}(\mathrm{i}x) &= \lim_{\nu \to -n}D_{\nu}(\mathrm{i}x) =\left(  \lim_{\nu \to -n} \Gamma(1+\nu) (\nu+n) \right) \left(  \lim_{\nu \to -n} \frac { \tilde D_{\nu}(\mathrm{i}x)}{\nu+n} \right)\nonumber \\
&= \operatorname{res}_{-n+1}(\Gamma) \frac{\dd}{\dd\nu}\tilde D_{\nu}(\mathrm{i}x)\Big|_{\nu=-n}\,.
\end{align}
As the following integral is absolutely integrable, we may exchange integration with the differential and get
\begin{align*}
\frac{\dd}{\dd\nu}\tilde D_{\nu}(\mathrm{i}x)\Big|_{\nu=-n} &= \frac 1 {2\pi \mathrm{i}}  \exp(x^2/4) \int_{-\infty}^{(0+)} \mathrm d t \exp\big( \mathrm{i} xt - \tfrac 1 2 t^2\big) t^{n-1} (-\ln(t)) \,.
\end{align*}
Similarly to the calculation for $\nu=0$, we observe that the integrand is  product of a holomorphic function on $\C$ and the logarithm. The difference is that the singularity at $0$ is now only of the type $\ln(t)$ and thus, we can contract the integration path to $\R^-$ and the logarithm cancels out the factor $2 \pi \mathrm{i}$. Hence, we get 
\begin{align*}
\frac{\dd}{\dd\nu}\tilde D_{\nu}(\mathrm{i}x)\Big|_{\nu=-n}  &= \exp(x^2/4)  \int_{\R^-} \mathrm d t \,\exp( \mathrm{i} xt) \exp \big( - \tfrac 1 2 t^2\big) t^{n-1}  \,.
\end{align*}
We note that for any $0 \le j \le n-1$, we have 
\begin{align*}
\frac{\dd^j}{\dd t^j} \exp \big( - \tfrac 1 2 t^2\big) t^{n-1}   \Big|_{t=0}&= 
	\begin{cases}
		0 & \text{ if } j<n-1 \\
		(n-1)! & \text{ if } j=n-1
	\end{cases} \,,\\
\lim_{t' \to -\infty} \frac{\dd^j}{\dd t^j} \exp \big( - \tfrac 1 2 t^2\big) t^{n-1}   \Big|_{t=t'} \exp(\mathrm{i}xt') &=0\,.
\end{align*}
Thus, an $n$-fold integration by parts tells us
\begin{align*}
\frac{\dd}{\dd\nu}\tilde D_{\nu}(\mathrm{i}x)\Big|_{\nu=-n}  &= \exp(x^2/4)  \int_{\R^-} \mathrm d t \,\exp( \mathrm{i} xt) \exp \big( - \tfrac 1 2 t^2\big)  t^{n-1}  \\
&=  \frac{\exp(x^2/4)}{(\mathrm{i}x)^n} (-1)^{n-1}  \left( \left[  \frac{\dd^{n-1}}{\dd t^{n-1}} \exp \big( - \tfrac 1 2 t^2\big) t^{n-1} \Big |_{t=s}  \exp(\mathrm{i}xs) \right]_{-\infty}^0 \right.
\\
&\phantom{\frac{\exp(x^2/4)}{(\mathrm{i}x)^n} (-1)^{n-1}\left( \left[\right.\right.}\left.- \int_{\R^-} \mathrm d s \,\frac{\dd^n}{\dd t^n} \exp \big( - \tfrac 1 2 t^2\big) t^{n-1}   \Big|_{t=s} \exp(\mathrm{i}xs) \right) \\
&= \frac{\exp(x^2/4)}{(\mathrm{i}x)^n} (-1)^{n-1} \left( (n-1)! - \int_{\R^-} \mathrm d s \,\frac{\dd^n}{\dd t^n} \exp \big( - \tfrac 1 2 t^2\big) t^{n-1}   \Big|_{t=s} \exp(\mathrm{i}xs) \right)  \,.
\end{align*}
Due to the Riemann--Lebesgue lemma, the final integral converges to zero as $x \to \pm \infty$. Thus, for $x \to \pm \infty$, we get the asymptotic expansion\footnote{We could get an $1/x^2$ error term by two more integrations by parts.}
\begin{align*}
\frac{\dd}{\dd\nu}\tilde D_{\nu}(\mathrm{i}x)\Big|_{\nu=-n} =  \frac{\exp(x^2/4)}{(\mathrm{i}x)^n} (-1)^{n-1}  (n-1)! \left( 1+ o (1) \right)\,.
\end{align*}
We recall that a meromorphic function $f$ has a first order pole at $z \in \C$, if and only if $\lim_{h \to 0} f(z+h)h$ exists and does not vanish. Furthermore, this limit is the residue of $f$ at $z$. Using the recurrence relation $z \Gamma(z)=\Gamma(z+1)$, we can calculate the residue of $\Gamma$ at $-n+1$ with $n\in\Z^+$ as 
\begin{align} \label{residue of Gamma}
\operatorname{res}_{-n+1}(\Gamma)= \lim_{h \to 0} \Gamma(-n+1+h)h= \lim_{h\to 0} \frac { \Gamma(1+h) h} { \prod_{j=0}^{n-1} (h-j)} = \frac{(-1)^{n-1}}{(n-1)!} \,.
\end{align}
Thus, we finally conclude for $x\to\pm\infty$ using \eqref{D_{-n}(ix)}
\begin{align*}
D_{-n}(\mathrm{i}x)= \operatorname{res}_{-n+1}(\Gamma)  \frac{\dd}{\dd\nu}\tilde D_{\nu}(\mathrm{i}x)\Big|_{\nu=-n}  =   \frac{\exp(x^2/4)}{(\mathrm{i}x)^n} \big(1+o(1)\big)\, ,
\end{align*}
which is \eqref{D_-n(ix)}.

For the next step, consider
\begin{align}
\left(\exp(x^2/4)\tilde D_\nu(x) \right)' &= \frac{1}{2\pi \mathrm{i}} \int_{-\infty}^{(0+)} \mathrm d t \exp(xt-\tfrac  1 2 t^2) t^{-\nu} 
=  \exp(x^2/4) \tilde D_{\nu-1}(x) \, ,
\end{align}
which is \eqref{exp D diff 1}. Similarly, starting with \eqref{D nu modified integral} with $d=1$ (which is valid for all $x\in \R$), we get
\begin{align*}
\left(\exp(-x^2/4)\tilde D_\nu(x) \right)' &= -(\nu+1) \frac{1}{2\pi \mathrm{i}} \left(\int_{\R-\mathrm{i}}-\int_{\R+\mathrm{i}}\right) \mathrm d t \exp(-\tfrac  1 2 t^2) (t+x)^{-\nu-2} \\
&= -(\nu+1)   \exp(-x^2/4) \tilde D_{\nu+1}(x) \, ,
\end{align*}
which is \eqref{exp D diff 2}.

If we apply the last two equations at $x=0$, we get
\begin{align}  \label{D nu at 0 nu-2}
\tilde D_{\nu-1}(0)= \tilde D_\nu'(0) = -(\nu+1) \tilde D_{\nu+1}(0) \, .
\end{align}
To continue our study of $D_\nu(0)$, we consider the case $\nu<0$ and get 
\begin{align*} &D_\nu(x)
\\
&=\frac{\Gamma(\nu +1)}{2\pi\mathrm{i}} \exp(-x^2/4) \int_0^\infty \dd s\, \exp\big(-xs-\tfrac{1}{2}s^2\big) \Big[(s\exp(-\mathrm{i}\pi))^{-(\nu+1)} - (s\exp(+\mathrm{i}\pi))^{-(\nu+1)}\Big]
\\
&=\frac{\Gamma(\nu+1)}{2\pi\mathrm{i}}\exp(-x^2/4)\int_0^\infty \dd s\, \exp\big(-xs-\tfrac{1}{2}s^2\big) s^{-(\nu+1)} \Big[\exp\big(\mathrm{i}\pi(\nu+1)\big) - \exp\big(-\mathrm{i}\pi(\nu +1)\big)\Big]
\\
&=-\frac{\Gamma(\nu+1)}{\pi}\,\sin(\pi\nu) \exp(-x^2/4)\int_0^\infty \dd s\, \exp\big(-xs-\tfrac{1}{2}s^2\big) s^{-(\nu+1)}\,.
\end{align*}
This yields for $\nu<0$,
\begin{align*} 
D_\nu(0) &= -\frac{\Gamma(\nu+1)}{\pi}\,\sin(\pi\nu) \int_0^\infty \dd s\, \exp\big(-\tfrac{1}{2}s^2\big) s^{-(\nu+1)} =-\frac{\Gamma(\nu+1)}{\pi}\,\sin(\pi\nu) \, 2^{-\tfrac{\nu}{2}-1} \Gamma(-\tfrac{\nu}2)
\\
&=2 (-\tfrac{\nu}{2}) \Gamma(\nu) \Gamma(-\tfrac{\nu}{2}) \,\frac{\sin(\pi\nu)}{\pi} \,2^{-\tfrac{\nu}{2}-1}=\Gamma(-\tfrac{\nu}{2}+1) \Gamma(\nu)\,\frac{\sin(\pi\nu)}{\pi} \,2^{-\tfrac{\nu}{2}}
\\
&=\frac{\pi}{\sin(\pi\tfrac{\nu}{2})} \,\frac{\Gamma(\nu)}{\Gamma(\tfrac{\nu}{2})} \,\frac{\sin(\pi\nu)}{\pi} \,2^{-\tfrac{\nu}{2}}=\frac{2\cos(\pi\nu/2)}{\sqrt{\pi}} \, 2^{\nu-1} \,\Gamma(\tfrac{\nu}{2}+\tfrac12) \,2^{-\tfrac{\nu}{2}}
\\
&= \frac{2^{\tfrac{\nu}{2}}}{\sqrt{\pi}} \,\Gamma\big(\tfrac{\nu+1}{2}\big) \sin \big( \pi \big(\tfrac{\nu+1}{2}\big) \big) =2^{\tfrac{\nu}{2}} \sqrt{\pi} \frac{1}{\Gamma(\tfrac{1}{2}-\tfrac{\nu}{2})}\,.
\end{align*}
Here, we have used the substitution $\tfrac12 s^2 = t$ in the integral and the very definition of the $\Gamma$-function, the recurrence relation $\Gamma(z+1) = z \Gamma(z+1)$ for $z=\nu$ and for $z=-\nu/2$, the reflection property for $z=\nu/2$ and the duplication formula $\Gamma(2z) = \pi^{-1/2} 2^{2z-1} \Gamma(z) \Gamma(z+\tfrac{1}{2})$ for $z=\nu/2$, and finally the reflection property for $z=(1+\nu)/2$. 

To conclude the same claim for $\nu>0$ with $\nu \not \in \Z$, we use \eqref{D nu at 0 nu-2} to show by induction over $n$, that it holds for all $\nu \in \R \setminus \Z$ with $\nu<2n$. The induction step is
\begin{align*}
D_{\nu+2}(0)&= \frac{-\Gamma(\nu+3)}{(\nu+2)\Gamma(\nu+1)}  D_\nu(0) =  (-2)\tfrac{\nu+1}2   \frac{2^{\tfrac{\nu}{2}}}{\sqrt{\pi}} \,\Gamma\big(\tfrac{\nu+1}{2}\big) \sin \big( \pi \big(\tfrac{\nu+1}{2}\big) \big) 
\\
&= \frac{2^{\tfrac{\nu+2}{2}}}{\sqrt{\pi}} \,\Gamma\big(\tfrac{\nu+3}{2}\big) \sin \big( \pi \big(\tfrac{\nu+3}{2}\big) \big)  \,.
\end{align*}
Finally, \eqref{d nu D bound} tells us that $D_\nu$ is a twice differentiable in $\nu$, at least for $\nu>-1$. Thus, as the right-hand side of \eqref{D nu at 0 final} is continuous on $\nu>-1$, we can conclude that equality holds for $\nu \in \N$, too, which finally shows \eqref{D nu at 0 final} for all $\nu \in \R\setminus \Z^-$. 

\end{proof}

\begin{cl} \label{D0 D1 cl}
For any $x \in \R$, we have $D_0(x)=\exp(-x^2/4)$ and $D_1(x)=x \exp(-x^2/4)$. 
\end{cl}
\begin{proof}
Due to \eqref{exp D diff 1} and $\tilde D_{-1}=0$ (see below \eqref{loop rep of tilde D}), we can conclude that for every $n \in \N$, there is a polynomial $p_n(x)$ of degree at most $n$, such that $D_n(x)=p_n(x)\exp(-x^2/4)$. For $n=0,1$, the asymptotic expansion \eqref{D large x} shows $p_n(x)=x^n$. We remark that with \eqref{exp D diff 1} and \eqref{exp D diff 2}, one can actually show $p_n(x)=2^{-n/2}H_n(x/\sqrt 2)$, where $H_n$ are the Hermite polynomials, see \cite[19.13.1]{AS}.
\end{proof}

\begin{lemma} \label{Sturm comparison principle lem}
Let $\nu ,\nu' \in \R \setminus \Z^-$ with $\nu <\nu'$ and $x_1 \in \R, x_2 \in (-\infty,\infty]$, such that $x_1<x_2, D_\nu(x_1)=0$ and $\lim_{x \to x_2}D_\nu(x)=0$ and $D_\nu(x) \neq 0$ for any $x\in (x_1,x_2)$. Then, there is an $x_0 \in (x_1,x_2)$ with $D_{\nu'}(x_0)=0$.
\end{lemma}
\begin{remark}
This lemma and its proof are based on the Sturm comparison principle. The case $x_2=\infty$ uses the same arguments as the standard proof, but is not generally included. Thus, we prove both cases.
\end{remark}
\begin{proof}
Assume by contradiction that no such $x_0$ exists. Thus, both $D_\nu$ and $D_{\nu'}$ have a constant sign on $(x_1,x_2)$. Define $f \coloneqq \pm D_\nu, g \coloneqq \pm D_{\nu'}$, such that $f$ and $g$ are positive on $(x_1,x_2)$. Now consider
\begin{align*}
\left[ f'(x)g(x)-f(x)g'(x) \right]_{x_1}^{x_2}  &=  \int_{x_1}^{x_2} \mathrm d x\,\big(f'(x)g(x)-f(x)g'(x)\big)'  \\
&=  \int_{x_1}^{x_2} \mathrm d x\,\big(f''(x) g(x) - f(x) g''(x)\big) \\
&= \int_{x_1}^{x_2} \mathrm d x\,\left( \big( \tfrac 1 4 x^2 - \nu - \tfrac 12 \big) -   \big( \tfrac 1 4 x^2 - \nu' - \tfrac 12 \big)   \right) f(x) g(x)  \\
&= \int_{x_1}^{x_2} \mathrm d x\,(\nu'-\nu) f(x) g(x) >0\,.
\end{align*}
We used the differential equation \eqref{tilde D nu diff eq}. Since $f$ and $g$ are strictly positive and $\nu' > \nu$, the right-hand side is strictly positive. For the left-hand side, we observe that $f'(x_1) \ge 0$, since $f(x_1)=0$ and $f$ is positive on $(x_1,x_2)$. If $x_2<\infty$, we similarly see $f'(x_2)\le 0$, while for $x_2= \infty$, the asymptotic expansion \eqref{D large x} in combination with \eqref{exp D diff 1} tells us that $\lim_{x\to \infty} \max(\lvert f (x) \rvert, \lvert f'(x) \rvert, \lvert g(x) \rvert, \lvert g'(x) \rvert)= 0$ and thus the boundary term at $\infty$ vanishes. Thus, the possible boundary terms are $-f'(x_1)g(x_1) \le 0, f'(x_2)g(x_2) \le 0$ and $0$. In conclusion, the left-hand side of the equation is not positive, which yields the desired contradiction.
\end{proof}

\begin{cl} \label{negative nu no zeros lem}
For $\nu <0$ with $\nu \not \in \Z^-$, the function $D_\nu$ has no zeros.
\end{cl}

\begin{proof}
As $D_\nu$ is continuous, $D_\nu^{-1}(0)$ is closed. Due to \eqref{D large x} and \eqref{D large negative x}, we see that $D_\nu^{-1}(0)$ is bounded. Thus, assuming $D_\nu$ had a real zero, it would have a largest real zero $x_1$. Applying \autoref{Sturm comparison principle lem} with $\nu'=0$ and $x_2=\infty$ yields that $D_0(x)=\exp(-x^2/4)$ has a real zero, which is false. Thus, $D_\nu$ has no real zeros.
\end{proof}

\begin{cl} \label{D nu unique simple zero cl}
For $\nu \in (0,1]$, $D_\nu$ has exactly one zero and it is simple. Furthermore, the implicitly defined function  $\zeta \colon (0,1] \to (-\infty,0]$ with $D_\nu\big(\zeta(\nu)\big)=0$ is strictly increasing and differentiable.
\end{cl}

\begin{proof}
 If $D_\nu$ had a non-simple zero $x_0$, we would see that $D_\nu(x_0)=D_\nu'(x_0)=0$. However, due to the Picard--Lindelöf theorem, there is only one function $w$ solving \eqref{tilde D nu diff eq} and $w(x_0)=w'(x_0)=0$, which is obviously $w \equiv 0$. Thus, $D_\nu$ has only simple zeros. The function $D_1(x)=x \exp(-x^2/4)$ (see \autoref{D0 D1 cl}) clearly has exactly one zero at $0$. For $\nu\in(0,1)$, the asymptotics \eqref{D large x} and \eqref{D large negative x} show opposite signs for $D_\nu(x)$ as $x \to \pm \infty$, by the intermediate value theorem, $D_\nu$ has at least one zero. Furthermore, as the zeros are simple, the sign changes at each zero and thus, $D_\nu$ has an odd number of zeros.

As $\nu<1$, we can use $\nu'=1$ in \autoref{Sturm comparison principle lem}. If $D_\nu$ had more than one zero, it would have at least three zeros. Let $x_1,x_2,x_3$ be three such zeros with $x_1<x_2<x_3$ and no further zeros between $x_1$ and $x_3$. Now, if $x_2 \le 0$,  \autoref{Sturm comparison principle lem} yields that $D_1$ has a negative zero, which is false. On the other hand, if $x_2 >0$, \autoref{Sturm comparison principle lem} for $x_2$ and $x_3$ tells us that $D_1$ has a positive zero, which is also false. Thus, for any $\nu \in (0,1]$, $D_\nu$ has a unique zero, which we call $\zeta(\nu)$ and note  $\zeta(1)=0$.

As $\zeta(\nu)$ is a simple zero of $D_\nu$, we know $D_\nu'(\zeta(\nu))\neq 0$ and thus, the implicit function theorem tells as that $\zeta$ is a well-defined $\mathsf{C}^1$-function with
\begin{align*}
\zeta'(\nu) =- \frac { \big(\tfrac {\partial } {\partial \nu} D_\nu\big)\big(\zeta(\nu)\big)} { D_\nu'\big(\zeta(\nu)\big) } \, .
\end{align*} 

For any $\nu<1, \nu'\le 1$ with $\nu<\nu'$, \autoref{Sturm comparison principle lem} with $(x_1,x_2)=\big(\zeta(\nu), \infty\big)$ yields that $\zeta(\nu') \in \big(\zeta(\nu),\infty\big)$, which states that $\zeta$ is strictly increasing. 
\end{proof}

\begin{cl}\label{nu(k) asymptotics lem}
For $k\ge0$, we define
\begin{equation}\label{def of nu(k)}
	\nu(k)\coloneqq \inf\big\{\nu >0 : D_\nu(-{\sqrt 2} k)=0\big\}\,.
\end{equation}
Then, the function $[0,\infty)\ni k\mapsto \nu(k)\in (0,1]$ is strictly decreasing and as $k \to \infty$, we have
\begin{equation} \label{nu(k) asymptotic}
\nu(k) = \frac{k}{\sqrt{\pi}} \exp(-k^2)\Big(1+\mathcal O(k^{-2})\Big) \,.
\end{equation}
\end{cl}

\begin{proof}
Due to \autoref{D nu unique simple zero cl}, we see that $\nu(k) = \zeta^{-1}(-\sqrt2  k)$. Thus, as $\zeta$ is strictly increasing, the function $k \mapsto \nu(k)$ is strictly decreasing. In particular, for any $k \ge 0$, we have $\nu(k) \le \nu(0)=1$. Thus, it suffices to consider sufficiently large $k$.

Consider the function $f_k \colon [0,1] \to \R$ given by $f_k(\nu)\coloneqq D_{\nu}(- \sqrt 2 k)$. We have shown (see \autoref{D0 D1 cl}, \eqref{d nu D at nu = 0} and \eqref{d nu D bound}) 
\begin{alignat}{2}
a_k&\coloneqq f_k(0)&&=\exp(-k^2/2) \,,\\
-b_k&\coloneqq f_k'(0)&&= -\frac{\sqrt{\pi}}{k}\exp(k^2/2) \Big(1+\mathcal O(k^{-2})\Big) \,, \\
\lVert f_k'' \rVert_{\Lp^\infty((0,1))}& \le C \exp(k^2/2) &&\eqqcolon c_k\,.
\end{alignat}
For the last inequality we used in addition that $\sup\big\{\big|\frac{\dd^j}{\dd \nu^j}\Gamma(\nu+1)\big| : j\in\{0,1,2\},\nu\in[0,1]\big\}$ is finite. Thus, $f_k$ is strictly decreasing on $[0,\max(1, b_k/c_k)]$. If $b_k<c_k$, which holds for sufficiently large $k$, we can conclude that $f_k$ has at most one zero in $[0, b_k/c_k]$. The estimates
\begin{align}
a_k-b_kx -\frac{c_k}2 x^2 \le f_k(x) \le a_k-b_kx +\frac{c_k}2 x^2 
\end{align}
imply that   it has exactly one zero on that interval, provided that $2a_k<b_k^2/c_k$, which holds for sufficiently large $k$. 
 From the zeros of the two quadratic polynomials bounding $f_k$, we obtain
\[ -\frac{b_k}{c_k}+\sqrt{\frac{b_k^2}{c_k^2} + 2\frac{a_k}{c_k}} \le \nu(k) \le \frac{b_k}{c_k}-\sqrt{\frac{b_k^2}{c_k^2} - 2\frac{a_k}{c_k}}\,.
\]
Using $\sqrt{1-x} = 1-\tfrac{1}{2}x+\mathcal O(x^2)$ we have the expansion
\begin{align*} 
\frac{b_k}{c_k}-\sqrt{\frac{b_k^2}{c_k^2} - 2\frac{a_k}{c_k}}&=\frac{b_k}{c_k}\left(1-\sqrt{1-2 \frac{a_k}{b_k}\frac{c_k}{b_k}}\right)
\\
&=\frac{b_k}{c_k}\left(1-\left(1-\frac{a_k}{b_k}\frac{c_k}{b_k} + \mathcal O\left(\tfrac{a_k}{b_k}\tfrac{c_k}{b_k}\right)^2\right)\right)
\\
&=\frac{a_k}{b_k}\left(1+ \mathcal O \left( k^2\exp(-k^2)\right)\right)\\
&= \frac k {\sqrt \pi} \exp(-k^2) \Big(1+\mathcal O(k^{-2})\Big)\,,
\end{align*}
and similarly for the other root.
\end{proof}

\begin{lemma} \label{eigenfunctions are D nus lem}
Let $f$ be an eigenfunction of $H^+(k)$ with eigenvalue $\lambda$. Then, there is a $c \in \R$ such that
\begin{align}
f(x)=c D_{\tfrac {\lambda-1} 2 } \left( \sqrt 2 (x-k) \right) \,, \quad x\in[0,\infty)\,.
\end{align}
Furthermore, $D_{\tfrac {\lambda-1} 2 } (- \sqrt 2  k)=0$. 
\end{lemma}

\begin{proof}
As $H^+(k)$ is a positive operator, we have $\lambda \ge 0$. Let us define $\nu \coloneqq \tfrac {\lambda-1}2 \ge- \tfrac 1 2 $ and 
\begin{align*}
g(x) \coloneqq f\big( \tfrac x {\sqrt 2} +k \big) \,,\quad x\in[-\sqrt{2}k,\infty)\, .
\end{align*}
The claim boils down to showing that $g(-\sqrt 2 k)=0$ and $g(x)=c D_\nu(x)$ for some $c \in \R$. 

We will now show the following:
\begin{align}
g \in \mathsf{C}^2\big((-\sqrt 2 k ,&\infty)\big) \, , 									  \label{eigenfunctions are Dnu g claim 1} \\
\lim_{x \to \infty} g(x) \exp(-x^2/4) \lvert x \rvert^{\nu +1 } &= 0 \, ,  \label{eigenfunctions are Dnu g claim 2} \\
g''(x)- \left( \tfrac 1 4 x^2 - \nu - \tfrac 1 2 \right) g(x) &=0 \, ,		 \label{eigenfunctions are Dnu g claim 3}  \\
g(-\sqrt 2 k) &=0 \, .											 \label{eigenfunctions are Dnu g claim 4} 
\end{align} 
As $f$ is an eigenfunction of $H^+(k)$, we get $f \in \mathcal Q(q^+_k) \subset \mathsf{H}^1(\R^+)$, see \eqref{H+k form domain}. Furthermore, as $(x-k)^2 f \in \Lp^2_{\text{loc}}(\R^+)$ and $H^+(k)f= \big(-\Delta + (x-k)^2\big) f \in \Lp^2(\R^+)$, we can conclude $f \in \mathsf{H}^2_{\text{loc}}(\R^+)$. By the Sobolev imbedding theorem~\cite[Theorem 5.4, Case C]{Adams}, we can conclude $\mathsf{H}^2_{\text{loc}}(\R^+) \subset  \mathsf{C}^1(\R^+)$ and $\mathsf{H}^1(\R^+) \subset \Lp^\infty(\R^+)$. Thus, we get $f \in \mathsf{C}^1(\R^+) \cap \Lp^\infty(\R^+)$. As $f \in \Lp^\infty(\R^+)$, we can conclude \eqref{eigenfunctions are Dnu g claim 2}. Furthermore, as $H^+(k)f= \lambda f$ and $f \in \mathsf{C}^1(\R^+)$, we obtain $f \in \mathsf{C}^2(\R^+)$ and thus \eqref{eigenfunctions are Dnu g claim 1}. Finally $H^+(k)f= \lambda f$  translates to \eqref{eigenfunctions are Dnu g claim 3} and the boundary condition $f(0)=0$ from $f \in \mathcal Q (q^+_k) $ becomes \eqref{eigenfunctions are Dnu g claim 4}. 

Due to the Picard--Lindelöf theorem, all $\mathsf{C}^2\big((- \sqrt 2 k , \infty)\big)$-functions, which solve \eqref{eigenfunctions are Dnu g claim 3}  form a two-dimensional space. Thus, whenever we know two linearly independent solutions of \eqref{eigenfunctions are Dnu g claim 3}, all other solutions are linear combinations of these two. 

For $\nu \in \Z$, as $\nu \ge-\tfrac 1 2$, we get $\nu \in \N$. Here, we choose the solutions $D_\nu$ and $h(x) \coloneqq D_{-\nu-1}(\mathrm{i}x)$. Due to \eqref{tilde D nu diff eq}, evaluated at $\mathrm{i}x$, we see
\begin{align*}
-h''(x) =  D_{\nu-1}''(\mathrm{i}x) &= \big(\tfrac{1}{4}(\mathrm{i}x)^2 - (-\nu-1)-\tfrac12\big)D_{\nu-1}(\mathrm{i}x)=\big(-\tfrac{1}{4}x^2 + \nu+\tfrac12\big) h(x)\,,
\end{align*}
which means that $h$ does indeed solve \eqref{eigenfunctions are Dnu g claim 3}. The asymptotic expansions for $x \to \infty$, \eqref{D large x} and \eqref{D_-n(ix)}, tell us that $h$ and $D_\nu$ are linearly independent. As $h$ is complex-valued, we get $a,c \in \C$ with $g=ah + c D_\nu$. Thus, \eqref{eigenfunctions are Dnu g claim 2} yields $a=0$. Now, as $g$ and $D_\nu$ are real-valued, as discussed below \eqref{loop rep of tilde D}, it follows that $c \in \R$ and thus $g(x)=cD_\nu(x)$.

For $\nu \not \in \Z$, let us consider the two functions $D_\nu$ and $x \mapsto D_{\nu}(-x)$. Due to \eqref{tilde D nu diff eq}, they both solve \eqref{eigenfunctions are Dnu g claim 3}.  Due to \eqref{D large x} and \eqref{D large  negative x}, they have wildly different asymptotic behaviour as $x \to \infty$. Thus, there are $a,c \in \R$, such that 
\begin{align*}
g(x)= a D_{\nu}(-x) + c D_{\nu}(x)\,.
\end{align*}
Now \eqref{eigenfunctions are Dnu g claim 2} tells us that $a=0$. Thus, we have shown $g(x)=c D_{\nu}(x)$, which completes the proof.
\end{proof}

\begin{cl} \label{psik is a Dnu cl}
For $k < 0$, the operator $H^+(k)$ has no eigenfunction with eigenvalue $\lambda \le 3$, while for $k>0$, $H^+(k)$ has exactly one eigenfunction with eigenvalue $\lambda \le 3$ satisfying $\lambda=2\nu(k)+1$. See \eqref{def of nu(k)} for the definition of the function $\nu$. That is, the dimension of the kernel,
\begin{align}
\operatorname{dim ker}\big(H^+(k)-\l\big) = \begin{cases}1&\mbox{ if } k \ge 0,\l=2\nu(k)+1\\0&\mbox{ if } k \ge 0, \l\not=2\nu(k)+1\\0&\mbox{ if } k < 0,\l \le 3\end{cases}\,.
\end{align}
The (normalized) ground-state eigenfunction $\psi_k$ is given by
\begin{align}
\psi_k(x) = \frac { D_{\nu(k)} \big( \sqrt 2 (x-k) \big) } { 2^{-\tfrac 1 4} \lVert D_{\nu(k)} \rVert_{\Lp^2((-\sqrt 2 k, \infty))}}\,, \quad  x \in [0,\infty)\,.
\end{align}
That is, 
\begin{align} \label{eigenvalue equ for H(k) plus}
H^+(k)\psi_k = \big(2\nu(k)+1\big)\psi_k\,,\quad \psi_k(0)=0\,,\quad\|\psi_k\|_{\mathsf{L}^2(\R^+)} =1\,.
\end{align}
\end{cl}

\begin{proof} By \autoref{compact resolvent}, the spectrum of $H^+(k)$ is discrete.
Let $\lambda$ be an eigenvalue of $H^+(k)$ and define $\nu \coloneqq \tfrac {\lambda-1} 2$. As $\lambda \ge 0$ we get $\nu \ge- \tfrac 12 $.  The final condition $D_{\tfrac {\lambda -1} 2 } (-\sqrt 2 k) =0$ of \autoref{eigenfunctions are D nus lem} shows that we need to understand the zeros of $D_\nu$ for $\nu \in [-\tfrac 12 , 1]$. Thus, we take a look at \autoref{negative nu no zeros lem}, \autoref{D nu unique simple zero cl} and \autoref{nu(k) asymptotics lem}. In combination, they tell us that in the given range $\nu \in [-\tfrac 12 , 1]$, the function $D_\nu$ has a zero at $-\sqrt 2 k$ only if $k \ge 0$ and $\nu = \nu(k)$. We are left to show that in this case, there actually is an eigenfunction. As $D_\nu \in \mathsf{C}^\infty(\R)$ and due to the asymptotic expansion \eqref{D large x}, we see that $D_\nu \in \Lp^2\big((-\sqrt 2 k, \infty)\big)$. Thus, the candidate function
\begin{align*}
w(x) \coloneqq  \frac { D_{\nu(k)} \big( \sqrt 2 (x-k) \big) } { 2^{-\tfrac 1 4} \lVert D_{\nu(k)} \rVert_{\Lp^2((-\sqrt 2 k, \infty))}}\,, \quad  x \in [0,\infty)
\end{align*}
is well-defined and $\lVert w \rVert_{\Lp^2(\R^+)}=1$. We also note $w(0)= D_{\nu(k)} ( -\sqrt 2 k ) =0$. As $w \in \Lp^2(\R^+) \cap \mathsf{C}^\infty ([0,\infty))$ and $w$ satisfies the differential equation $-w'' (x) + x^2 w(x)= \big(1+ 2 \nu(k)\big) w(x)$, we can conclude that $w$ is in the operator domain of $H^+(k)$ and it is a normalized eigenfunction with eigenvalue $\lambda=\big(1+2\nu(k)\big)<3$ of $H^+(k)$.
\end{proof}

\end{appendix}

\section*{Acknowledgement}
The work of P.P. was partially financed by the German Research Foundation's TRR 352-Project-ID 470903074.

\bibliography{HalfplaneLandauBib}{}
\bibliographystyle{abbrvurl}

\end{document}